\documentclass[final,5p,times,twocolumn]{elsarticle}

\usepackage[utf8]{inputenc}
\usepackage[T1]{fontenc}
\usepackage[english]{babel}
\usepackage{lmodern}
\usepackage{microtype}

\usepackage{amsmath,amssymb,amsfonts,amsthm,mathtools,bm}
\usepackage{multirow}
\usepackage{booktabs}
\usepackage{graphicx}
\usepackage{xcolor}
\usepackage{float}
\usepackage{enumitem}
\usepackage[ruled,vlined]{algorithm2e}
\usepackage{geometry}              
\usepackage{url}                   

\graphicspath{{./}}

\AtBeginDocument{\setlength{\mathindent}{0pt}}

\newtheorem{lemma}{Lemma}
\theoremstyle{definition}

\theoremstyle{remark}
\newtheorem{remark}{Remark}

\def\E{\mathbb{E}}
\def\P{\mathbb{P}}
\def\l{\langle}

\newcommand{\erf}{\mathrm{erf}}
\newcommand{\erfc}{\mathrm{erfc}}
\newcommand{\sign}{\mathrm{sign}}

\allowdisplaybreaks
\journal{Neural Networks}

\begin{document}

\begin{frontmatter}

\title{Exponential Capacity in Multilayer Hetero-Associative Neural Networks}

\author[mat,gnfm,nanotec]{Elena Agliari}
\author[gnfm,nanotec,sbai,infn]{Adriano Barra}
\author[gnfm,sbai]{Andrea Ladiana\corref{cor1}}
\ead{andrea.ladiana@uniroma1.it}\author[mat,gnfm]{Andrea Lepre}

\cortext[cor1]{Corresponding author.}

\affiliation[mat]{organization={Dipartimento di Matematica, Sapienza Universit\`a di Roma}, city={Rome}, country={Italy}}
\affiliation[gnfm]{organization={Istituto Nazionale d’Alta Matematica, GNFM}, city={Rome}, country={Italy}}
\affiliation[nanotec]{organization={CNR-Nanotec, Unità di Lecce}, city={Lecce}, country={Italy}}
\affiliation[sbai]{organization={Dipartimento di Scienze di Base e Applicate all’Ingegneria, Sapienza Universit\`a di Roma}, city={Rome}, country={Italy}}
\affiliation[infn]{organization={Istituto Nazionale di Fisica Nucleare, Sezione di Roma 1}, city={Rome}, country={Italy}}

\begin{abstract}
Exponential Hopfield networks store a number of patterns that grows
exponentially with the number of neurons, and in their classical formulation they are auto-associative: they
complete a corrupted copy of a memory into the memory itself. Many of the tasks
one wants such a network to perform are instead hetero-associative, mapping a
cue to a different target. We introduce and analyse an exponential neural
network of $L$ layers of $N$ binary neurons, each layer carrying its own
dataset, whose energy is an exponential of the product of the per-layer Mattis
overlaps, so that it is minimised precisely when every layer retrieves the
pattern of the same index; the stored association must be a surjective function
of the cue, and we show why nothing else can be stored at all. A
cavity/signal-to-noise analysis, made exact at
leading order by a large-deviation evaluation of the noise, shows that the
aligned hetero-associative state is a fixed point of the zero-temperature
dynamics up to a number of stored patterns $P_c\sim e^{N\rho_L}$, exponential in
the layer size, with an explicit rate $\rho_L$ that grows like $L\log 2$;
enlarging the basins of attraction lowers the rate but never destroys its
exponential character. Comparing the theory with structured data
we
find that the exponential capacity and the predicted basins survive correlated,
many-to-one patterns: the network is a near-perfect content-addressable memory.
The same closed forms describe, without refitting, a synthetic manifold, real
T-cell-receptor/epitope triples and natural-language intent data, so the
mechanism is domain-universal. Generalisation to unseen cues, though
significantly above chance, stays below memorisation, and it is the geometry of
the encoding, rather than the data domain, that sets how far above chance it
reaches. In this family, exponential storage and strong generalisation are
distinct capabilities.
\end{abstract}


\end{frontmatter}


\section{Introduction}
\label{sec:intro}

An associative memory is judged by two numbers: how many patterns it can store,
and how much of a pattern it needs to see before it recalls the rest. For four
decades the first number set the pace. Hopfield's model~\cite{Hopfield1982}, the
harmonic oscillator of neural computation, was shown by Amit, Gutfreund and
Sompolinsky to hold a number of memories growing only linearly in the
number of neurons, $P\simeq 0.14\,N$~\cite{AGS1985a,AGS1985b}. Reading the model
as a pairwise spin glass made the ceiling look fundamental, and made the way
past it obvious: add interactions. Baldi--Venkatesh and
Gardner~\cite{Baldi1987,Gardner1987,Gardner1985} showed that $p$-body couplings
lift the capacity to $P\sim N^{\,p-1}$~\cite{Bao2022}, and the dense associative
memories of
Krotov and Hopfield~\cite{Krotov2016,Krotov2018,KrotovNature2023,KrotovHopfield2020} turned this
into a working principle for pattern recognition. Summing all the dense orders
at once is the natural endpoint of the program, and it delivers an
exponential capacity $P\sim e^{cN}$: for binary neurons in the model of
Demircigil et al.~\cite{Demircigil2017}, for continuous ones in the modern
Hopfield network of Ramsauer et al.~\cite{Ramsauer2021}, further developed into
energy-based transformer blocks~\cite{Hoover2024}, that underlies the
attention mechanism, and, analysed through the lens of glassy statistical
mechanics and Derrida's random-energy model, in the work of Lucibello and
M\'ezard~\cite{LucibelloMezard2024,Derrida1981,Dohmatob2023,Hillar2018}.

The binary case admits an especially transparent treatment. Albanese et al.~\cite{ALBANESE2026131223} write the energy as a
sum of exponentials of the loss, $\mathcal{H}=-N\sum_{\mu}e^{N(m_\mu-1)}$, and
show by a signal-to-noise argument, with no replicas and no mean-field
assumption, that each stored pattern is a fixed point of the zero-temperature
dynamics up to an exponentially large load, with basins of attraction that
shrink, but never close, as the load approaches capacity. That model is the
direct ancestor of the present one, and we lean on it throughout. It is worth
noting that this is not the only route to rigour available for this class of
models: the Guerra-interpolation programme has established, with comparable
rigour but different machinery (an interpolating free energy rather than a
cavity expansion), the free-energy and capacity of dense and Hebbian
associative networks~\cite{Agliari2020dense,Barra2018,Albanese2024hebbian}. We
use the cavity/large-deviation route throughout because it is the one that
extends most directly to the product-of-overlaps energy of
Section~\ref{sec:model}, without asserting it is the only one that would work.
The same multi-layer, hetero-associative construction principle --there called
\emph{multidirectional}-- has in fact been
analysed at (generalised) Hebbian order by replica methods across various architectures~\cite{centonze2024statistical, agliari2025generalized, Alessandrelli2025TAM}. 

\paragraph{The gap: association is not completion}
All of the above are auto-associative. The network is handed a corrupted copy of
a stored vector and returns that vector; input and output live in the same space,
and ``recall'' means ``clean-up''. Much of what one asks of an associative memory
is instead hetero-associative: cue and target are distinct objects, in general
in different spaces, and the task is to return the target given the cue. Pairing
the two chains of a receptor with the antigen they recognise, a sentence with its
intent, or one sensory modality with another are so many instances of this single
abstract map. An auto-associative network can only imitate it, by concatenating
cue and target into one long vector and hoping the dynamics does not tear the
halves apart. What is missing is a model in which the hetero-associative map is
the object the energy is built around, and in which the exponential capacity that
makes these networks attractive is provably retained.

Not every relation between cue and target can be stored, and it is worth saying
at once which can. The rule the network learns must be a \emph{function} of the
cue --a cue mapped to two targets makes the two votes cancel in the field, and
the target layer relaxes to a mixture of them-- and it must be
\emph{surjective} onto the target codebook, since a target no stored cue points
to is a memory with an empty basin. Injectivity, by contrast, is neither
required nor wanted: the many-to-one regime is the interesting one, and it is
the one real data supply. Remark~\ref{rem:surjective} makes this precise and
quantifies the failure mode; every dataset below is passed through a filter that
enforces it.

\paragraph{This paper}
We introduce an exponential neural network of $L$ layers of $N$ binary neurons.
Each layer stores its own dataset of $P$ patterns, and the energy is an
exponential of the \emph{product} of the per-layer Mattis overlaps, so that it is
minimised precisely when every layer simultaneously retrieves the pattern of the
same index: the hetero-associative ground state. The linear-exponent
auto-associative model~\cite{ALBANESE2026131223} is recovered at $L=1$; a
$\mathbb{Z}_2$-symmetric squared-overlap variant sits between them and is
analysed in an appendix.

Our contribution is twofold. On the theory side
(Sections~\ref{sec:model}--\ref{sec:basins}) we carry the cavity/signal-to-noise
program through for an arbitrary number of layers $L$. The novelty relative to the
auto-associative case is that the per-pattern noise no longer factorises over
sites: the product structure makes its second moment a genuine large-deviation
integral, which we evaluate exactly at leading order by a saddle point on the
symmetric ray of layer magnetisations. The output is a closed-form storage
capacity, exponential in the layer size $N$,
\begin{equation}
  P_c\;\sim\;e^{N\rho_L},
  \label{eq:intro-rate}
\end{equation}
whose rate $\rho_L$ is fixed, for every $L$, by a single scalar saddle-point
equation: the stationarity condition of a one-dimensional variational functional,
evaluated on the symmetric ray along which all layers align coherently with the
same pattern. The rate is monotonically increasing in $L$ and approaches
$L\log 2$, so binding one further modality multiplies the capacity by a factor
that is itself exponential in $N$. To this we add a matching analysis of the
basins under a corrupted cue, which exposes a
clean trade-off --wider networks store exponentially more but tolerate a
proportionally smaller corruption radius-- together with the gap between the
annealed (moment) estimate and the typical realisation that governs finite-size
dynamics.

On the empirical side (Sections~\ref{sec:hmm}--\ref{sec:clinc}) we ask the
question the theory cannot: what happens when the stored patterns are not
i.i.d.\ Rademacher vectors but real, structured, many-to-one data? We answer it
in three steps, moving from a generator we control to two domains that share
nothing but the surjective structure. First, on the Hidden Manifold
Model~\cite{Goldt2020,Gerace2020} --a controlled generator of patterns lying
near a low-dimensional manifold, with a surjective target-- we show that the
exponential capacity survives, degrading gracefully as the manifold shrinks, that
the basins behave as the i.i.d.\ theory predicts, and that generalisation to
unseen points of the manifold is real but weak. Second, on real
T-cell-receptor/epitope triples from VDJdb~\cite{Shugay2018,Bagaev2020}, encoded
by Atchley factors~\cite{Atchley2005} and a locality-sensitive
hash~\cite{Charikar2002}, we find the same phenomenon in sharp form: the network
is a near-perfect content-addressable memory --the two receptor chains recall
the epitope essentially without error, and its basins coincide with the i.i.d.\
prediction-- but its ability to route an unseen receptor to the right epitope,
though several times above a label-permutation null, remains below its
memorisation. Third, on natural-language intent
data --CLINC150~\cite{Larson2019}, a corpus with no biology, no geometry and no
alphabet in common with the previous two-- the same $L=2$ closed forms describe
capacity and basins without a single refitted constant, which is the sharpest
statement of universality we can make, while generalisation to unseen
utterances climbs to $0.58$ against $0.11$ on receptors.

Exponential storage and reliable generalisation thus emerge as distinct
capabilities of one network, and we make the separation quantitative. The
asymmetry should not be read as a defect. A memory of this family is built to
store: with $N$ binary neurons there are $2^{N}$ configurations and the network
occupies an exponential fraction of them with stored associations, so it is a
massive content-addressable repository rather than a low-complexity hypothesis
class, and no bound entitles one to expect strong extrapolation from an object
of that description. What is interesting is not that generalisation is bounded
but what sets the bound: it is the geometry of the encoding, not the data
domain. That data geometry, not merely data quantity, governs whether such
memories generalise is a theme of the random-features and hidden-manifold
Hopfield literature~\cite{Negri2023,Kalaj2024}, which our results extend to the
binary exponential model.

\paragraph{Outline and notation}
Section~\ref{sec:model} defines the model and fixes notation
(Table~\ref{tab:notation}); Section~\ref{sec:dynamics} derives the local field and
the update rule (Algorithm~\ref{algo:hetero}); Sections~\ref{sec:stability}
and~\ref{sec:storage} establish the exponential capacity; Section~\ref{sec:basins}
treats the basins. The numerical narrative occupies
Sections~\ref{sec:hmm}--\ref{sec:clinc}, and Section~\ref{sec:discussion} draws
the memory-versus-classifier lesson. All heavy computations are deferred to the
appendices.

\section{The model}
\label{sec:model}

The system consists of $L$ layers, each composed of $N$ binary neurons,
\begin{equation}
  \sigma_i^{a}\in\{-1,+1\},\qquad a=1,\dots,L,\quad i=1,\dots,N,
  \label{eq:neurons}
\end{equation}
We refer to $L$ as the \emph{width} of the network and to $N$ as the \emph{layer
size}. 
The $L$ layers are not
stacked in cascade, each transforming the output of the previous one, but
mutually coupled through a single symmetric energy and updated one neuron per
layer at a time (Section~\ref{sec:dynamics}), so a wider network here is one that
binds more modalities at once, not one that composes more transformations in
sequence. The neurons come
together with $L$ datasets, one per layer, each of $P$ independent Rademacher patterns,
\begin{equation}
  \xi_i^{\mu,a}\in\{-1,+1\},\qquad
  \P\!\bigl(\xi_i^{\mu,a}=\pm 1\bigr)=\tfrac12,
  \label{eq:patterns}
\end{equation}
for $\mu=1,\dots,P$, $a=1,\dots,L$ and $i=1,\dots,N$.
The patterns are mutually independent across $(\mu,a,i)$.\footnote{This
layer-wise independence is the structural choice that drives the whole analysis:
it suppresses inter-layer statistical mixing at the source. Were a single pattern
replicated across layers, the fluctuations of distinct layer fields at the same
site would become entangled, and cross-pattern terms would survive every average
taken below instead of vanishing by parity. It is an idealisation: real
hetero-associative data couple the layers through a shared latent cause, and
quantifying the price of violating it is one purpose of the experiments.} We
retain the inverse temperature
$\beta:=1/T$ (set to $\beta\to\infty$ throughout) as a control parameter, together
with the number of stored patterns $P$ per layer\footnote{The appropriate intensive
measure of storage, namely the classical ratio $P/N$ being here exponentially large is discussed in Section~\ref{sec:storage}.}. The Mattis magnetisations
(equivalently, Mattis overlaps: we use the two names interchangeably)
\begin{equation}
  m_\mu^a \;:=\; \frac{1}{N}\sum_{i=1}^{N}\xi_i^{\mu,a}\,\sigma_i^{a},
  \label{eq:mattis}
\end{equation}
(for each $\mu=1,\dots,P$ and $a=1,\dots,L$) as order parameters. Table~\ref{tab:notation} collects, once and for all, every
symbol used in the main text.

The cost function reads
\begin{equation}
  \mathcal{H}(\bm\sigma\,|\,\bm\xi)
  \;:=\;
  -N\sum_{\mu=1}^{P}
  \exp\!\biggl[\,N\sum_{a<b}\bigl(m_\mu^{a}\,m_\mu^{b}-1\bigr)\biggr].
  \label{eq:hamiltonian}
\end{equation}
The exponent vanishes on the perfect hetero-associative configuration
$m_\mu^{a}=1$ for every $a$, normalised through the $-1$ subtraction so that on
the recalled archetype $\sum_{a<b}(1\cdot 1-1)=0$.\footnote{A genuinely $L$-way
product $\exp[N\prod_a m_\mu^a-N]$ would be
the literal AND of the $L$ retrieval events, and is not the choice made here.
It would collapse to $0$ discontinuously the moment a single layer left
perfect recall, with no graceful degradation under corruption; the pairwise
sum, in contrast, stays $\mathcal{O}(1)$ as long as even one pair of layers
remains aligned (the mechanism behind the basins of Section~\ref{sec:basins}),
reduces to the $L=1$ model of~\cite{ALBANESE2026131223} without a separate
prescription, and is the only symmetric multilinear form admitting the
collective/orthogonal split of~\eqref{eq:bilinear-id} on which every
saddle-point argument below relies.} Two limits anchor
\eqref{eq:hamiltonian}. The first is $L=1$, and it is instructive that it is
\emph{not} a special case: with a single layer the double sum is empty, the
exponent vanishes identically and $\mathcal{H}\equiv-NP$ ceases to depend on the
configuration. Hetero-association is intrinsically a coupling between at least
two layers, and the auto-associative exponential networks are recovered not by
setting $L=1$ in~\eqref{eq:hamiltonian} but by replacing the missing partner
layer with the pattern itself, $m_\mu^{a}m_\mu^{b}\mapsto m_\mu$ or
$m_\mu^{a}m_\mu^{b}\mapsto m_\mu^{2}$. The first substitution returns the model
of Ref.~\cite{ALBANESE2026131223}, with energy $-N\sum_\mu e^{N(m_\mu-1)}$; the
second returns its $\mathbb{Z}_2$-symmetric sibling, whose saddle point turns
out to coincide with the $L=3$ specialisation of the analysis below. Both are
carried out in parallel with the multilayer computation in
\ref{app:squared_overlap}. The second anchor is the dense series: expanding the
exponent around perfect recall returns, order by order, multi-body couplings of
growing degree, so \eqref{eq:hamiltonian} is again a resummation of all dense
interactions, now with the interaction legs distributed across
layers.\footnote{The identification with the dense associative memories is
literal only in the single-layer case: expanding $e^{N(m_\mu-1)}$ in powers of
$m_\mu$ reproduces exactly the $p$-body couplings of Krotov and
Hopfield~\cite{Krotov2016}, one per order. For $L\ge2$ the $k$-th order of
$\exp[N\sum_{a<b}m_\mu^{a}m_\mu^{b}]$ is a sum of products of $2k$ Mattis
overlaps drawn from distinct layer pairs, hence a $2k$-body spin coupling whose
legs sit in different layers: the hetero-associative analogue of the dense
series rather than the series itself.}

\paragraph{What can be stored}
The energy~\eqref{eq:hamiltonian} treats all layers alike, but a retrieval task
does not: it designates some layers as cue and one as target. Fix the
convention used throughout, layers $1,\dots,L-1$ cue and layer $L$ target, and
read the $P$ stored indices as the graph of a relation
\begin{equation}
  \Psi \;=\; \bigl\{\bigl(\,\mathcal{C}^{\mu},\,\bm\xi^{\mu,L}\,\bigr)\bigr\}_{\mu=1}^{P},
  \qquad
  \mathcal{C}^{\mu} := \bigl(\bm\xi^{\mu,1},\dots,\bm\xi^{\mu,L-1}\bigr),
  \label{eq:rule}
\end{equation}
between cue tuples and targets, and let
$\mathcal{T}=\{\bm\xi^{\mu,L}\}_{\mu\le P}$ be the target codebook, of size
$K=|\mathcal{T}|$. Not every relation is storable: $\Psi$ must be a surjective
\emph{function} from the stored cues onto $\mathcal{T}$. The constraint is not a
modelling preference but a property of the local field, so we state and quantify
it in Remark~\ref{rem:surjective} of Section~\ref{sec:dynamics}, once the field
has been derived; it is the
reason every dataset in Sections~\ref{sec:hmm}--\ref{sec:clinc} is passed
through a {function filter} before the network sees it.

\paragraph{Collective and orthogonal modes}
The bilinear form in the exponent admits the orthogonal decomposition
\begin{equation}
  \sum_{a<b} m^{a}\,m^{b}
  \;=\;
  \tfrac12\Bigl[\,\bigl(\textstyle\sum_a m^{a}\bigr)^{2}
                 -\textstyle\sum_a (m^{a})^{2}\,\Bigr].
  \label{eq:bilinear-id}
\end{equation}
The first term, $\bigl(\sum_a m^{a}\bigr)^{2}$, is the squared collective
magnetisation: it isolates the symmetric mode in which all layers align
coherently with a common archetype, the natural order parameter of
hetero-associative recall. The second term, $\sum_a (m^{a})^{2}$, weights the
residual layer-by-layer dispersion around that collective alignment. The
Hamiltonian~\eqref{eq:hamiltonian} therefore rewards configurations in which
every layer simultaneously and uniformly retrieves the same pattern index, and
penalises any layer-discordant deviation. This separation between a single
collective mode and $L-1$ orthogonal fluctuation modes governs every
saddle-point analysis encountered below.

\begin{table*}[t!]
\centering
\renewcommand{\arraystretch}{1.25}
\begin{tabular}{c l}
\toprule
Symbol & Meaning \\
\midrule
$N$ & neurons per layer (layer size) \\
$L$ & number of layers (network width) \\
$P$ & number of stored patterns per layer \\
$\mathcal{C}^{\mu}=(\bm\xi^{\mu,1},\dots,\bm\xi^{\mu,L-1})$ & stored cue tuple of index $\mu$ \\
$\Psi:\ \mathcal{C}^{\mu}\mapsto\bm\xi^{\mu,L}$ & the stored association rule, a surjective function (Remark~\ref{rem:surjective}) \\
$\mathcal{T},\ K=|\mathcal{T}|$ & target codebook and its size; $P/K$ is the compression \\
$\sigma_i^{a}\in\{-1,+1\}$ & state of neuron $i$ in layer $a$ \\
$\xi_i^{\mu,a}\in\{-1,+1\}$ & bit $i$ of pattern $\mu$ in layer $a$ (Rademacher) \\
$m_\mu^{a}=\frac1N\sum_i\xi_i^{\mu,a}\sigma_i^{a}$ & Mattis overlap of layer $a$ with pattern $\mu$ \\
$\hat m_\mu^{a}=\frac1N\sum_{j\neq i}\xi_j^{\mu,a}\sigma_j^{a}$ & cavity (on-site-excluded) overlap \\
$F_\mu^{a}=\sum_{b\neq a}\hat m_\mu^{b}$ & cross-layer cavity field feeding layer $a$ \\
$\hat E_\mu = N\sum_{a<b}\hat m_\mu^{a}\hat m_\mu^{b}-N\binom{L}{2}$ & cavity energy of pattern $\mu$ \\
$\Phi_\mu^{(\backslash a)}=\exp(\sum_{c\neq a}\xi_i^{\mu,c}\sigma_i^{c}F_\mu^{c})$ & on-site off-layer factor \\
$h_i^{a}$ & local field on neuron $(i,a)$, Eq.~\eqref{eq:local-field} \\
$X_i^{a}=\xi_i^{1,a}h_i^{a}$ & stability variable of the recalled state \\
$\mu_1,\ \sigma^2$ & mean and variance (signal and noise) of $X_i^{a}$ \\
$\rho_L$ & noise/storage rate, Eq.~\eqref{eq:rho-main} \\
$\alpha=\tfrac1N\log P$ & exponential storage rate (intensive load); capacity at $\alpha=\rho_L$ \\
$\gamma$ & reduced load, $\gamma\to0$ retrieval / $\gamma\!\approx\!1$ transition, Eq.~\eqref{eq:gamma-main} \\
$K_L$ & polynomial prefactor of the per-pattern noise, Eq.~\eqref{eq:KL-main} \\
$m^\ast,\ x^\ast=2(L-1)m^\ast$ & symmetric saddle, Eq.~\eqref{eq:saddle-eq-main} \\
$r\in(0,1]$ & initial overlap of a corrupted cue ($d=(1-r)/2$ Hamming radius) \\
$\varepsilon_L(r)$ & storage exponent under corruption, Eq.~\eqref{eq:Pc-r} \\
$\alpha_D=D/N$ & manifold aspect ratio ($D$ latent dim., Section~\ref{sec:hmm}) \\
\bottomrule
\end{tabular}
\caption{Notation used throughout the main text. Cavity quantities are defined
at the site $(i,a)$ being updated; $\binom{L}{2}=L(L-1)/2$.}
\label{tab:notation}
\end{table*}

\section{Local field and dynamical update rule}
\label{sec:dynamics}

The dynamics is studied in the cavity formulation: we isolate a single neuron
$(i,a)$, express the energy as a term independent of it plus a term linear in it,
and read off the field that drives its update. Removing that one neuron from the
interaction network is what defines the \emph{cavity}, and it fixes the meaning
of the whole family of names used below: the cavity overlap $\hat m_\mu^{a}$ is
the Mattis overlap of the punctured system, the cavity field $F_\mu^{a}$ the
field the rest of the network exerts into the hole, the cavity energy
$\hat E_\mu$ and the cavity exponents of Section~\ref{sec:basins} the
corresponding energies and large-deviation rates evaluated there. The device
goes back to Onsager's reaction field~\cite{Onsager1936} and is standard in the
statistical mechanics of disordered
systems~\cite{MezardParisiVirasoro1986,Guerra1995cavity,Barra2006irreducible,MPV1987},
including for the Hopfield model
specifically~\cite{Pastur1999replica,Mezard2017cavity}; we recall the names here
once, since not all of them are equally common outside that literature.
Splitting the on-site contribution
from the magnetisation,
\begin{equation}
  m_\mu^{a} \;=\; \hat m_\mu^{a} + \frac{\xi_i^{\mu,a}\,\sigma_i^{a}}{N},
  \qquad
  \hat m_\mu^{a} \;:=\; \frac{1}{N}\sum_{j\neq i}\xi_j^{\mu,a}\,\sigma_j^{a},
  \label{eq:cavity-split}
\end{equation}
and inserting~\eqref{eq:cavity-split} in the bilinear form,
\begin{align}
  N\sum_{a<b}m_\mu^{a}m_\mu^{b}
  &\;=\;
  N\sum_{a<b}\Bigl(\hat m_\mu^{a}+\tfrac{\xi_i^{\mu,a}\sigma_i^{a}}{N}\Bigr)
            \Bigl(\hat m_\mu^{b}+\tfrac{\xi_i^{\mu,b}\sigma_i^{b}}{N}\Bigr) \notag\\
  &\;=\;
  N\sum_{a<b}\hat m_\mu^{a}\hat m_\mu^{b}
  +\sum_{a<b}\bigl[\xi_i^{\mu,a}\sigma_i^{a}\hat m_\mu^{b}+\xi_i^{\mu,b}\sigma_i^{b}\hat m_\mu^{a}\bigr]\notag\\
  &\qquad
  +\mathcal{O}(N^{-1}).
  \label{eq:bilinear-expansion}
\end{align}
The identity $\sum_{a<b}[X_a Y_b+X_b Y_a]=\sum_c X_c\sum_{d\neq c}Y_d$ with
$X_c=\xi_i^{\mu,c}\sigma_i^{c}$ and $Y_d=\hat m_\mu^{d}$ collapses the second sum
to $\sum_c\xi_i^{\mu,c}\sigma_i^{c}\,F_\mu^{c}$, where
$F_\mu^{c}:=\sum_{d\neq c}\hat m_\mu^{d}$ (the total cavity overlap that all
layers other than $c$ already have with pattern $\mu$) is the inter-layer
cavity field, defined together with the remaining cavity quantities
in~\eqref{eq:cavity-quantities} below. Substituting in the exponential
of~\eqref{eq:hamiltonian} and using $e^{x+\mathcal{O}(N^{-1})}=e^{x}(1+\mathcal{O}(N^{-1}))$,
\begin{equation}
  \mathcal{H}(\bm \sigma | \bm \xi)
  \;=\;
  -N\sum_{\mu=1}^{P}e^{\hat E_\mu}\,
  \exp\!\Bigl(\sum_{c=1}^{L}\xi_i^{\mu,c}\sigma_i^{c}\,F_\mu^{c}\Bigr)
  \bigl(1+\mathcal{O}(N^{-1})\bigr).
  \label{eq:H-cavity}
\end{equation}
Factoring out layer $a$ in the inner exponential and applying the binary identity
$e^{x\sigma}=\cosh x+\sigma\sinh x$ with $x=\xi_i^{\mu,a}F_\mu^{a}$ and
$\sigma=\sigma_i^{a}\in\{-1,+1\}$ produces the additive split
\begin{equation}
  \exp\!\Bigl(\xi_i^{\mu,a}\sigma_i^{a}F_\mu^{a}\Bigr)
  \;=\;
  \cosh\!\bigl(\xi_i^{\mu,a}F_\mu^{a}\bigr)
  +\sigma_i^{a}\sinh\!\bigl(\xi_i^{\mu,a}F_\mu^{a}\bigr).
\end{equation}
Collecting the $\sigma_i^{a}$-independent terms in $C_i^{a}$ and the linear ones
in $h_i^{a}$ yields
\begin{equation}
  \mathcal{H}(\bm \sigma | \bm \xi)
  \;=\;
  -N\bigl[\,C_i^{a}+\sigma_i^{a}\,h_i^{a}\,\bigr]
  \bigl(1+\mathcal{O}(N^{-1})\bigr),
  \label{eq:H-decomp}
\end{equation}
with the local field and shift, respectively,
\begin{align}
  h_i^{a}
  &\;:=\;
  \sum_{\mu=1}^{P} e^{\hat E_\mu}\,
  \sinh\!\bigl(\xi_i^{\mu,a}\,F_\mu^{a}\bigr)\,
  \Phi_\mu^{(\backslash a)},
  \label{eq:local-field}\\
  C_i^{a}
  &\;:=\;
  \sum_{\mu=1}^{P} e^{\hat E_\mu}\,
  \cosh\!\bigl(\xi_i^{\mu,a}\,F_\mu^{a}\bigr)\,
  \Phi_\mu^{(\backslash a)}.
  \label{eq:Ci}
\end{align}
The cavity quantities are
\begin{equation}
\begin{aligned}
  F_\mu^{a} &\;:=\;\sum_{b\neq a}\hat m_\mu^{b},
  \qquad
  \hat E_\mu \;:=\; N\sum_{a<b}\hat m_\mu^{a}\hat m_\mu^{b}-N\binom{L}{2},\\
  \Phi_\mu^{(\backslash a)} &\;:=\; \exp\!\Bigl(\sum_{c\neq a}\xi_i^{\mu,c}\sigma_i^{c}\,F_\mu^{c}\Bigr).
\end{aligned}
  \label{eq:cavity-quantities}
\end{equation}
Each of these has a plain reading. The \emph{inter-layer cavity field}
$F_\mu^{a}$ is the total overlap that all layers other than $a$ already
have with pattern $\mu$: it is the pressure the rest of the network exerts on
layer $a$ to also retrieve $\mu$, and it is what makes the model
hetero-associative: a neuron in one layer is driven by the state of the
others. The \emph{cavity energy} $\hat E_\mu$ measures how well pattern $\mu$ is
collectively retrieved across all layer pairs; its $-N\binom{L}{2}$ floor sends
$e^{\hat E_\mu}\to0$ for any pattern that is not being retrieved, so only the
winning pattern contributes to the field. The on-site factor
$\Phi_\mu^{(\backslash a)}$ is the same weight restricted to the off-layer bits
at the site under update. The remainder in~\eqref{eq:H-decomp} collects the
on-site square contributions
$\sum_{a<b}(\xi_i^{\mu,a}\sigma_i^{a})(\xi_i^{\mu,b}\sigma_i^{b})/N$ generated
by~\eqref{eq:cavity-split}; these are bounded in absolute value by
$\binom{L}{2}/N$, are invariant under $\xi_i^{a}\sigma_i^{a}\mapsto-\xi_i^{a}\sigma_i^{a}$
to leading order, and do not affect the single-flip energy difference at
$\mathcal{O}(1)$. The cavity decomposition is exact at this order.

The field~\eqref{eq:local-field} has a transparent reading. Each stored pattern
$\mu$ casts a vote on neuron $(i,a)$, weighted by $e^{\hat E_\mu}$ (how well
the rest of the network already agrees with pattern $\mu$ across all layers)
and directed by $\sinh(\xi_i^{\mu,a}F_\mu^{a})$, the alignment that the
other layers ask of layer $a$. The exponential weight is what makes the
network hetero-associative and high-capacity at once: a pattern that is being
collectively retrieved dominates the sum, while the $e^{-N\binom{L}{2}}$ floor of
$\hat E_\mu$ silences every pattern that is not.

The single-flip energy variation is $\Delta E_i^{a}=2N\,h_i^{a}\,\sigma_i^{a}$,
exact when all spins other than $(i,a)$ are held fixed, so the zero-temperature
Glauber update is steepest descent,
\begin{equation}
  \sigma_i^{a}(t+1) \;=\; \sign\!\bigl[\,h_i^{a}(t)\,\bigr],
  \label{eq:update}
\end{equation}
in words (Algorithm~\ref{algo:hetero}): every neuron looks at how strongly each
stored pattern is currently being retrieved network-wide, and flips to agree with
the winner.

The schedule with which~\eqref{eq:update} is applied deserves a word, because
``parallel'' is used here in a restricted sense. One elementary step selects a
single site index $i$ and updates the $L$ neurons $(i,1),\dots,(i,L)$ --one per
layer-- simultaneously, from the fields evaluated on the current state; the
site index is then advanced along a fixed or randomly shuffled permutation of
$\{1,\dots,N\}$, and a sweep is complete once all $N$ sites have been visited.
The update is thus \emph{parallel across layers and sequential within each
layer}: the $N$ spins of a given layer are never flipped at once. Keeping the
within-layer update sequential is what preserves the cavity
decomposition~\eqref{eq:H-decomp}, which holds every other spin of layer $a$
fixed while $(i,a)$ is updated; flipping a whole layer synchronously would
change the overlaps $\hat m_\mu^{a}$ by $\mathcal{O}(1)$ and invalidate the
fields on which the flips were decided.

The simultaneity across layers is a genuine, if mild, synchronicity, and it is
worth locating exactly. The cavity overlaps $\hat m_\mu^{b}$, hence the fields
$F_\mu^{a}$ and the weights $e^{\hat E_\mu}$, exclude site $i$ in {every}
layer and are therefore unaffected by the $L$ flips; the only dependence of
$h_i^{a}$ on the spins updated alongside it is through the on-site factor
$\Phi_\mu^{(\backslash a)}$. In the retrieval regime the sum
in~\eqref{eq:local-field} is dominated by the one pattern with
$e^{\hat E_\mu}=\mathcal{O}(1)$, and the $L$ simultaneous updates then all point
at that same stored index, so the coupling is benign; away from it the rule is
simply taken as the definition of the dynamics. In either case the stability
analysis of Sections~\ref{sec:stability}--\ref{sec:storage} is a single-site
statement, $X_i^{a}>0$, and is insensitive to the order in which sites are
visited.

\begin{algorithm*}[t!]
\caption{Zero-temperature update of the exponential hetero-associative network}
\label{algo:hetero}
\KwIn{layer datasets $\{\xi^{\mu,a}\}$; initial state $\bm\sigma^{(0)}\in\{-1,+1\}^{LN}$; number of sweeps $N_p$.}
\KwOut{fixed point $\bm\sigma^{\ast}$.}
\For{$t = 1$ \KwTo $N_p$}{
  \ForEach{site $i$ in a fixed or shuffled permutation of $1,\dots,N$ \emph{(sequentially)}}{
    \ForEach{pattern $\mu$ and layer $a$}{
      $\hat m_\mu^{a}\leftarrow \tfrac1N\sum_{j\neq i}\xi_j^{\mu,a}\sigma_j^{a}$ \tcp*{how well pattern $\mu$ is retrieved in layer $a$, site $i$ excluded}
    }
    \ForEach{pattern $\mu$}{
      $\hat E_\mu \leftarrow N\sum_{a<b}\hat m_\mu^{a}\hat m_\mu^{b}-N\binom{L}{2}$ \tcp*{collective retrieval weight of $\mu$}
    }
    \ForEach{layer $a = 1,\dots,L$ \emph{(simultaneously: one neuron per layer)}}{
      $h_i^{a}\leftarrow \sum_{\mu} e^{\hat E_\mu}\sinh\!\bigl(\xi_i^{\mu,a}F_\mu^{a}\bigr)\,\Phi_\mu^{(\backslash a)}$ \tcp*{Eq.~\eqref{eq:local-field}}
      $\sigma_i^{a}\leftarrow \sign\!\bigl(h_i^{a}\bigr)$\;
    }
  }
}
\end{algorithm*}

The field also settles the question left open in Section~\ref{sec:model}: which
association rules~\eqref{eq:rule} the energy can hold.

\begin{remark}[The stored rule must be a surjective function]
\label{rem:surjective}
For~\eqref{eq:hamiltonian} to operate as a hetero-associative memory, $\Psi$
must be a {function} of the cue and a {surjection} onto the target
codebook $\mathcal{T}$.

\emph{(i) Single-valuedness.} Suppose $q\ge2$ stored indices share a cue,
$\mathcal{C}^{\lambda}=\mathcal{C}$ for $\lambda\in S$ with $|S|=q$, but carry
distinct targets. Clamping the cue at $\mathcal{C}$ makes the cavity energy
$\hat E_\lambda$, the inter-layer field $F_\lambda^{L}$ and the on-site factor
$\Phi_\lambda^{(\backslash L)}$ of~\eqref{eq:cavity-quantities} identical for
every $\lambda\in S$, and exponentially suppressed for every other pattern.
Since $\sinh$ is odd, the target field~\eqref{eq:local-field} collapses to
\begin{equation}
  h_i^{L}\;=\;e^{\hat E}\,\sinh\!\bigl(F^{L}\bigr)\,\Phi^{(\backslash L)}
              \sum_{\lambda\in S}\xi_i^{\lambda,L}
  \;\bigl(1+o(1)\bigr),
  \label{eq:majority}
\end{equation}
so the update~\eqref{eq:update} returns the componentwise majority of the $q$
contradictory targets. For independent targets the overlap of that majority with
any one of them is exactly
$\binom{q-1}{(q-1)/2}2^{-(q-1)}\simeq\sqrt{2/(\pi q)}$ for odd $q$: unity at
$q=1$, one half at $q=3$, decaying to zero thereafter. A cue carrying two
targets is therefore not stored badly but not stored at all; what the network
returns is a mixture of them, and the failure is a property of the energy rather
than of the dynamics.

\emph{(ii) Surjectivity.} Conversely, the addressable alphabet is exactly the
image $\Psi(\{\mathcal{C}^{\mu}\})$. A code sitting in layer $L$ that no stored
cue points to never dominates the field, its weight $e^{\hat E}$ is
$e^{-N\binom{L}{2}}$-suppressed for every cue, so it carries an empty basin
while still contributing its share to the noise floor of
Section~\ref{sec:stability}: capacity spent on a memory that cannot be recalled.
Reading $\Psi$ as a surjection onto $\mathcal{T}$ is what excludes this.

Injectivity, by contrast, is neither required nor desirable. The regime of
interest is $P\gg K$, with $P/K$ the mean number of cues per target: it is the
regime real data supply (many receptors per epitope, many phrasings per intent),
and $P/K$ controls the depth of the target's basin
(Section~\ref{sec:hmm}). A bijective rule would make the reverse direction a
function as well, and the model would collapse to an auto-associative memory on
the concatenated vector: the situation the construction was meant to improve
on.
\end{remark}

\paragraph{Computational cost: time and memory}
Algorithm~\ref{algo:hetero} is exact, but not free, and the exponential
capacity of Section~\ref{sec:storage} carries an exponential price tag that is
worth making explicit. The first loop evaluates $PL$ inner products of length
$N$: $\Theta(NLP)$ operations. The second costs $\Theta(PL^{2})$. The third,
read literally, costs $\Theta(NL^{2}P)$, because $\Phi_\mu^{(\backslash a)}$
sums $L-1$ terms at every site; caching the full on-site sum
$T_i^{\mu}:=\sum_{c}\xi_i^{\mu,c}\sigma_i^{c}F_\mu^{c}$ once per $(i,\mu)$ and
subtracting the $a$-th term at evaluation time reduces this to $\Theta(NLP)$.
One parallel sweep of the whole network therefore costs $\Theta(NLP)$ time
whenever $N\gtrsim L$ (true throughout this paper), and $\Theta(NLP)$ bits is
also the memory floor, set by the $L$ stored datasets and dwarfing the
$\Theta(NL)$ state and the $\Theta(PL)$ cavity scalars
$\{\hat m_\mu^{a},\hat E_\mu,F_\mu^{a}\}$. Time and memory per sweep are thus
both, to leading order, one read of the stored data, and $n_{\mathrm{steps}}$
sweeps cost $\Theta(n_{\mathrm{steps}}NLP)$.

The cost is linear in $P$, but the capacity of Section~\ref{sec:storage} is
exponential in $N$: running the network anywhere near capacity costs
$\Theta(NL\,e^{N\rho_L})$ per sweep, at exactly the rate $\rho_L$ that sets the
storage. Already at $N=64$ and $L=2$ no machine could hold the stored data, let
alone sweep them. The capacity statement is about which configurations are fixed
points of~\eqref{eq:update}, not a promise that all of them can be enumerated;
this is why every experiment below uses $N$ of order ten and loads far below
$P_c$, where the degradation transition falls in an accessible window. The
arithmetic is made explicit in \ref{app:cost}.

\section{Stability of the recalled ground state}
\label{sec:stability}

We analyse the stability of the hetero-associative configuration
\begin{equation}
  (\bm\sigma^{1},\dots,\bm\sigma^{L}) \;=\; (\bm\xi^{1,1},\dots,\bm\xi^{1,L}),
  \label{eq:trial}
\end{equation}
in which all layers retrieve the same archetype index $\mu_0=1$, each from its
own dataset. By~\eqref{eq:update}, stability against a single-spin flip at site
$i$ in layer $a$ amounts to
\begin{equation}
  X_i^{a} \;:=\; \xi_i^{1,a}\,h_i^{a}\bigl|_{\bm\sigma^{b}=\bm\xi^{1,b}} \;>\; 0.
  \label{eq:stability}
\end{equation}
The strategy is the classic signal-to-noise decomposition of attractor neural
networks~\cite{AGS1985a,Amit1989,Coolen2005}, in the one-step, zero-temperature
form used for exponential models in
Refs.~\cite{Demircigil2017,ALBANESE2026131223}: the field is a sum of
$P$ pattern votes; the vote of the recalled pattern is a deterministic
\emph{signal}, the remaining $P-1$ are mean-zero \emph{noise}, and one asks when
the first dominates the fluctuations of the second. In the limit of
large $N$ and $P$\footnote{Throughout, $N\to\infty$ is taken first, at fixed
$\mu$-th pattern, to obtain the exact rate $\rho_L$ and prefactor $K_L$ below;
the load $P$ is then let grow, at fixed $N$, as $P=\gamma\,e^{N\rho_L}$ for a
fixed load fraction $\gamma$, so that the one-step overlap collapses onto the
universal profile $m_1^{(1)}=\erf(1/\sqrt{2\gamma})$ of~\eqref{eq:m1step} (used
again in \ref{app:hmm_protocol}). No double limit is taken simultaneously.} the Central Limit Theorem renders $X_i^{a}$ Gaussian,
$X_i^{a}\sim\mathcal{N}\!\bigl(\mu_1,\sqrt{\mu_2-\mu_1^{2}}\bigr)$, and stability
holds with overwhelming probability as long as the signal dominates the noise
standard deviation. (The Berry--Esseen control of this approximation, uniform in
$N$ because the noise terms are i.i.d.\ across the independent layer datasets, is
given in \ref{app:moments}.) Everything therefore rests on two moments.

\subsection{The signal}
\label{sec:mu1}

At the trial state~\eqref{eq:trial} the cavity magnetisation of the recalled
archetype is deterministic, $\hat m_1^{a}=(N-1)/N$ for every $a$, while for
$\mu\neq 1$ the layer-$a$ on-site factor $\xi_i^{\mu,a}$ is independent of every
quantity in the corresponding noise term and has zero mean. The mean of
$X_i^{a}$ reduces to the signal alone; the computation (\ref{app:mu1})
gives
\begin{equation}
    \mu_1 \;=\; \E\bigl[X_i^{a}\bigr]
    \;=\; e^{-(L-1)}\sinh(L-1) + \mathcal{O}(N^{-1}).
  \label{eq:mu1}
\end{equation}
An independent check via the discrete energy difference at the trial state,
$\Delta E_i^{a}=N\bigl(1-e^{-2(L-1)}\bigr)$, is given in
\ref{app:flip-check}. The signal is of order one and grows with the width:
$\mu_1=\tfrac12(1-e^{-2(L-1)})\to\tfrac12$ as
$L\to\infty$.

\subsection{The noise}
\label{sec:mu2}

The second moment splits as
$\mu_2=\sum_{\mu,\nu}\E[X_i^{(\mu|a)}X_i^{(\nu|a)}]=\mu_2^{(\mathrm{diag})}+\mu_2^{(\mathrm{off})}$.
Layer-wise dataset independence kills the off-diagonal part exactly
(\ref{app:offdiag}): a single uncancelled on-site Rademacher factor
carries zero mean. The diagonal part is a deterministic signal-square,
\begin{equation}
  \E\bigl[(X_i^{(1|a)})^{2}\bigr] \;=\; e^{-2(L-1)}\sinh^{2}(L-1) + \mathcal{O}(N^{-1}),
  \label{eq:Xself-sq}
\end{equation}
plus $P-1$ identical per-pattern noise contributions. Here the model departs from
its single-layer ancestor. Each noise term reduces, after averaging the on-site
factor, to the cavity expectation
\begin{equation}
  \E\bigl[(X_i^{(\mu|a)})^{2}\bigr]
  \;=\;
  \E\!\Bigl[\,e^{2\hat E_\mu}\sinh^{2}(F_\mu^{a})
            \prod_{c\neq a}\cosh(2F_\mu^{c})\,\Bigr],
  \quad \mu\neq 1,
  \label{eq:noise-sq-main}
\end{equation}
where the expectation runs over the cavity magnetisations
$\bm{\hat m}_\mu\in[-1,1]^{L}$. In the auto-associative model the exponent is
\emph{linear} in the Rademacher variables and this average factorises over sites
into a closed form; the product-over-layers structure of~\eqref{eq:noise-sq-main}
makes the exponent effectively \emph{quadratic}, the average no longer
factorises, and a genuine large-deviation evaluation is required.

By Cram\'er's theorem the empirical magnetisation of each layer obeys a large
deviation principle with the symmetric-Bernoulli rate function; independence
across layers adds the rates; and Varadhan's
lemma~\cite{Varadhan1966,DemboZeitouni1998} turns~\eqref{eq:noise-sq-main}
into a variational problem. The functional is permutation-symmetric in the
layers, and its unique non-trivial saddle sits on the symmetric ray
$m^{c}\equiv m^{\ast}$, the maximally aligned direction of~\eqref{eq:bilinear-id}, solving
\begin{equation}
  m^{\ast}\;=\;\tanh\!\bigl(2(L-1)\,m^{\ast}\bigr).
  \label{eq:saddle-eq-main}
\end{equation}
We stress that the full derivation is in 
\ref{app:saddle}. Its output is two quantities. The exponential decay
rate of a single-pattern noise contribution is
\begin{equation}
  {\;
  \begin{aligned}
    \rho_L &\;=\; L\bigl[(L-1)-\phi_L(x^{\ast})\bigr],\\
    \phi_L(x) &\;=\; -\frac{x^{2}}{4(L-1)}+\log\cosh(x),
  \end{aligned}
  \;}
  \label{eq:rho-main}
\end{equation}
with $x^{\ast}=2(L-1)m^{\ast}$ the unique positive solution of
$\tanh(x)=x/[2(L-1)]$. The polynomial prefactor encoding the Gaussian
fluctuations around the saddle and the on-site insertions evaluated at
$\bm m=m^{\ast}\bm 1$ is
\begin{equation}
  K_L \;=\;
  \frac{(2\pi)^{L/2}}{\sqrt{\det\mathcal{H}_L^{\ast}}}\,
  \sinh^{2}\!\bigl((L-1)m^{\ast}\bigr)\,
  \cosh\!\bigl(2(L-1)m^{\ast}\bigr)^{L-1},
  \label{eq:KL-main}
\end{equation}
with $\mathcal{H}_L^{\ast}$ the Hessian of the variational functional at the
saddle (\ref{app:saddle}). Hence
\begin{equation}
  \E\bigl[(X_i^{(\mu|a)})^{2}\bigr]
  \;=\;
  K_L\,e^{-N\rho_L}\bigl(1+o(1)\bigr).
  \label{eq:Xmu-sq-main}
\end{equation}
Combining~\eqref{eq:Xself-sq} with the $P-1$ noise contributions,
\begin{equation}
    \mu_2 \;=\; e^{-2(L-1)}\sinh^{2}(L-1)
    \;+\; (P-1)\,K_L\,e^{-N\rho_L}\bigl(1+o(1)\bigr),
  \label{eq:mu2}
\end{equation}
\begin{equation}
  \sigma^{2}\;=\;\mu_2-\mu_1^{2}\;=\;(P-1)\,K_L\,e^{-N\rho_L}\bigl(1+o(1)\bigr).
  \label{eq:variance}
\end{equation}
The signal is order one; the per-pattern noise variance is exponentially small in
$N$. The whole storage phenomenon is the competition between these two facts,
made quantitative in the next section. Numerical values of $x^\ast$, $m^\ast$,
$\phi_L$ and $\rho_L$ for moderate $L$ are tabulated in
\ref{app:saddle-numbers}; the rate grows monotonically and
$\rho_L\sim L\log 2$.

\begin{remark}[Annealed versus typical noise]
\label{rem:annealed-noise}
The word \emph{annealed} is used here in the sense familiar from the free
energy, $\log\E[Z]$ against $\E[\log Z]$, transposed from a partition function
to a moment: the quantity we evaluate is the disorder average of an exponential,
$\E[e^{2\hat E_\mu}\cdots]$, not the exponential of a disorder average. As
always, the two differ when the average is dominated by rare realisations, and
here it is. The Laplace evaluation of~\eqref{eq:noise-sq-main} is controlled by
the saddle $\bm{\hat m}_\mu\equiv m^\ast\bm 1$, an alignment of magnitude
$\mathcal{O}(1)$ between the cavity state and an unretrieved pattern; for a
genuine Rademacher pattern such an alignment is exponentially rare, the typical
cavity overlaps being $\mathcal{O}(N^{-1/2})$. The typical per-pattern noise is
accordingly parametrically smaller than~\eqref{eq:Xmu-sq-main}, so the
variance~\eqref{eq:variance} is a conservative overestimate of the noise and the
resulting capacity a conservative underestimate. Nothing is lost in rigour by
this (the bound is one-sided in the safe direction) but the gap widens
sharply with $L$, and it is the reason the empirical retrieval in
Sections~\ref{sec:hmm}--\ref{sec:clinc} tends to overshoot the annealed-noise
prediction at $L\ge 3$.
\end{remark}

\section{Storage capacity}
\label{sec:storage}

Within the Gaussian approximation the stability condition $X_i^{a}>0$ holds with
probability
\begin{equation}
  \P\bigl(X_i^{a}>0\bigr)
  \;=\;
  1-\tfrac12\,\erfc\!\Bigl(\tfrac{\mu_1}{\sqrt{2\sigma^{2}}}\Bigr).
\end{equation}
Requiring a per-spin error probability $\le N^{-a}$, $a>0$, so that the union
bound over the $NL$ pairs $(i,a)$ stays summable in the thermodynamic limit,
gives $\mu_1^{2}/(2\sigma^{2}) \ge a\log N + \mathcal{O}(\log\log N)$.
Substituting the signal~\eqref{eq:mu1} and the variance~\eqref{eq:variance},
\begin{equation}
    {P \;\le\; 1+\frac{e^{-2(L-1)}\sinh^{2}(L-1)}{2\,a\,K_L\,\log N}\,e^{N\rho_L},}
  \label{eq:Pmax}
\end{equation}
so that the leading exponential storage capacity is
\begin{equation}
  P_c \;\sim\; e^{N\rho_L},
  \qquad
  \rho_L \;=\; L\bigl[(L-1)-\phi_L(x^{\ast})\bigr].
  \label{eq:Pc}
\end{equation}
Equivalently, the Mattis magnetisation after one sweep of the update
rule~\eqref{eq:update}, one neuron per layer per step, simultaneously across
the $L$ layers, until all $N$ sites have been visited once
(Algorithm~\ref{algo:hetero}), reads
\begin{equation}
  {m_1^{(1)}
  \;=\;
  \erf\!\left(\frac{e^{-(L-1)}\sinh(L-1)}
                            {\sqrt{2(P-1)\,K_L\,e^{-N\rho_L}}}\right),}
  \label{eq:m1step}
\end{equation}
which tends to unity as long as its argument diverges, i.e.\ as long as
$P\,e^{-N\rho_L}\to 0$. The capacity is therefore exponential in $N$, with a rate
$\rho_L$ that increases with the width $L$, $\rho_L\sim L\log 2$ as $L\to\infty$
(\ref{app:saddle-numbers}). Wider hetero-associative networks, binding more
layers at once, are exponentially more capacious.

\paragraph{The intensive load}
Because the capacity $P_c\sim e^{N\rho_L}$ is exponential in $N$, the
Amit--Gutfreund--Sompolinsky ratio $P/N$ -the intensive control parameter of
the classical theory- is itself exponentially large here and says nothing
about proximity to capacity. Two intensive quantities take its place. The first
is the exponential storage rate
\begin{equation}
  \alpha \;:=\; \frac{\log P}{N},
  \qquad \alpha_c=\rho_L,
  \label{eq:alpha-main}
\end{equation}
which measures $P$ on its natural exponential scale and reaches capacity exactly
at $\alpha=\rho_L$. The second, finer one is the reduced load
\begin{equation}
  \gamma \;:=\; \frac{(P-1)\,K_L\,e^{-N\rho_L}}{\mu_1^{2}}
        \;=\; \frac{\sigma^{2}}{\mu_1^{2}},
  \label{eq:gamma-main}
\end{equation}
namely the ratio of the noise variance~\eqref{eq:variance} to the squared
signal~\eqref{eq:mu1}, in terms of which the one-step overlap~\eqref{eq:m1step}
collapses onto the parameter-free profile $m_1^{(1)}=\erf(1/\sqrt{2\gamma})$:
$\gamma\to0$ is deep retrieval, $\gamma\approx1$ the transition, $\gamma\gg1$
failure. It is $\gamma$, not $P$, that is comparable across layer sizes and widths, and
we use it as the intensive load throughout; the figures nonetheless keep $P$ on
the abscissa, where the exponential span of the capacity is directly visible.

Two remarks fix the meaning of this result before we test it. First, because
$P_c$ is exponential in $N$, the memory-degradation transition is visible only at
small $N$: already at $N=64$, $L=2$ one has $P_c\sim e^{86}$, out of
computational reach both in memory and in time, in a sense we make precise at
the end of Section~\ref{sec:dynamics} ($P_c\approx2.8\times10^{37}$ patterns,
whose per-sweep cost alone dwarfs any conceivable machine). Every simulation below therefore uses
$N$ of order ten, where the transition falls in an accessible window. Second, as anticipated in
Remark~\ref{rem:annealed-noise}, the annealed prediction for the noise places the
transition at exponentially smaller $P$ than the network actually
realises; the curves labelled ``typical'' below use the empirically calibrated
per-pattern variance and are the ones that track the data.

Figure~\ref{fig:storage} validates the exponential capacity and displays its
analytic backbone. Panel~(a) compares the one-step
prediction~\eqref{eq:m1step} with the i.i.d.\ Monte-Carlo battery at width $L=2$
and layer sizes $N=8,\dots,12$: the recall plateau $m_1^{(1)}\simeq1$ gives way to the
disordered regime $m_1^{(1)}\simeq0$ where $P\sim e^{N\rho_2}$, and the transition
marches to exponentially larger $P$ as $N$ grows, exactly as~\eqref{eq:Pc}
demands. Panel~(b) solves the symmetric saddle~\eqref{eq:saddle-eq-main}
graphically, its inset showing the rate $\rho_L$ climbing to its $L\log2$
asymptote.

\begin{figure*}[t!]
    \centering
    \includegraphics[width=\linewidth]{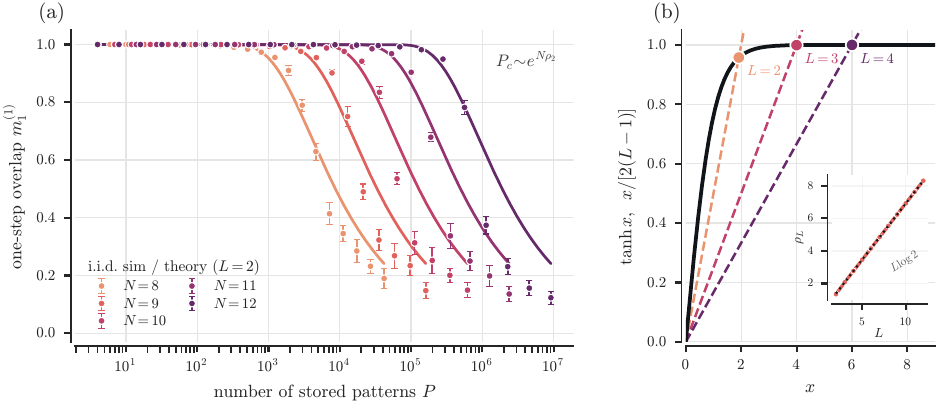}
    \caption{\textbf{Exponential storage capacity.}
    \textbf{(a)} One-step overlap $m_1^{(1)}$ versus the number of stored patterns
    $P$ for the i.i.d.\ ensemble at width $L=2$ and layer sizes $N=8,9,10,11,12$: solid
    lines are the closed form~\eqref{eq:m1step}, markers Monte-Carlo
    (mean\,$\pm$\,std over disorder realisations). The degradation transition
    tracks $P_c\sim e^{N\rho_2}$ and shifts right with $N$; because the capacity is
    exponential in $N$, only small $N$ brings it into an accessible window.
    \textbf{(b)} Graphical solution of the symmetric saddle
    $\tanh x^\ast=x^\ast/[2(L-1)]$: the non-trivial crossing $x^\ast$ (markers)
    drifts towards $2(L-1)$ as the network widens. Inset: the noise rate
    $\rho_L=L[(L-1)-\phi_L(x^\ast)]$ against the width $L$, compared to
    $L\log2$ (dotted).}
    \label{fig:storage}
\end{figure*}

\section{Basins of attraction}
\label{sec:basins}

Having established that the hetero-associative configuration~\eqref{eq:trial} is a
fixed point up to an exponential load, we ask how far from it the dynamics can
start and still return. Following the protocol of the single-layer
model~\cite{ALBANESE2026131223}, we initialise in a corrupted archetype,
\begin{equation}
  \sigma_j^{a}(0) \;=\; s_j^{a}\,\xi_j^{1,a},
  \qquad
  \E\bigl[s_j^{a}\bigr]=r\in(0,1],
  \label{eq:corrupted-init}
\end{equation}
with masks $s_j^{a}\in\{-1,+1\}$ i.i.d.\
independently across sites and layers, so that every layer carries the same
initial overlap $m_1^{a}(0)=r$, i.e.\ a Hamming distance per neuron
$d=(1-r)/2$. The stability variable is again $X_i^{a}=\xi_i^{1,a}h_i^{a}$,
now evaluated at~\eqref{eq:corrupted-init}, and decomposes over patterns as
\begin{equation}
\begin{aligned}
  X_i^{a}&\;=\;\sum_{\mu=1}^{P}X_i^{(\mu|a)},\\
  X_i^{(\mu|a)}&\;:=\;\xi_i^{1,a}\,e^{\hat E_\mu}\,\sinh\!\bigl(\xi_i^{\mu,a}F_\mu^{a}\bigr)\,\Phi_\mu^{(\backslash a)}.
\end{aligned}
  \label{eq:X-decomp-r}
\end{equation}
All details are in \ref{app:corrupted}; three facts organise the result.

\emph{First, the noise is blind to the corruption.} For $\mu\neq1$ the relabelled
variables $u_j^{(\mu,c)}=\xi_j^{\mu,c}s_j^{c}\xi_j^{1,c}$ are i.i.d.\ symmetric
Rademacher for every $r$, so each off-pattern contributes zero mean and
the same per-pattern variance $K_L\,e^{-N\rho_L}$ as in~\eqref{eq:variance}: the
noise floor does not move when the cue degrades.

\emph{Second, the signal pays a large-deviation cost.} At the corrupted state the
recalled cavity overlaps concentrate at $r$ rather than $1$, and the typical
signal decays as
\begin{equation}
  \frac{1}{N}\log X_i^{(1|a)}
  \;\xrightarrow[N\to\infty]{\mathrm{a.s.}}\;
  -\binom{L}{2}\bigl(1-r^{2}\bigr),
  \label{eq:signal-typ-main}
\end{equation}
whereas the annealed signal $\mu_1(r)=\E[X_i^{(1|a)}]$ decays with the
strictly smaller rate
\begin{equation}
\begin{aligned}
  \Sigma_L(r) \;=\;{}& \binom{L}{2}\bigl(1+m_s^{2}\bigr)\\
  &-L\Bigl[\log\cosh\bigl((L-1)m_s+a_r\bigr)-\log\cosh(a_r)\Bigr]\\[2pt]
  &\;<\;\binom{L}{2}\bigl(1-r^{2}\bigr),
  \qquad a_r:=\tanh^{-1}(r),
\end{aligned}
  \label{eq:Sigma-main}
\end{equation}
where $m_s=m_s(r)$ solves the $r$-biased counterpart of~\eqref{eq:saddle-eq-main},
\begin{equation}
  m_s\;=\;\tanh\bigl((L-1)\,m_s+a_r\bigr).
  \label{eq:saddle-corrupted-main}
\end{equation}
\emph{Third, the signal stays strictly positive} at every site: corruption only
shrinks it, so the transition remains a competition between an (exponentially
small) signal and the noise.

Writing $\mu_1(r):=\E\bigl[X_i^{(1|a)}\bigr]$ for the annealed signal at
corruption $r$ (the same first moment as~\eqref{eq:mu1}, now evaluated at the
corrupted initial state~\eqref{eq:corrupted-init}, so that
$\mu_1(r)=C_L(r)\,e^{-N\Sigma_L(r)}$ with the rate~\eqref{eq:Sigma-main} and the
prefactor $C_L(r)$ of Eq.~\eqref{eq:CLr-app}, and $\mu_1(1)=\mu_1$) and
repeating the signal-to-noise argument of Section~\ref{sec:storage},
\begin{equation}
  m_1^{(1)}
  \;=\;
  \erf\!\left(\frac{\mu_1(r)}{\sqrt{2(P-1)\,K_L\,e^{-N\rho_L}}}\right),
  \label{eq:m1step-r}
\end{equation}
which tends to unity iff $P\,e^{-N\varepsilon_L(r)}\to0$, with storage exponent
and capacity
\begin{equation}
  {\;
  P_c(r)\;\sim\;e^{N\varepsilon_L(r)},
  \qquad
  \varepsilon_L(r)\;=\;\rho_L-2\,\Sigma_L(r),
  \;}
  \label{eq:Pc-r}
\end{equation}
in the annealed scheme, and $\varepsilon_L^{\mathrm{typ}}(r)=\rho_L-L(L-1)(1-r^{2})$
if the annealed signal rate is replaced by the typical one~\eqref{eq:signal-typ-main}.
The quantitative load estimate, obtained as in~\eqref{eq:Pmax} by requiring a
per-spin error probability $\le N^{-a}$, is
\begin{equation}
  P \;\le\; 1+\frac{C_L(r)^{2}}{2\,a\,K_L\,\log N}\;e^{N\varepsilon_L(r)},
  \label{eq:Pmax-r}
\end{equation}
with the prefactor $C_L(r)$ of Eq.~\eqref{eq:CLr-app}. The capacity stays
exponential at any corruption below threshold, the price of larger basins
being a smaller rate; above threshold one-step recall of the corrupted cue fails.
The thresholds,
\begin{equation}
\begin{aligned}
  &\Sigma_L(r_c)=\tfrac{\rho_L}{2}
  &&\text{(annealed)},\\
  &r_c^{\mathrm{typ}}=\sqrt{1-\tfrac{\rho_L}{L(L-1)}}
  &&\text{(typical)},
\end{aligned}
  \label{eq:rc-def}
\end{equation}
are reported in Table~\ref{tab:basins}. Both criteria expose the same trade-off:
as the network widens the rate grows ($\rho_L\sim L\log2$) but the basins
shrink. In the typical criterion $1-r_c^{2}=\rho_L/[L(L-1)]\simeq\log2/(L-1)\to0$,
so the tolerated Hamming radius vanishes as $d_c\simeq\log2/[4(L-1)]$; the
annealed criterion saturates instead at the finite limit
$r_c\to\sqrt2-1\approx0.4142$ (\ref{app:corrupted}). The annealed signal
$\mu_1(r)$ is dominated by exponentially rare corruption masks, so $r_c^{\mathrm{ann}}$
is the optimistic estimate and $r_c^{\mathrm{typ}}$ the conservative one, and the
two bracket the finite-$N$ recovery threshold.

Which of the two thresholds the finite-$N$ dynamics realises is an empirical
question, but one for which the theory already indicates the answer, and the
argument is worth giving before the simulations confirm it. The one-step
prediction~\eqref{eq:m1step-r} is built on the annealed first moment
$\mu_1(r)=\E[X_i^{(1|a)}]$, the prescription that fixes the corrupted-cue capacity
in the single-layer model~\cite{ALBANESE2026131223}. There the exponent is linear
in the masks, the annealed average factorises site by site, and annealed and
typical coincide, so no distinction arises. The multilayer exponent is instead
quadratic in the masks, which is what opens the gap
$r_c^{\mathrm{ann}}<r_c^{\mathrm{typ}}$, and one must ask which estimate the
averaged recall follows. The signal is carried by a single stored pattern, whose
collective cavity exponent $\hat E_1$ is one random variable per disorder
realisation, fluctuating by $\mathcal{O}(\sqrt N)$ in the
exponent~\eqref{eq:Ehat-corrupted}. The retrieval statistic reported by both
theory and experiment is an average over independent dataset re-draws, and that
average is dominated by the favourable tail of $\hat E_1$, precisely the
configurations the annealed mean $\mu_1(r)$ weights. It is therefore the annealed
mean, not the smaller typical value, that the averaged recall tracks; the typical
prescription would become operative only in a strict $N\to\infty$ limit taken at
fixed sub-exponential load, a regime incompatible with the exponential storage
studied here. Sections~\ref{sec:hmm}--\ref{sec:clinc} confirm this directly:
Figure~\ref{fig:basins}(a) shows the measured transition sitting on the annealed
curve, near $r_c^{\mathrm{ann}}$ and far from $r_c^{\mathrm{typ}}$. Consistently
with Remark~\ref{rem:annealed-noise}, the residual gap is an overshoot
toward even larger basins --the per-pattern noise being itself an annealed
overestimate-- so both corrections point the same way: the network is at least as
tolerant as the annealed signal-to-noise closed form, and never as pessimistic as
the typical bound. At $r=1$ the two coincide and Section~\ref{sec:storage} is
recovered, $\varepsilon_L(1)=\varepsilon_L^{\mathrm{typ}}(1)=\rho_L$.

\begin{remark}[Symmetric corruption versus the clamped-cue protocol]
\label{rem:clamped-cue}
The closed forms above corrupt every layer by the same amount, $r^{a}\equiv r$.
The hetero-associative recall protocol of Sections~\ref{sec:hmm}--\ref{sec:clinc}
instead clamps the $L-1$ cue layers at their exact value and lets only the
target layer evolve from an uncorrupted start: an asymmetric limit of the family
above. The tilted saddle of \ref{app:corrupted} carries one bias field
$a_{r^{a}}=\tanh^{-1}(r^{a})$ per layer, one for each of the $L$ corruption
levels; clamping the cue amounts to sending $a_{r^{a}}\to\infty$ (i.e.\
$m^{a}\to1$) in the $L-1$ of them that belong to the cue layers. The coupled
system then collapses to the single scalar equation
$m_T=\tanh\bigl[(L-1)+a_{r_T}\bigr]$ for the target layer, in place of the
symmetric~\eqref{eq:saddle-corrupted-main}. We use the symmetric curve as the
reference in Figure~\ref{fig:vdjdb-results}(b) of Section~\ref{sec:vdjdb}
because the two agree to the accuracy of that plot, not because they are the
same object.
\end{remark}

\begin{table}[t!]
\centering
\small
\begin{tabular}{c|ccccc}
\toprule
$L$ & $\rho_L$ & $r_c$ (ann.) & $d_c$ (ann.) & $r_c^{\mathrm{typ}}$ & $d_c^{\mathrm{typ}}$ \\
\midrule
2  & 1.3470 & 0.3195 & 0.3402 & 0.5714 & 0.2143 \\
3  & 2.0784 & 0.4032 & 0.2984 & 0.8085 & 0.0958 \\
4  & 2.7726 & 0.4128 & 0.2936 & 0.8769 & 0.0616 \\
5  & 3.4657 & 0.4140 & 0.2930 & 0.9092 & 0.0454 \\
10 & 6.9315 & 0.4142 & 0.2929 & 0.9607 & 0.0196 \\
\bottomrule
\end{tabular}
\caption{Basin-of-attraction thresholds: minimal initial overlap $r_c$ (maximal
Hamming radius $d_c=(1-r_c)/2$) compatible with an exponential storage capacity,
under the annealed and the typical signal estimates. As $L\to\infty$,
$r_c^{\mathrm{ann}}\to\sqrt2-1$ while $r_c^{\mathrm{typ}}\to1$.}
\label{tab:basins}
\end{table}

\begin{figure*}[t!]
  \centering
  \includegraphics[width=\linewidth]{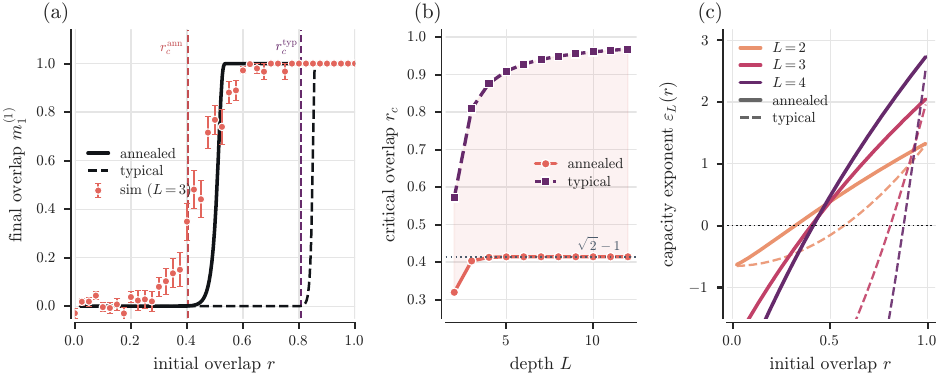}
  \caption{\textbf{Basins of attraction: larger basins cost rate.}
  \textbf{(a)} Basin recovery at $L=3$ ($N=30$, $P=3.5\times10^{4}$, i.i.d.\
  ensemble): one-step overlap after the dynamics versus the cue overlap $r$,
  against the corrupted-cue prediction~\eqref{eq:m1step-r} built on the
  annealed signal $\mu_1(r)$ (solid) and on the typical signal
  (dashed). The measured transition sits on the annealed curve, at
  $r\!\approx\!0.47$ next to $r_c^{\mathrm{ann}}$, and nowhere near the typical
  threshold $r_c^{\mathrm{typ}}\!=\!0.81$: the annealed branch is the physical
  one. If anything the data overshoot the annealed curve toward larger
  basins, the annealed per-pattern noise being a conservative overestimate
  (Remark~\ref{rem:annealed-noise}). The number of stored patterns is set from
  the annealed rule $P=e^{N\varepsilon_L(r^{\ast})}/\log N$ at $r^{\ast}=0.5$,
  exactly as in the single-layer reference~\cite{ALBANESE2026131223}.
  \textbf{(b)} The thresholds $r_c(L)$: the annealed branch saturates at
  $\sqrt{2}-1$ (dotted) while the typical branch climbs towards $1$, so basins
  shrink as the network widens.
  \textbf{(c)} Capacity exponent $\varepsilon_L(r)$ versus initial overlap $r$
  for $L=2,3,4$, in the annealed (solid) and typical (dashed) schemes; the
  zero-crossing is the recovery threshold $r_c$. Wider networks start higher
  (larger $\rho_L$) but cross zero at larger $r$.}
  \label{fig:basins}
\end{figure*}

\section{Structured data I: the Hidden Manifold Model}
\label{sec:hmm}

The theory of Sections~\ref{sec:model}--\ref{sec:basins} rests on one assumption:
the $L$ layer datasets are mutually independent Rademacher vectors. Real
hetero-associative data are nothing of the sort. Their layers are correlated,
because cue and target are different views of one shared cause, and their
patterns lie near a low-dimensional manifold rather than filling the hypercube.
A second feature is a modelling choice on our part rather than a property forced
by the data: we take the target to be a many-to-one (surjective) function of the
cue. Nothing requires associations to be many-to-one in general; but restricting
to that case is what supplies a single-valued, well-defined rule to store, and it
is also the structure of the two problems we go on to treat, where many distinct
cues legitimately share one target. Before touching real data we therefore ask a
controlled question, with a generator in which manifold dimension and
surjectivity are knobs we turn:

\begin{quote}
When the stored patterns are drawn near a low-dimensional manifold, with a
many-to-one target, does the exponential capacity survive --- and does
memorisation buy any generalisation to unseen points of the manifold?
\end{quote}

The generator is the Hidden Manifold Model~\cite{Goldt2020,Gerace2020}. Each of
the $P$ stored indices $\mu$ owns a latent code $z^{\mu}\sim\mathcal{N}(0,I_D)$ in
dimension $D$; the $L-1$ cue layers push it through a fixed random feature map,
one per layer, and threshold, $\xi^{\mu,a}=\sign(F^{a}z^{\mu}/\sqrt D)$, so that
cue layers of a given index share the latent $z^{\mu}$ --the correlation that
makes hetero-association possible. The target layer is set by a surjective
rule: the first $n_{\mathrm{bits}}$ latent signs select one of $K=2^{n_{\mathrm{bits}}}$
fixed prototypes $\rho_{k(z^{\mu})}$, so whole regions of latent space collapse to
a common target. The two counts are distinct and both essential: $P$ defines the load
(how many cues are stored), $K$ the number of target prototypes, and the map is surjective precisely because $P\gg K$, with
$P/K$ the mean number of cues per prototype. The single geometric control
parameter is the aspect ratio $\alpha_D=D/N\in(0,1]$:
 small $\alpha_D$ is a
tightly curved, strongly correlated manifold; $\alpha_D=1$ is the near-i.i.d.\
edge of the same pipeline. Full construction (the region map $k(z)$, the
prototypes, the symmetric-target control, and the held-out-region control) is
in \ref{app:hmm_protocol}.

A convention on the numbers, valid for this section and the two that follow: all
empirical quantities stemming from the numerical experiments are reported as
means $\pm$ one standard deviation over independent dataset re-draws, each
re-draw itself an average over evaluation trials, with the seed counts listed in
\ref{app:repro}.

\paragraph{The data are a sign-image of a geometry}
The first thing to verify is that the manifold is really there. The diagnostic is
the pairwise pattern overlap $m_{\mu\nu}=\frac1N\sum_{i}\xi_i^{\mu,a}\xi_i^{\nu,a}$,
the cosine between two stored codes in a layer. For sign patterns it concentrates
on the arcsine law $\E[m_{\mu\nu}]=(2/\pi)\arcsin\rho_z$, with
$\rho_z=\langle z^{\mu},z^{\nu}\rangle/(\|z^{\mu}\|\|z^{\nu}\|)$ the latent cosine
similarity: an exact, parameter-free curve (derived from Grothendieck's identity
in \ref{app:hmm_protocol}). Figure~\ref{fig:hmm-manifold}(a) shows the measured
overlaps lying on it, turning ``the data lie near a manifold'' into a verified
fact rather than a hope. The same panel reads off the price of structure: the
standard deviation of the cue overlap crosses over from the i.i.d.\ value
$1/\sqrt N$ at $\alpha_D=1$ to the much larger $(2/\pi)/\sqrt D$ as
$\alpha_D\to0$. A smaller manifold makes the stored patterns more
correlated (more overlapping than Rademacher, whenever $D<(2/\pi)^{2}N$), which is
exactly the extra noise that the clean theory does not carry, and exactly why
capacity should fall as the manifold shrinks.

\begin{figure*}[t!]
  \centering
  \includegraphics[width=\linewidth]{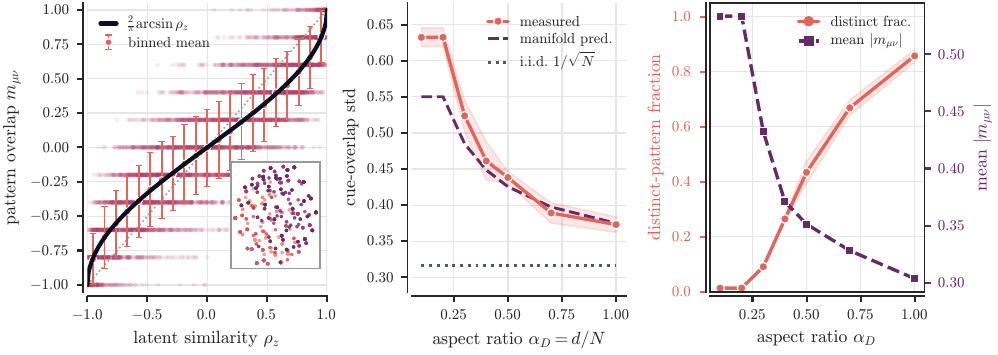}
  \caption{\textbf{The data manifold.}
  \textbf{(a)} Pattern overlap $m_{\mu\nu}$ against latent similarity $\rho_z$
  (density) with binned means (markers) lying on the arcsine law
  $\frac{2}{\pi}\arcsin\rho_z$ (black); inset, a 2D PCA embedding of the cue
  patterns coloured by target region, showing the surjective clustering.
  \textbf{(b)} The cue-overlap standard deviation crosses over from the i.i.d.\
  value $1/\sqrt N$ at $\alpha_D=1$ to a strongly correlated regime as
  $\alpha_D\to0$, tracking the manifold prediction.
  \textbf{(c)} The distinct-pattern fraction (left) collapses at small $\alpha_D$
  while the mean overlap (right) rises: apparent stability at small $d$ is pattern
  collision, not capacity.}
  \label{fig:hmm-manifold}
\end{figure*}

\paragraph{The exponential capacity survives, and fails gracefully}
Figure~\ref{fig:hmm-capacity}(a) scans the load $P$ at fixed layer size $N=10$. At low
load the i.i.d.\ ensemble retrieves perfectly and tracks the finite-$N$
theory~\eqref{eq:m1step} to within finite-size noise (with the transition sitting
at larger $P$ than the annealed estimate, as Remark~\ref{rem:annealed-noise}
anticipates); the manifold ensembles sit
slightly below, the correlation of Figure~\ref{fig:hmm-manifold}(b) acting as an
extra noise term on top of the $e^{-N\rho_L}$ floor, and the plateau drops as
$\alpha_D$ falls, the practical form of the ``capacity decreases with manifold
dimension'' statement. The high-load behaviour is more interesting, and at first
sight paradoxical: past capacity the i.i.d.\ overlap decays toward zero, while
the manifold overlap saturates on a positive floor. The floor is not
superior retrieval. It is the structural self-overlap
$\E|m_{\mu\nu}|\approx0.30$ of the manifold: once the network can no longer
resolve individual memories it relaxes to the manifold rather than to a
random spurious state, and keeps the residual overlap that any two manifold
points share. Read against the measured floor, the manifold curve confirms rather
than contradicts the theory, and confirms that failure is graceful.

Part of the early saturation at small $\alpha_D$ so is not a retrieval effect at all
but a counting one, and the inset of Figure~\ref{fig:hmm-capacity}(a) isolates
it. It tracks the collision rate $(P-P_{\mathrm{unique}})/P$: the fraction of the
$P$ drawn latents whose sign-image $\sign(F z)$ duplicates a pattern already in
the store. Because $z\mapsto\sign(Fz)$ can realise at most
$C(N,D)=2\sum_{k<D}\binom{N-1}{k}$ distinct codes (Cover's count for a central hyperplane arrangement, \ref{app:hmm_protocol}),
the codebook is finite and collisions are inevitable once $P$ becomes comparable
to $C(N,D)$: at $N=10$ this count is $92$, $764$ and $1024$ for
$\alpha_D=0.3,0.6,1.0$, and the measured rate crosses one half at
$P\approx62$, $325$ and $605$ respectively, in the same order and of the same
magnitude. The consequence for panel~(a) is a bookkeeping one that we
insist on because it is easy to misread in the opposite direction: beyond that
load the network is no longer storing new memories, only additional copies of the
few it can address, so the overlap plateau overstates rather than understates the
number of distinct memories held. 
The distinct-pattern fraction
(Figure~\ref{fig:hmm-manifold}(c)) must be read together with the stability score;
the honest ``capacity falls with $D$'' statement is the one taken from the
matched-load transition of Figure~\ref{fig:hmm-capacity}(a), not from raw
stability.

\begin{figure*}[t!]
  \centering
  \includegraphics[width=\linewidth]{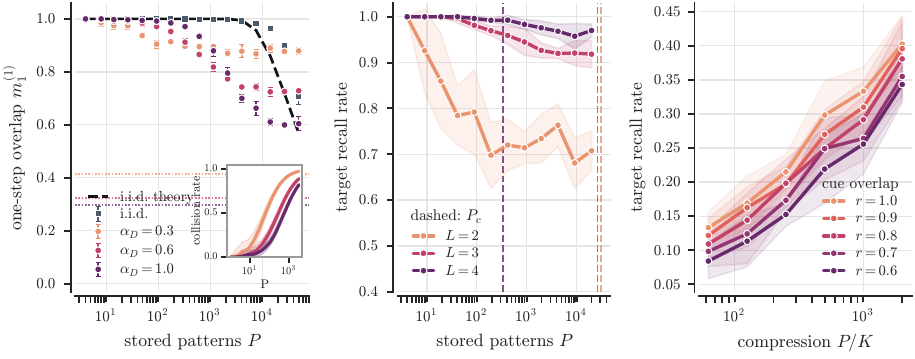}
  \caption{\textbf{Capacity, width and surjective compression.}
  \textbf{(a)} One-step overlap versus load $P$ ($N=10$) for the i.i.d.\ ensemble
  and Hidden-Manifold ensembles at several $\alpha_D$: capacity falls as the
  manifold shrinks, and past capacity the manifold state relaxes onto its
  structural floor (dotted) rather than a spurious pattern. Inset, same abscissa
  and colours: the collision rate $(P-P_{\mathrm{unique}})/P$, the fraction of
  drawn latents that land on a pattern already stored. It leaves zero at a load
  of the order of Cover's count $C(N,D)$, earliest for the smallest
  $\alpha_D$, and the onset of each curve's rise coincides with the onset of
  the corresponding overlap plateau.
  \textbf{(b)} Target-recall rate versus $P$ for width $L=2,3,4$ at matched layer
  size; the transition shifts to larger $P$ with $L$ (dashed: predicted $P_c$), so
  more cue layers tolerate exponentially more memories.
  \textbf{(c)} Recall versus surjective compression $P/K$ at several cue-corruption
  levels $r$: prototypes backed by more cue patterns have deeper basins.}
  \label{fig:hmm-capacity}
\end{figure*}

\paragraph{Basins, and which of the two thresholds is realised}
At a matched sub-critical load the basins tell the complementary story
(Figure~\ref{fig:basins}(a)). We sweep the cue overlap $r$ at $L=3$, fixing the
load from the annealed rule $P=e^{N\varepsilon_L(r^{\ast})}/\log N$ so that the
annealed threshold lands at a chosen $r^{\ast}$, the same operational recipe used
in the single-layer reference~\cite{ALBANESE2026131223}. The measured transition
tracks the annealed one-step curve~\eqref{eq:m1step-r} and sits at
$r\approx0.47$, right next to $r_c^{\mathrm{ann}}$, and nowhere near the
typical value $r_c^{\mathrm{typ}}=0.81$ of Table~\ref{tab:basins}: the favourable
corruption masks that dominate the annealed signal are exactly the ones sampled
by the disorder average, so it is the annealed --the optimistic-- branch that is
physical. This settles the question left open in Section~\ref{sec:basins}, in the
opposite sense to a naive ``rare events do not matter'' expectation: because the
signal rides on a single pattern, its rare-but-favourable fluctuations are not
averaged away. If anything the data overshoot the annealed prediction toward
still larger basins, the per-pattern noise being itself an annealed overestimate
(Remark~\ref{rem:annealed-noise}), so both the corrupted-cue signal and the noise
push the same way. The i.i.d.\ basin is the cleaner and deeper one; the manifold
basin rides the structural self-overlap floor at strong corruption, exactly as in
panel~(a).

\paragraph{Width helps recall}
Because the rate $\rho_L$ grows with the width $L$, more cue layers should tolerate more
stored patterns. Figure~\ref{fig:hmm-capacity}(b) confirms it directly: at matched
layer size the target-recall curves shift to larger $P$ as $L$ goes from $2$ to $4$,
each transition sitting near its predicted $P_c\sim e^{N\rho_L}$. The case that
matters for the receptor data of Section~\ref{sec:vdjdb} --two cue chains
driving a third layer, $L=3$-- is the middle curve, and having both cues rather
than one is worth an order of magnitude in tolerated load.

\paragraph{The headline: memorisation without much generalisation}
Beyond memorisation there is a second question, and it is the one the model was
not built to answer: whether a cue drawn from a never-stored latent is routed to
the correct target prototype. We hold two of eight latent regions out
by construction and measure four rates: memorisation (stored cues),
generalisation (fresh cues from seen regions), a novel-region control (fresh cues
from held-out regions), and chance $1/n_{\mathrm{seen}}=0.167$. The three
comparisons together are what make the reading unambiguous
(Figure~\ref{fig:hmm-generalization}). Memorisation is high. Generalisation sits
clearly above chance, so the manifold structure is being used, and climbs
with coverage $P/n_{\mathrm{seen}}$, from $\approx0.21$ at the sparsest sampling
toward $\approx0.39$, approaching memorisation only in the densely sampled limit,
where the stored cues tile the manifold so finely that the manifold itself acts as
the attractor. The novel-region control fixes the interpretation: its rate never
leaves chance, across the $12$ loads scanned it beats chance in none (Wilcoxon
signed-rank on the per-load means, two-sided $p\approx4.9\times10^{-4}$), whereas
generalisation beats chance in all $12$ ($p\approx2.4\times10^{-4}$) and beats the
novel-region control at every load (sign test, $p\approx2.4\times10^{-4}$). The
generalisation signal is therefore real rather than an encoding artefact, and at
the same time modest: an exponential ability to memorise does not, by itself,
confer a comparable ability to classify unseen inputs, short of sampling the
manifold densely enough that classification collapses back to memory. This is the
central empirical fact of the paper, and the next section shows it survives
contact with real data.

\begin{figure*}[t!]
  \centering
  \includegraphics[width=\linewidth]{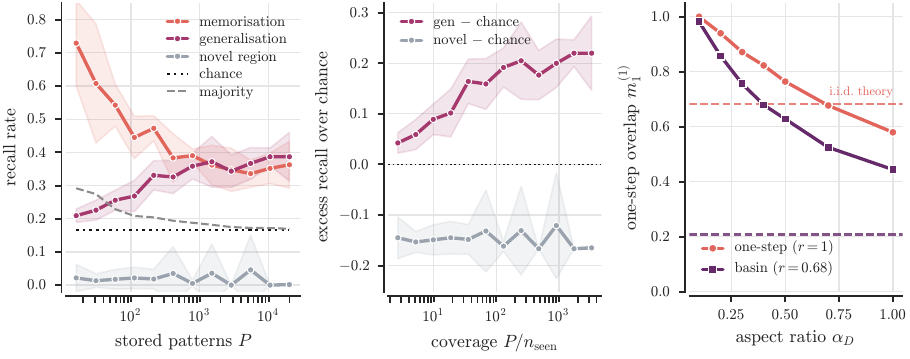}
  \caption{\textbf{Memorisation without generalisation.}
  \textbf{(a)} Memorisation, generalisation and novel-region recall versus load
  ($K=8$ regions, two held out), with chance (dotted) and majority (dashed).
  \textbf{(b)} Excess over chance versus coverage $P/n_{\mathrm{seen}}$:
  generalisation grows and stays well above the novel-region control, which never
  leaves chance: the signal is real but far below memorisation.
  \textbf{(c)} As $\alpha_D\to1$ the manifold retrieval descends onto the i.i.d.\
  theory (dashed lines); at small $\alpha_D$ collisions inflate the apparent
  overlap above it.}
  \label{fig:hmm-generalization}
\end{figure*}

\section{Structured data II: real T-cell receptors}
\label{sec:vdjdb}

We now replace the synthetic manifold with immunological repertoires. A T-cell
receptor recognises an antigen through the paired hypervariable loops of its two
chains, the $\alpha$- and $\beta$-CDR3; the antigen is a short peptide, the
epitope. The map from receptor to epitope is exactly the object the model is
built for: a hetero-association from two cue layers to a target, and a strongly
surjective one, since many unrelated receptors converge on the same epitope. We
take the triples $(\alpha\text{-CDR3},\ \beta\text{-CDR3},\ \text{epitope})$
from VDJdb~\cite{Shugay2018,Bagaev2020}, assign them to the three layers
$a\in\{1,2,3\}$, and run the same battery as in Section~\ref{sec:hmm}.
Table~\ref{tab:mapping} fixes the dictionary between the model and the biology.

\begin{table}[t!]
\centering
\small
\renewcommand{\arraystretch}{1.3}
\begin{tabular}{@{}p{0.44\linewidth}p{0.48\linewidth}@{}}
\toprule
\textbf{Model (theory)} & \textbf{Biology / real data} \\
\midrule
Layers $L$ & the two receptor chains ($\alpha$, $\beta$) and the target epitope ($L=3$) \\
Independent Rademacher patterns & Atchley factors passed through a SimHash \\
Surjective many-to-one map & many unrelated receptors converging on one epitope \\
\bottomrule
\end{tabular}
\caption{The model-to-biology dictionary for the VDJdb experiments.}
\label{tab:mapping}
\end{table}

\paragraph{Making the problem well posed}
Two things must be fixed before the network sees anything, and both are
prescribed by the theory. First, by Remark~\ref{rem:surjective} the stored map
must be a single-valued function: a receptor mapped to two different epitopes is
not a hetero-association but a contradiction, and~\eqref{eq:majority} says what
the network would return in its place. Starting from $137{,}484$ raw records we
apply a lean biological clean (human, curation score $\ge1$, paired
\texttt{complex.id}, valid sequences) and then a \emph{function filter} keeping
only receptors mapped to a single epitope, which leaves $1{,}052$ clean triples
over $220$ epitopes, strongly surjective, up to $146$ receptors per epitope
with $120$ singletons (Figure~\ref{fig:vdjdb-data}(a,b)): the immunological
image of the convergent recognition the model is meant to store. Second, the
sequences must become $\{-1,+1\}^{N}$ patterns without smuggling in a learned
representation. We use Atchley factors~\cite{Atchley2005}, five standardised
biophysical numbers per residue placed positionally onto a fixed-length vector,
fed to a locality-sensitive SimHash~\cite{Charikar2002}: a fixed, deterministic
standardise$\to$PCA-whiten$\to$Gaussian$\to$sign map, fitted per layer on the
training split, whose bits are balanced and near-orthogonal and whose Hamming
overlap obeys the same arcsine law $\E[m_{\mu\nu}]=1-\tfrac2\pi\arccos c$ as the
manifold model, in the biophysical cosine $c$
(Figure~\ref{fig:vdjdb-data}(c,d); full funnel and encoder in
\ref{app:vdjdb_protocol}). The outcome is quantitatively Rademacher-like (per-bit balance $\approx0.03$, mean absolute pattern overlap $\approx0.10$ on
the receptor layers) close enough to the theory's ideal that the i.i.d.\
predictions are usable as a yardstick.

\begin{figure*}[t!]
  \centering
  \includegraphics[width=\linewidth]{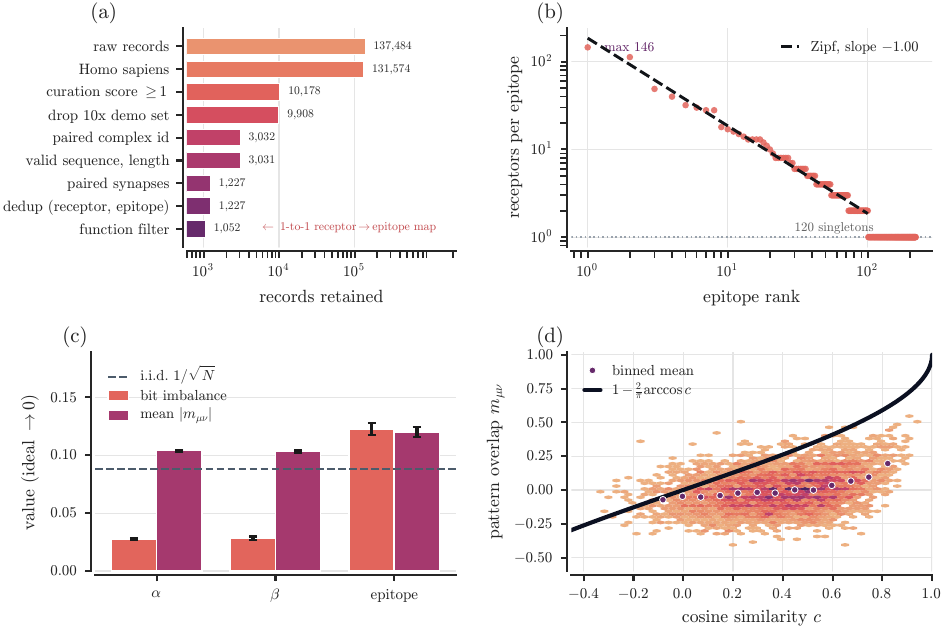}
  \caption{\textbf{From database to binary patterns.}
  \textbf{(a)} The cleaning funnel: $137{,}484$ raw records reduce to $1{,}052$
  clean triples once the function filter enforces a single-valued
  receptor$\to$epitope map.
  \textbf{(b)} Epitope cluster sizes (rank--frequency): strongly surjective, up
  to $146$ receptors per epitope, with $120$ singletons.
  \textbf{(c)} Near-Rademacher encoding: per-bit imbalance and mean pattern
  overlap by layer, against the i.i.d.\ value $1/\sqrt N$; the receptor layers are
  close to ideal, the small epitope alphabet less so.
  \textbf{(d)} The SimHash law: the binned-mean pattern overlap (markers) follows
  the arcsine shape $m_{\mu\nu}=1-\tfrac{2}{\pi}\arccos c$ (black) in the cosine
  similarity $c$ of the Atchley feature vectors; individual pairs (density)
  scatter broadly around it, as expected for a finite hash.}
  \label{fig:vdjdb-data}
\end{figure*}

\paragraph{Real data behaves like the theory}
The most striking result requires no fitting. We corrupt a stored receptor and
watch the dynamics recover it, sweeping the initial overlap $r$
(Figure~\ref{fig:vdjdb-results}(b)). The real-data basin curve is essentially
indistinguishable from the i.i.d.\ one, and both sit on the finite-$N$
prediction~\eqref{eq:m1step-r}: recovery is total for $r\gtrsim0.4$ and
collapses through a sharp transition near $r\approx0.2$--$0.3$. Correlated,
surjective, biologically generated patterns thus have basins of attraction that
the idealised Rademacher theory describes to within finite-size noise: the
independence assumption of Section~\ref{sec:model}, violated in every literal
respect by these data, is benign for retrieval.

\paragraph{Two chains recall the antigen, essentially without error}
The biologically central task is $(\alpha,\beta)\to\text{epitope}$: given both
receptor chains, name the antigen. Figure~\ref{fig:vdjdb-results}(a) reports a
recall rate of $0.9994\pm0.001$: perfect memory up to statistical noise. A single
chain is not enough ($\alpha\to\text{epitope}$ and $\beta\to\text{epitope}$
recover only $0.79$ and $0.83$) so the two chains are genuinely
complementary, each supplying information the other lacks, and only their
conjunction pins the epitope. The reverse, over-determined direction
$(\beta,\text{epitope})\to\alpha$ recalls the $\alpha$ chain perfectly
($1.00$), as an invertible cue should. As a stored associative memory of known
receptor--epitope \emph{bindings} (experimentally confirmed pairs, each one a
receptor observed to physically recognise and lock onto that particular peptide) the network is essentially flawless at $L=3$ on real data.

\begin{figure*}[t!]
  \centering
  \includegraphics[width=\linewidth]{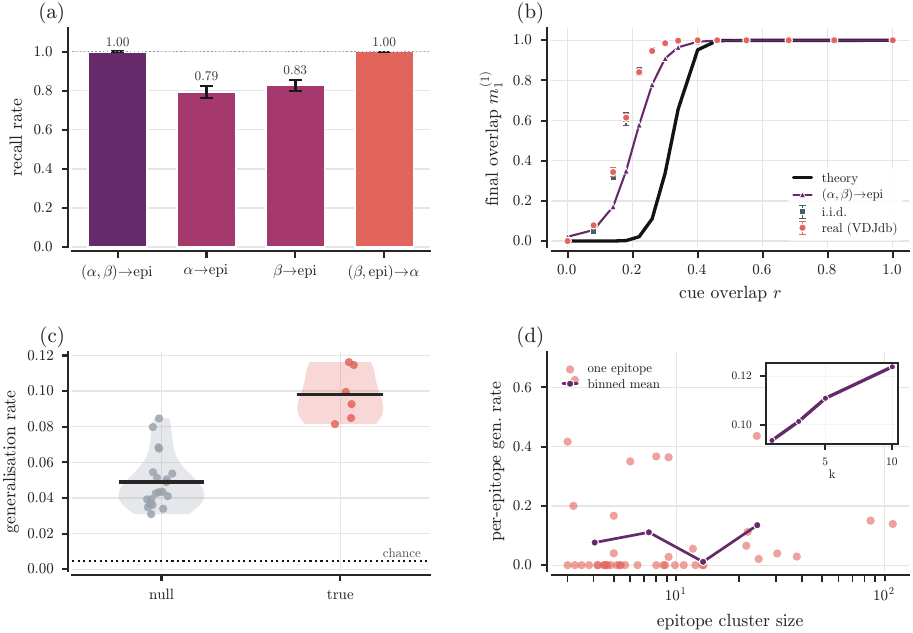}
  \caption{\textbf{Near-perfect memory, modest generalisation on real data.}
  \textbf{(a)} Recall for the biological tasks: both chains name the epitope
  essentially without error, each single chain only partially, and the
  over-determined $(\beta,\text{epitope})\to\alpha$ perfectly.
  \textbf{(b)} Basins of attraction: the real-data recovery curve coincides with
  the i.i.d.\ one and with the finite-$N$ theory (black); the hetero-associative
  cue $(\alpha,\beta)\to\text{epitope}$ needs a larger overlap to lock on.
  \textbf{(c)} Generalisation to unseen receptors versus a label-permutation null:
  the true signal ($\approx0.10$) is about twice the null and well above chance,
  real but modest.
  \textbf{(d)} Per-epitope generalisation versus cluster size. Each point is one
  epitope: its abscissa is the number of receptors carrying it, its ordinate the
  fraction of that epitope's held-out receptors routed back to it. Epitopes backed
  by more stored receptors generalise better, because a denser cluster tiles more
  of the receptor manifold around the shared target. Inset: top-$k$ accuracy of
  retrieving the correct epitope for a new receptor.}
  \label{fig:vdjdb-results}
\end{figure*}

\paragraph{The headline again, sharper: memory yes, prediction barely}
The database records a finite list of confirmed bindings; the biological question
is whether that list determines the rest. Does storing the known bindings let the
network predict the epitope of a receptor appearing in no stored pair, that
is, name the peptide a never-seen receptor would bind? We hold out a quarter of
the receptors of each epitope and
test recall on them. Memorisation of the training receptors is near-perfect
($0.997$); generalisation to held-out receptors is $0.11$--$0.14$, rising with
the number of receptors per epitope, against a chance level of $0.005$: roughly
twenty times chance, and about twice a label-permutation null (true
generalisation $\approx0.10$ versus null $\approx0.048$;
Figure~\ref{fig:vdjdb-results}(c)), a gap that a Mann--Whitney test on the $6$
true against the $20$ permuted replicates confirms is not incidental (one-sided
$p\approx9\times10^{-6}$, equivalently $z\approx3.2$). The signal is therefore
real: the encoding does place receptors of the same epitope in neighbouring
regions of pattern space, and the network exploits it. But it is small. Top-1
retrieval of the correct epitope for a new receptor is $0.09$
(Figure~\ref{fig:vdjdb-results}(d)), an order of magnitude below the network's
own memorisation. The pattern is the one of Section~\ref{sec:hmm}, now on real
data: an exponential, near-perfect associative memory whose classification of
novel inputs is real but modest.

This is not a verdict on the model. Predicting TCR specificity from sequence
alone is a hard problem in general, and the next section shows that the same
network, storing the same way, generalises five times better on a domain whose
encoder co-locates same-target cues more tightly. What the VDJdb experiment pins
down is what an exponential hetero-associative memory does and does not deliver
\emph{on this encoding}.

\section{Structured data III: natural language, and the universality of the
mechanism}
\label{sec:clinc}

\begin{figure*}[t!]
  \centering
  \includegraphics[width=\linewidth]{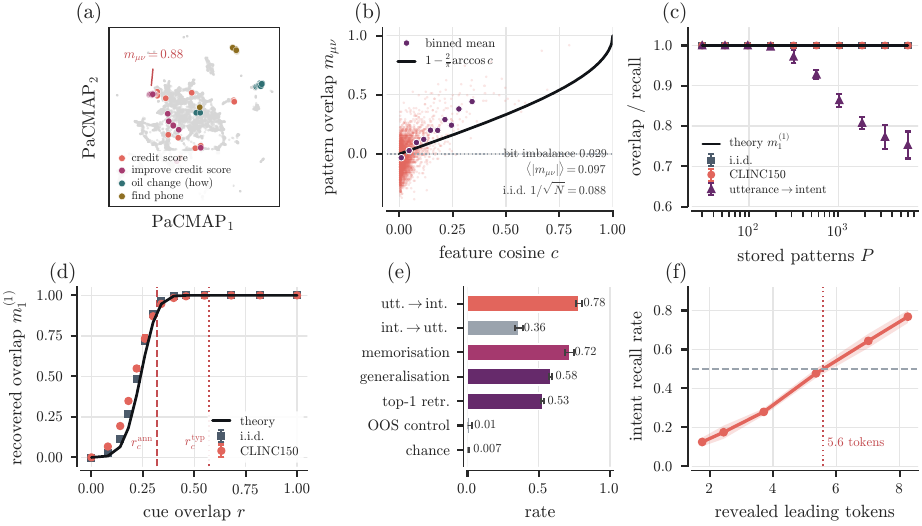}
  \caption{\textbf{The same network on natural language ($L=2$, $N=128$,
  CLINC150).}
  \textbf{(a)} PaCMAP embedding of the encoded utterances (grey: all $150$
  intents). Four intents are highlighted: the semantically affine pair
  \emph{credit score} / \emph{improve credit score}, whose clusters coincide and
  whose target codes overlap at $m_{\mu\nu}=0.88$, and two unrelated controls at
  $|m_{\mu\nu}|\le0.06$ from them.
  \textbf{(b)} The SimHash law: binary pattern overlap against feature cosine,
  with the arcsine prediction (black); the encoding is near-Rademacher
  (annotated values).
  \textbf{(c)} Capacity scan: the one-step overlap sits on the theoretical
  plateau for real and i.i.d.\ patterns alike over the whole accessible range of
  $P$ (the predicted capacity being $e^{172}$), while the utterance$\to$intent
  routing decays under interference between cues sharing a target.
  \textbf{(d)} Basins: real and i.i.d.\ recovery curves coincide and follow the
  finite-$N$ prediction~\eqref{eq:m1step-r}; the transition sits below
  $r_c^{\mathrm{ann}}$ and far from $r_c^{\mathrm{typ}}$.
  \textbf{(e)} Recall on the language tasks: the forward (surjective) direction,
  the ill-posed reverse one, memorisation, generalisation to unseen utterances,
  top-1 retrieval, the out-of-scope control and chance.
  \textbf{(f)} Cue content: intent recall against the number of revealed leading
  tokens; half the intents are recovered from $5.6$ words.}
  \label{fig:clinc}
\end{figure*}

Two datasets from a biological or geometric generator leave one question open:
how much of what we have measured belongs to the network and how much to the
domain? The way to separate them is to change the domain as far as possible while
changing nothing else.

CLINC150~\cite{Larson2019} labels short user utterances with one of $150$
intents. It is the linguistic image of the receptor problem, a single-valued,
strongly surjective, many-to-one map, utterance$\to$intent, and shares with it
nothing else: no alphabet, no metric, no generative model, no notion of distance
in common, only the structure Remark~\ref{rem:surjective} demands. The same
function filter removes just $4$ ambiguous utterances, leaving $22{,}491$
records over $150$ intents with a
compression $P/K\approx150$ cues per target, against $\approx5$ for VDJdb. We
take $L=2$ (utterance, intent), which makes the closed-form numbers of
Sections~\ref{sec:storage}--\ref{sec:basins} directly usable ($\rho_2=1.347$,
annealed $r_c=0.3195$, typical $r_c=0.5714$, no $L$-dependent inversion)
and binarise utterances with a word- and character-level TF--IDF map, a fixed,
learning-free representation playing the role the Atchley factors played on
receptors, passed through the {same} PCA-whitening-plus-SimHash binariser
(\ref{app:clinc}).

\paragraph{The memory side is the same theory, verbatim}
Figure~\ref{fig:clinc} collects the battery. The encoding is again
near-Rademacher: per-bit imbalance $0.029$ and mean overlap
$\langle|m_{\mu\nu}|\rangle=0.097$ against the i.i.d.\ value
$1/\sqrt N=0.088$ at $N=128$, with the pattern overlap following the arcsine law
$1-\tfrac2\pi\arccos c$ in the feature cosine, panel~(b), exactly as in
Figure~\ref{fig:vdjdb-data}(d). The capacity scan, panel~(c), sits on the
theoretical plateau over the whole accessible range: at $N=128$, $L=2$ the
predicted capacity is $P_c\sim e^{N\rho_2}=e^{172}$, and the one-step overlap
duly remains $1.000$ to four digits for real and i.i.d.\ patterns alike from
$P=30$ to $P=6000$. What does move is the hetero-associative routing, decaying
from $1.00$ to $0.75$ over the same span: interference among the growing number
of cues that share a target, not loss of memory. The basins, panel~(d),
reproduce the VDJdb finding on data that have nothing to do with it: real and
i.i.d.\ curves lie on top of each other and on~\eqref{eq:m1step-r}, with the
transition at $r\approx0.21$--$0.24$, below $r_c^{\mathrm{ann}}$ and far from
the typical threshold, once more overshooting the annealed prediction towards
larger basins in the direction Remark~\ref{rem:annealed-noise} anticipates.

\paragraph{The geometry the encoder builds}
Panel~(a) makes visible what the arcsine law states. In a PaCMAP embedding of
the encoded utterances each intent is a tight cluster (median radius $0.3$,
against a median inter-centroid distance of $10.3$); the semantically affine
pair \emph{credit score} / \emph{improve credit score}, the most similar of the
$11{,}175$, falls on top of itself (centroids $1.9$ apart, target codes
overlapping at $m_{\mu\nu}=0.88$) while two unrelated intents sit at
$|m_{\mu\nu}|\le0.06$ from them and from each other. Affine targets share a
region of pattern space and unrelated ones do not; it is this co-location, not
the storage rule, that generalisation feeds on.

\paragraph{What changes, and what does not}
Panel~(e) is the comparison the experiment was built for. The
hetero-associative task utterance$\to$intent recalls at $0.780\pm0.022$;
memorisation of stored utterances is $0.717\pm0.030$; generalisation to unseen
utterances of a seen intent is $0.584\pm0.009$, against a chance level of
$1/150=0.0067$, an out-of-scope control at $0.014$ and a label-permutation
null of $0.013\pm0.002$ over $20$ permutations ($z\approx3.0\times10^{2}$). Top-1
retrieval of the correct intent for a fresh utterance is $0.527\pm0.009$, rising
to $0.649$ at top-10; on receptors the corresponding numbers were $0.11$ and
$0.09$. Storage rule, dynamics and closed-form basins are identical in the two
cases; only the encoder differs, and generalisation differs with it by a factor
of five.

One ablation settles the attribution beyond argument. Replacing the TF--IDF
features by a text-blind random map (a deterministic hash of each utterance to
a fixed Gaussian vector, everything else unchanged, so that lexically similar
utterances receive unrelated codes) leaves memorisation at $0.349\pm0.023$
and sends generalisation to $0.008\pm0.002$, exactly the chance level
(\ref{app:clinc}). The network still stores; it simply has nothing to
interpolate between. Widening the representation, conversely, buys
generalisation monotonically: $0.35,0.50,0.58,0.63$ at $N=32,64,128,256$,
tracking the fall of the mean pattern overlap towards its i.i.d.\ floor
$1/\sqrt N$. Generalisation is a property of the map into pattern space; the
storage rule contributes the memory, and only the memory.

The deliberately ill-posed direction says the same thing from the other end.
Cueing with the intent and asking for an utterance recalls at $0.362$: the
reverse of a surjection is not a function, and Remark~\ref{rem:surjective}
predicts that the target layer returns the componentwise majority of the
$q\approx20$ utterances sharing that intent at $P=3000$, whose overlap with any
one of them would be $\sqrt{2/\pi q}\approx0.18$ for independent codes. The
measured overlap is $0.414\pm0.017$, more than twice as large, because
paraphrases of one intent are \emph{not} independent: the encoder places them
close, so their majority stays close to each. The non-invertibility is exactly
as severe as the encoder's clustering leaves it.

\paragraph{How much cue is enough}
Panel~(f) puts the basin question in units a reader of the data can check: we
reveal only the leading words of the utterance and re-encode, instead of
flipping random bits. Recall grows smoothly from $0.125$ at $1.8$ revealed
tokens to $0.768$ at the full $8.3$, crossing one half at $5.6$: two thirds
of a sentence suffice to fall inside a basin, and the smoothness of the curve is
the linguistic image of the graceful degradation predicted under bit
corruption.

\section{Discussion}
\label{sec:discussion}

We set out to make hetero-association a first-class citizen of the exponential
associative-memory family, and to find out what such a network does on real
data. The theory answers the first question cleanly. An energy built from the
product of per-layer Mattis overlaps has the perfectly aligned
hetero-associative state as a fixed
point of its zero-temperature dynamics up to a load $P_c\sim e^{N\rho_L}$, with a
rate $\rho_L=L[(L-1)-\phi_L(x^\ast)]$ that we compute in closed form and that
grows like $L\log2$: each further layer buys exponentially more capacity. The one structural
novelty relative to the auto-associative case -that the per-pattern noise no
longer factorises over sites- is resolved by a large-deviation saddle point on
the symmetric ray of layer magnetisations, and the same machinery, tilted, yields
the basins: enlarging them lowers the rate but never removes its exponential
character, with an explicit annealed/typical dichotomy that the simulations then
settle in favour of the conservative branch. The same field computation delimits
what such an energy can hold at all: the stored rule must be a surjective
function of the cue (Remark~\ref{rem:surjective}), a cue with $q$ targets being
answered by their componentwise majority rather than by any of them, and a
target with no cue occupying capacity it can never be addressed by.

\paragraph{An excellent memory, and geometry-limited generalisation}
The experiments answer the second question in two parts. First, the memory side
is unambiguous and domain-independent: on the Hidden Manifold Model, on real
VDJdb triples and on natural-language intent data alike, the network stores an
exponential number of structured, surjective associations and recalls them
robustly, with basins that match the idealised i.i.d.\ theory even when the
patterns are correlated and biologically generated, and the two receptor chains
name their epitope essentially without error. Second, the classifier side,
routing an unseen cue to the right target, is real but bounded, and its size
is governed by the encoding rather than by the storage rule: modest on receptor
sequences, five times larger on paraphrased text, in both cases growing as the
manifold is sampled densely enough to approach the memory itself. This is
consistent with, and sharpens, recent findings that data geometry controls
whether modern Hopfield networks generalise~\cite{Negri2023,Kalaj2024}: in the
binary exponential model, memorisation is exponential and essentially free,
while generalisation is bought from the geometry of the representation.

The asymmetry is worth stating as a design fact rather than as a shortcoming. A
high-capacity associative memory of this family is built first of all to
remember, that is to serve as a massive content-addressable repository, and this
one does so exceptionally well, while extending the auto-associative scenario
to the more realistic hetero-associative one, where cue and target are distinct
objects. With $N$ binary neurons per layer there are $2^{N}$ configurations, and
a device that files an exponential fraction of them under their own content is
not the sort of object from which strong extrapolation should be expected. How
much it also generalises is a question about the encoder, not about the storage
rule.

\paragraph{Domain universality, and where generalisation comes from}
Nothing in the construction is specific to a domain: by
Remark~\ref{rem:surjective} any single-valued, surjective, many-to-one map is a
candidate, and nothing else is. Section~\ref{sec:clinc} pushes that claim as far
as we can: a natural-language corpus shares with a T-cell repertoire no
alphabet, no metric and no generative model, only the surjective structure, and
on it the same $L=2$ closed forms describe capacity and basins with not one
refitted constant. The classifier side, however, does not follow the mechanism.
Generalisation to unseen utterances of a seen intent reaches $0.58$ against a
memorisation of $0.72$; on receptors the corresponding figures are an order of
magnitude smaller. Paraphrases of one intent share words, so TF--IDF places them
close together and a single stored utterance already carves a basin that catches
many others; distinct receptors of one epitope share little primary sequence, so
the biophysical encoder scatters them and only dense sampling of the epitope's
receptor cloud builds a basin. What generalises is not the memory but the
encoder's ability to co-locate cues that share a target: a property of the
representation, not of the storage rule.

The surjective many-to-one principle travels further still: a continual,
privacy-preserving setting where clients contribute low-rank Hebbian updates
towards shared archetypes~\cite{AlessandrelliFederated2026}, and a
continuous-signal domain --multi-channel sleep polysomnography-- encoded through
the same PCA-whitening-plus-SimHash pipeline of \ref{app:vdjdb_protocol} into a
tri-layer hetero-associative memory~\cite{Ladiana2026FiniteSize}, both point the
same way: the encoder, not the domain, is what the exponential mechanism cares
about.

\paragraph{The single-layer sibling}
\ref{app:squared_overlap} analyses a $\mathbb{Z}_2$-symmetric,
squared-overlap variant at $L=1$, whose energy penalises $1-m_\mu^2$ rather than
$1-m_\mu$. It restores the global spin-flip symmetry the linear model breaks,
requires the same saddle-point treatment as the multilayer noise, and has a
storage rate $\rho\approx0.6928$ that sits closer to the absolute ceiling $\log2$
than the linear model's $\rho_{\mathrm{lin}}\approx0.6750$, a clean example of
how the shape of the exponent trades basin width against capacity. Its saddle
coincides with the $L=3$ specialisation of the multilayer analysis, a structural
correspondence we find suggestive.

\paragraph{Limitations and outlook}
The analysis is one-step and zero-temperature: it certifies fixed points and
one-update recovery, not the full multi-step relaxation or a finite-temperature
free-energy landscape, and the Gaussian signal-to-noise scheme is an
approximation whose corrections we have bounded but not resummed. The
independence of layer datasets is an idealisation; the experiments quantify its
violation but a theory of the correlated case --where a shared latent couples
the layers, as in the data-- remains open, and is the natural bridge to the
random-features and hidden-manifold Hopfield
programme~\cite{Goldt2020,Negri2023,Kalaj2024}. Two further questions we set
aside here are pursued elsewhere. First, nothing in Sections~\ref{sec:model}--\ref{sec:basins}
rules out a \emph{chimeric} fixed point in which distinct layers lock onto
different pattern indices; a companion architecture removes such states by
construction through a consensus mechanism over the shared
memory~\cite{AgliariThermodynamicBinding2026}, rather than by analysing their
basin of attraction within the present energy. Second, the noise floor
$K_Le^{-N\rho_L}$ that limits capacity here is a property of the raw Hebbian
kernel; an off-line \emph{dreaming} \cite{fachechi2019dreaming,agliari2019dreaming} step that reweights that kernel's
eigenmodes before retrieval improves memorisation and can disentangle mixture
states in the $L$-layer setting~\cite{BarraDreaming2026}, suggesting that the
rate $\rho_L$ of Section~\ref{sec:storage} is a property of the construction as
given, not a hard ceiling on what a hetero-associative energy of this family
can achieve. A further limitation, easily
missed because the storage rate $\rho_L$ reads as an unqualified gain, is that
it is bought at an equal and exponential price: one parallel sweep of
Algorithm~\ref{algo:hetero} costs $\Theta(NLP)$ in both time and memory
(Section~\ref{sec:dynamics}), so operating at any fixed fraction of $P_c\sim
e^{N\rho_L}$ costs $\Theta(NL\,e^{N\rho_L})$, the identical rate that
measures capacity. Every architecture in this exponential family inherits the
same trade, a useful reminder that the capacity theorem describes the fixed
points of an $N\to\infty$ limit rather than any regime a real machine can
occupy; \ref{app:cost} gives the arithmetic. Finally, the bounded-generalisation
finding invites a sharper question than we have answered: is there a principled
modification of the exponent, or of the encoding, that converts some of the
exponential memorisation budget into generalisation, without collapsing to the
trivial dense-sampling limit? We believe the tools assembled here, theory and
battery together, are the right place to ask it.

\section*{Data and code availability}
The simulation engine used to generate the Monte Carlo results reported
throughout this paper is available at
\url{https://github.com/andrea-ladiana/exponential-lam-engine/}.

\section*{Declaration of competing interest}
The authors declare that they have no known competing financial interests or
personal relationships that could have appeared to influence the work
reported in this paper.

\section*{Acknowledgements}\noindent
E.A. acknowledges Sapienza Università di Roma (RM124190CB1269EB)
\\
\noindent
A.B. acknowledges support from Sapienza University of Rome, Prot. n. RM12519999AB8CA9, {\em Neural Networks and Learning Machines: asymptotic behaviors on structured datasets}, and acknowledges support from INFN, Sezione di Roma 1.
\\
\noindent
E.A, A.B., A.Ladiana and A.Lepre are members of the GNFM group within INdAM which is acknowledged. 

\medskip
\noindent
The authors acknowledge the use of the Lagrange Multi-GPU Server at the Department of Mathematics, Sapienza University of Rome, for computational resources supporting this work.

\newgeometry{left=3cm,right=3cm,top=2.5cm,bottom=2cm}
\onecolumn
\appendix
\section{First and second moments at the recalled state}
\label{app:moments}

This appendix collects the detailed computations summarised in \S\ref{sec:stability}. Throughout we work at the trial state~\eqref{eq:trial} and we use the relabelling $u_j^{(\mu,c)}:=\xi_j^{\mu,c}\xi_j^{1,c}\in\{-1,+1\}$ for $\mu\neq 1$, which is, for fixed $\mu\neq 1$, an i.i.d.\ Rademacher family across both $j\neq i$ and $c=1,\dots,L$ by independence of the layer-specific datasets.

\subsection*{Self-signal ($\mu=1$)}
\label{app:mu1}

For $\mu=1$ at the trial state, $\hat m_1^{a}=(N-1)/N$ for every $a$, so
\begin{align}
  F_1^{a} &\;=\;\sum_{b\neq a}\hat m_1^{b} \;=\; (L-1)+\mathcal{O}(N^{-1}),\\
  \hat E_1 &\;=\; N\binom{L}{2}\Bigl(\tfrac{N-1}{N}\Bigr)^{\!2}-N\binom{L}{2}
            \;=\; -L(L-1)+\mathcal{O}(N^{-1}).
\end{align}
The on-site factor is deterministic, since $\xi_i^{1,c}\sigma_i^{c}=(\xi_i^{1,c})^{2}=1$,
\begin{equation}
  \Phi_1^{(\backslash a)}
  \;=\;
  \exp\!\Bigl(\sum_{c\neq a}F_1^{c}\Bigr)
  \;=\;
  e^{(L-1)^{2}}\bigl(1+\mathcal{O}(N^{-1})\bigr).
\end{equation}
Combining,
\begin{equation}
  X_i^{(1|a)}
  \;=\;
  (\xi_i^{1,a})^{2}\,e^{\hat E_1}\sinh(F_1^{a})\,\Phi_1^{(\backslash a)}
  \;=\;
  e^{-L(L-1)}\sinh(L-1)\,e^{(L-1)^{2}}+\mathcal{O}(N^{-1}).
\end{equation}
Using $(L-1)^{2}-L(L-1)=-(L-1)$,
\begin{equation}
  X_i^{(1|a)}
  \;=\;
  e^{-(L-1)}\sinh(L-1)+\mathcal{O}(N^{-1}).
  \label{eq:Xself}
\end{equation}

\subsection*{Direct verification by single flip}
\label{app:flip-check}

At the trial state, $E_1=N(\binom{L}{2}\cdot 1-\binom{L}{2})=0$. Flipping $\sigma_i^{a}\to-\xi_i^{1,a}$ shifts $m_1^{a}\to 1-2/N$ while leaving every other Mattis magnetisation unchanged, so the new exponent of the recalled term is
\begin{equation}
  E_1' \;=\;
  N\Bigl[\binom{L-1}{2}\cdot 1+(L-1)\bigl(1-\tfrac{2}{N}\bigr)\Bigr]-N\binom{L}{2}
  \;=\; -2(L-1).
\end{equation}
The contribution from the noise patterns is exponentially suppressed by~\eqref{eq:Xmu-sq-main}, so $\Delta E_i^{a}=N(1-e^{-2(L-1)})$. Equating with $2N\,h_i^{a}\,\xi_i^{1,a}$,
\begin{equation}
  \xi_i^{1,a}\,h_i^{a}
  \;=\;\frac{1-e^{-2(L-1)}}{2}
  \;=\; e^{-(L-1)}\sinh(L-1),
\end{equation}
in agreement with~\eqref{eq:Xself}.

\subsection*{Noise contribution to $\mu_1$}
\label{app:noise-mean}

For $\mu\neq 1$ the cavity magnetisations $\hat m_\mu^{c}=N^{-1}\sum_{j\neq i}u_j^{(\mu,c)}$ are sums of $(N-1)$ i.i.d.\ Rademacher variables, so
\begin{equation}
  \hat m_\mu^{c}=\mathcal{O}(N^{-1/2}),\quad
  F_\mu^{c}=\mathcal{O}(N^{-1/2}),\quad
  \hat E_\mu=-N\binom{L}{2}+\mathcal{O}(1).
\end{equation}
Using $\sinh(\xi_i^{\mu,a}F_\mu^{a})=\xi_i^{\mu,a}\sinh(F_\mu^{a})$,
\begin{equation}
  X_i^{(\mu|a)}
  \;=\;
  \xi_i^{1,a}\,\xi_i^{\mu,a}\,e^{\hat E_\mu}\sinh(F_\mu^{a})\,\Phi_\mu^{(\backslash a)}.
  \label{eq:Xmu}
\end{equation}
The on-site factor $\xi_i^{1,a}$ is independent of every other random object in~\eqref{eq:Xmu}: of $\xi_i^{\mu,a}$ by independence of layer-$a$ patterns at the same site; of $\hat m_\mu^{c}$ for any $c$ ad $\mu$, since  $\mu \not = 1$; and of $\Phi_\mu^{(\backslash a)}$, which involves $\xi_i^{\mu,c}$ only for $c\neq a$ and these are independent of $\xi_i^{\mu,a}$ because the datasets are independent across layers. Hence $\E[X_i^{(\mu|a)}]=\E[\xi_i^{1,a}]\cdot\E[\cdots]=0$, exactly. Combining with~\eqref{eq:Xself} produces~\eqref{eq:mu1}.

\subsection*{Off-diagonal contributions to $\mu_2$}
\label{app:offdiag}

\paragraph{Case $\mu=1$, $\nu\neq 1$}
$X_i^{(1|a)}$ is deterministic at leading order, so
\begin{equation}
  \E\bigl[X_i^{(1|a)}X_i^{(\nu|a)}\bigr]
  \;=\; X_i^{(1|a)}\cdot\E[X_i^{(\nu|a)}] \;=\; 0.
\end{equation}

\paragraph{Case $\mu\neq\nu$, both $\neq 1$}
Substituting~\eqref{eq:Xmu} for both factors,
\begin{equation}
  X_i^{(\mu|a)}X_i^{(\nu|a)}
  \;=\;
  (\xi_i^{1,a})^{2}\,\xi_i^{\mu,a}\xi_i^{\nu,a}\,
  e^{\hat E_\mu+\hat E_\nu}\sinh(F_\mu^{a})\sinh(F_\nu^{a})\,
  \Phi_\mu^{(\backslash a)}\Phi_\nu^{(\backslash a)}.
\end{equation}
The on-site product $\xi_i^{\mu,a}\xi_i^{\nu,a}$ is independent of every other quantity in the expression, by the same independence argument as in \ref{app:noise-mean}, and has zero expectation. Hence the off-diagonal contribution vanishes exactly.

\subsection*{On-site square average $\E[(\Phi_\mu^{(\backslash a)})^{2}]$}
\label{app:Phi-sq}

For $\mu\neq 1$,
\begin{equation}
  (\Phi_\mu^{(\backslash a)})^{2}
  \;=\;
  \exp\!\Bigl(2\sum_{c\neq a}\xi_i^{\mu,c}\sigma_i^{c}F_\mu^{c}\Bigr).
\end{equation}
At the trial state, $\sigma_i^{c}=\xi_i^{1,c}$, so $\xi_i^{\mu,c}\sigma_i^{c}=\xi_i^{\mu,c}\xi_i^{1,c}\in\{-1,+1\}$. Layer-$c$ on-site factors are mutually independent across $c\neq a$ because the datasets are layer-independent; averaging factor by factor and using $\E_{u}[e^{xu}]=\cosh x$ for $u\in\{-1,+1\}$ Rademacher,
\begin{equation}
  \E\bigl[(\Phi_\mu^{(\backslash a)})^{2}\bigr]
  \;=\;
  \prod_{c\neq a}\cosh\!\bigl(2F_\mu^{c}\bigr).
  \label{eq:Phi-sq}
\end{equation}

\subsection*{Berry--Esseen control of the Gaussian approximation}
\label{app:berry-esseen}

Section~\ref{sec:stability} approximates $X_i^{a}$ by a Gaussian via the Central
Limit Theorem applied to the sum of $P-1$ noise contributions
$X_i^{(\mu|a)}$, $\mu\neq1$, which are i.i.d.\ across $\mu$ at fixed $N$ because
the layer datasets are mutually independent (\S\ref{sec:model}). We make the
rate of this approximation explicit.

\paragraph{A deterministic bound on the noise terms}
\begin{lemma}
For every $N$, every site $(i,a)$, every $\mu\neq1$ and every realisation of the
disorder,
\begin{equation}
  \bigl|X_i^{(\mu|a)}\bigr| \;\le\; M_L \;:=\; \sinh(L-1)\,e^{(L-1)^{2}}.
  \label{eq:ML-bound}
\end{equation}
\end{lemma}
\begin{proof}
Write $X_i^{(\mu|a)}=\xi_i^{1,a}e^{\hat E_\mu}\sinh(F_\mu^{a})\Phi_\mu^{(\backslash a)}$.
Every cavity magnetisation is an empirical average of $\pm1$'s, so
$|\hat m_\mu^{b}|\le1$ and $|F_\mu^{c}|=|\sum_{d\neq c}\hat m_\mu^{d}|\le L-1$ for
every $c$; hence $|\sinh(F_\mu^{a})|\le\sinh(L-1)$ and, since
$\Phi_\mu^{(\backslash a)}=\exp\bigl(\sum_{c\neq a}u_i^{(\mu,c)}F_\mu^{c}\bigr)$
with $u_i^{(\mu,c)}=\pm1$, $|\log\Phi_\mu^{(\backslash a)}|\le\sum_{c\neq a}|F_\mu^{c}|\le(L-1)^{2}$,
so $\Phi_\mu^{(\backslash a)}\le e^{(L-1)^{2}}$. Finally
$\hat E_\mu=N[\Theta(\hat{\bm m}_\mu)-\binom{L}{2}]$ with
$\Theta(\bm m):=\sum_{c<d}m^{c}m^{d}$ affine in each coordinate separately on
$[-1,1]^{L}$; a function affine in each coordinate attains its extrema at a
vertex of the box, and among the $2^{L}$ vertices
$\bm\varepsilon\in\{-1,+1\}^{L}$,
$\Theta(\bm\varepsilon)=\tfrac12\bigl[(\sum_c\varepsilon_c)^{2}-L\bigr]\le\tfrac12(L^{2}-L)=\binom{L}{2}$,
with equality iff all $\varepsilon_c$ agree in sign. Hence $\Theta\le\binom{L}{2}$
everywhere on the box, so $\hat E_\mu\le0$ and $e^{\hat E_\mu}\le1$. Multiplying
the three bounds (with $|\xi_i^{1,a}|=1$) gives the claim.
\end{proof}

\paragraph{Finite third moment and the Berry--Esseen bound}
Since $|X_i^{(\mu|a)}|^{3}=|X_i^{(\mu|a)}|\cdot(X_i^{(\mu|a)})^{2}\le M_L\,(X_i^{(\mu|a)})^{2}$,
taking expectations and using~\eqref{eq:Xmu-sq-main},
\begin{equation}
  \rho_3 \;:=\; \E\bigl|X_i^{(\mu|a)}\bigr|^{3}
  \;\le\; M_L\,\E\bigl[(X_i^{(\mu|a)})^{2}\bigr]
  \;=\; M_L\,K_L\,e^{-N\rho_L}\bigl(1+o(1)\bigr).
  \label{eq:rho3-bound}
\end{equation}
The $P-1$ noise terms are i.i.d.\ at fixed $N$, mean zero, with variance
$\sigma_1^{2}:=K_L e^{-N\rho_L}(1+o(1))$ and third absolute moment bounded
by~\eqref{eq:rho3-bound}; the Berry--Esseen theorem (sharp constant
$C\le0.4748$) then gives
\begin{equation}
  \sup_x\Bigl|\P\Bigl(\tfrac{X_i^{a}-\mu_1}{\sigma}\le x\Bigr)-\Phi(x)\Bigr|
  \;\le\; C\,\frac{\rho_3}{\sigma_1^{3}\sqrt{P-1}}
  \;\le\; \frac{C\,M_L}{\sigma_1\sqrt{P-1}},
  \qquad \sigma^{2}=(P-1)\sigma_1^{2}.
  \label{eq:berry-esseen}
\end{equation}
The constant $C$ is universal: it depends on neither $N$, $L$, $P$ nor the law
of the noise terms, which is the precise content of ``uniform in $N$''
invoked in \S\ref{sec:stability}.

\begin{remark}[Scope of the bound]
The right-hand side of~\eqref{eq:berry-esseen} is informative ($\to0$) once
$P-1\gg M_L^{2}/(K_Le^{-N\rho_L})$, i.e.\ once $P$ exceeds a constant multiple
of $e^{N\rho_L}$ -- at or beyond the critical load itself. The elementary
bound~\eqref{eq:ML-bound} is therefore silent on the sub-critical retrieval
regime $P\ll P_c$ that the Monte Carlo battery of
Sections~\ref{sec:storage}--\ref{sec:clinc} actually probes: it certifies that
the CLT approximation is meaningful with a rate once the load is comparable
to or above capacity, not that the approximation is loose below it (the
extensive numerical agreement there is evidence of the latter, not a proof of
it). Sharpening~\eqref{eq:rho3-bound} into a bound that also vanishes deep in
the sub-critical regime would require the exact third moment, rather than the
crude bound $M_L$, through the same large-deviation machinery as
\ref{app:saddle}; we leave this refinement open.
\end{remark}

\section{Saddle-point analysis: rate, prefactor, fluctuations}
\label{app:saddle}

The diagonal noise contribution~\eqref{eq:noise-sq-main} requires the asymptotic evaluation, in $N$, of
\begin{equation}
  \mathcal{I}_N \;:=\; \E\!\Bigl[\,e^{2\hat E_\mu}\sinh^{2}(F_\mu^{a})\prod_{c\neq a}\cosh(2F_\mu^{c})\,\Bigr],
  \qquad \mu\neq 1,
  \label{eq:IN}
\end{equation}
the expectation running over the cavity magnetisations $\bm{\hat m}_\mu=(\hat m_\mu^{1},\dots,\hat m_\mu^{L})\in[-1,1]^{L}$. The route taken here --cavity fields, large deviations, Gaussian fluctuations around a saddle-- is standard statistical-mechanics toolkit~\cite{MPV1987,Coolen2005}, applied to a case (a product, rather than a sum, of per-layer overlaps in the exponent) where it does not reduce to a closed form. We carry out the analysis in three steps: large-deviation reduction to a variational problem ; identification of the symmetric saddle and proof of its uniqueness ; Gaussian fluctuation expansion and explicit evaluation of the prefactor $K_L$ .

\subsection*{Large-deviation reduction}
\label{app:saddle-LDP}

By Cramér's theorem, the empirical magnetisation of a single layer satisfies a large-deviation principle on $[-1,1]$ with rate function
\begin{equation}
  I_R(m) \;=\; \tfrac{1+m}{2}\log(1+m)+\tfrac{1-m}{2}\log(1-m),
  \label{eq:Cramer}
\end{equation}
the Cramér transform of the symmetric Bernoulli distribution. Independence across layers gives the joint rate
\begin{equation}
  I(\bm m) \;=\; \sum_{c=1}^{L} I_R(m^{c}),
  \qquad \bm m\in[-1,1]^{L}.
  \label{eq:joint-rate}
\end{equation}
Substituting $2\hat E_\mu = 2N\sum_{c<d}\hat m_\mu^{c}\hat m_\mu^{d}-2N\binom{L}{2}$ in~\eqref{eq:IN} and applying the following lemma:
\begin{lemma}[Varadhan's Lemma]
Let $\mathcal{X}$ be a regular topological space and let $(P_N)_{N \in \mathbb{N}}$ be a sequence of probability measures on $\mathcal{X}$ satisfying a Large Deviation Principle with rate function $I: \mathcal{X} \to [0, \infty]$. 

Furthermore, let $F: \mathcal{X} \to \mathbb{R}$ be a continuous function bounded from above. Then, the following limit holds:
$$
\lim_{N \to \infty} \frac{1}{N} \log \mathbb{E}_{P_N} \left[ e^{N F(X)} \right] = \sup_{x \in \mathcal{X}} \left[ F(x) - I(x) \right].
$$
\end{lemma}
The polynomial factors $\sinh^{2}$, $\cosh$ contribute only to the sub-exponential prefactor — the leading exponential rate is
\begin{equation}
  \frac{1}{N}\log\mathcal{I}_N
  \;\xrightarrow[N\to\infty]{}\;
  -2\binom{L}{2}+\sup_{\bm m\in[-1,1]^{L}}\Psi(\bm m),
  \qquad
  \Psi(\bm m) \;:=\; 2\sum_{c<d}m^{c}m^{d}-I(\bm m).
  \label{eq:rate-formula}
\end{equation}

\subsection*{Symmetric saddle and uniqueness}
\label{app:saddle-symmetric}
The functional $\Psi$ is invariant under permutations of the layers and even in each $m^{c}$. Its first-order conditions read $2\sum_{d\neq c}m^{d}=\tanh^{-1}(m^{c})$ for $c=1,\dots,L$. Subtracting two of these equations,
\begin{equation}
  2(m^{c'}-m^{c}) \;=\; \tanh^{-1}(m^{c})-\tanh^{-1}(m^{c'}),
\end{equation}
which, combined with the strict monotonicity of $\tanh^{-1}$ on $(-1,1)$, forces $m^{c}=m^{c'}$ for every pair $(c,c')$. Every interior stationary point therefore lies on the symmetric ray $m^{c}\equiv m^{*}$. On that ray, $\Psi$ reduces to
\begin{equation}
  \Lambda_L(m) \;:=\; 2\binom{L}{2}m^{2}-L\,I_R(m)
                \;=\; L(L-1)\,m^{2}-L\,I_R(m),
  \label{eq:Lambda}
\end{equation}
whose stationarity condition reads
\begin{equation}
  \tanh^{-1}(m^{*}) \;=\; 2(L-1)\,m^{*}
  \quad\Longleftrightarrow\quad
  m^{*} \;=\; \tanh\!\bigl(2(L-1)\,m^{*}\bigr).
  \label{eq:saddle-eq}
\end{equation}
For $L\ge 2$, equation~\eqref{eq:saddle-eq} has the trivial root $m^{*}=0$ (a local minimum of $\Lambda_L$) and two non-trivial symmetric roots $\pm m^{*}_{0}\neq 0$, which realise the supremum by continuity of $\Lambda_L$ on $[-1,1]$ and by $\Lambda_L(\pm 1)=-\infty$. We pick the positive root $m^{*}=m^{*}_{0}>0$.

With the change of variables $x=2(L-1)m$, equation~\eqref{eq:saddle-eq} becomes $\tanh(x^{*})=x^{*}/[2(L-1)]$, the stationarity condition of
\begin{equation}
  \phi_L(x) \;:=\; -\frac{x^{2}}{4(L-1)}+\log\cosh(x).
  \label{eq:phi}
\end{equation}
The supremum~\eqref{eq:rate-formula} on the symmetric ray equals $L\,\phi_L(x^{*})$. Indeed, set $g(m):=(L-1)m^{2}-I_R(m)$; with $x=\tanh^{-1}(m)$ one has $I_R(m)=mx-\log\cosh(x)$, whence
\begin{equation}
  g(\tanh(x))
  \;=\;
  (L-1)\tanh^{2}(x)-x\tanh(x)+\log\cosh(x).
\end{equation}
At the saddle, $\tanh(x^{*})=x^{*}/[2(L-1)]$ gives $(L-1)\tanh^{2}(x^{*})=x^{*\,2}/[4(L-1)]$ and $x^{*}\tanh(x^{*})=x^{*\,2}/[2(L-1)]$, so that
\begin{equation}
  g(\tanh(x^{*}))
  \;=\;
  -\frac{x^{*\,2}}{4(L-1)}+\log\cosh(x^{*})
  \;=\;\phi_L(x^{*}).
  \label{eq:sup-phi}
\end{equation}
Combining~\eqref{eq:rate-formula} and~\eqref{eq:sup-phi},
\begin{equation}
  \frac{1}{N}\log\mathcal{I}_N
  \;\xrightarrow[N\to\infty]{}\;
  -L(L-1)+L\,\phi_L(x^{*}),
  \label{eq:rate-final}
\end{equation}
which defines the noise rate
\begin{equation}
  \rho_L \;:=\; L\bigl[(L-1)-\phi_L(x^{*})\bigr],
  \label{eq:rho-app}
\end{equation}
in agreement with~\eqref{eq:rho-main}.

\subsection*{Gaussian fluctuations and prefactor $K_L$}
\label{app:saddle-prefactor}

Having identified the saddle $\bm{m} = m^{*}\bm{1}$ and the exponential rate
$\rho_L$, we now extract the polynomial prefactor $K_L$ by performing a
systematic Gaussian expansion of $\mathcal{I}_N$ around that saddle.

Since each cavity magnetisation $\hat m_\mu^{c}$ is an empirical average of
$N-1$ i.i.d.\ Rademacher variables, it fluctuates on the scale $N^{-1/2}$
around any deterministic value. We therefore introduce rescaled deviations
from the saddle,
\begin{equation}
  \hat m_\mu^{c} \;=\; m^{*} + \frac{y_c}{\sqrt{N}},
  \qquad y_c \in \mathbb{R}, \quad c = 1,\dots,L,
  \label{eq:rescaled-dev}
\end{equation}
and expand the action $N\Psi(\bm{\hat m}_\mu)$ in powers of $N^{-1/2}$.
Because $m^{*}\bm{1}$ is a stationary point of $\Psi$, the linear term in
$\bm{y}$ vanishes identically, and the expansion reads
\begin{equation}
  N\,\Psi(\bm{\hat m}_\mu)
  \;=\;
  N\,\Psi(m^{*}\bm{1})
  \;-\;
  \frac{1}{2}\,\bm{y}^{\mathsf{T}}\mathcal{H}_L^{*}\,\bm{y}
  \;+\;
  \mathcal{O}(N^{-1/2}),
  \label{eq:action-expansion}
\end{equation}
where $\mathcal{H}_L^{*}$ denotes the negative of the Hessian of $\Psi$
evaluated at $m^{*}\bm{1}$. To compute its entries we differentiate
$\Psi(\bm m) = 2\sum_{c<d}m^c m^d - \sum_c I_R(m^c)$ twice: the diagonal
entries receive no contribution from the bilinear term and give
$-\partial_{m^c}^2 I_R(m^c)\big|_{m^*} = \frac{1}{1-m^{*\,2}}$, using
$I_R''(m)=(1-m^2)^{-1}$; the off-diagonal entries come entirely from the
bilinear term and equal $-2\cdot(-1)=2\cdot(-1)$... more precisely, the
Hessian of $-\Psi$ at off-diagonal positions is $-2$, so altogether
\begin{equation}
  (\mathcal{H}_L^{*})_{cd}
  \;=\;
  \frac{1}{1-m^{*\,2}}\,\delta_{cd}
  \;-\;
  2\,(1-\delta_{cd}).
  \label{eq:Hessian-entries}
\end{equation}
Writing $\bm{J}=\bm{1}\bm{1}^{\mathsf{T}}$ for the $L\times L$ all-ones matrix,
this is equivalently $\mathcal{H}_L^{*}=\bigl(\frac{1}{1-m^{*\,2}}+2\bigr)I_L - 2\bm{J}$,
a rank-one update of a scalar multiple of the identity whose spectrum is
immediate: the symmetric eigenvector $\bm{1}/\sqrt{L}$ has eigenvalue
$\lambda_{\parallel}=\frac{1}{1-m^{*\,2}}-2(L-1)$, while every vector in
$\bm{1}^{\perp}$ has eigenvalue $\lambda_{\perp}=\frac{1}{1-m^{*\,2}}+2$.
The latter is manifestly positive; positivity of $\lambda_{\parallel}$ follows
from the geometry of the non-trivial fixed point: at $m^{*}=\tanh(2(L-1)m^{*})$
the slope of $\tanh$ satisfies $1-m^{*\,2}=\mathrm{sech}^2(2(L-1)m^*)
< \frac{1}{2(L-1)}$, since the non-trivial crossing occurs where the curve
$\tanh(x)$ is already less steep than the line $x/[2(L-1)]$. Hence
$\mathcal{H}_L^{*}$ is positive definite and the Gaussian integral below
converges.

We now turn to the polynomial insertions. At $\bm{\hat m}_\mu = m^{*}\bm{1}$
the cavity field of every layer equals $F_\mu^c = (L-1)m^{*}$, so
\begin{equation}
  \sinh^{2}(F_\mu^{a})\Big|_{m^*\bm{1}}
  =
  \sinh^{2}\!\bigl((L-1)m^{*}\bigr),
  \qquad
  \prod_{c\neq a}\cosh\!\bigl(2F_\mu^{c}\bigr)\Big|_{m^*\bm{1}}
  =
  \cosh\!\bigl(2(L-1)m^{*}\bigr)^{L-1}.
  \label{eq:poly-at-saddle}
\end{equation}
Being smooth functions of $\bm{\hat m}$, these factors deviate from their
saddle values only at order $N^{-1/2}$, and their fluctuations do not
contribute to the leading prefactor.

Collecting all ingredients, we substitute the Taylor
expansion~\eqref{eq:action-expansion} and the frozen
insertions~\eqref{eq:poly-at-saddle} into~\eqref{eq:IN}. By Cram\'{e}r's
theorem the joint density of $\bm{\hat m}_\mu$ is
$\propto e^{-NI(\bm{\hat m}_\mu)+o(N)}$, and the change of variables
$\hat m_\mu^c = m^* + y_c/\sqrt{N}$ (whose Jacobian $N^{-L/2}$ is absorbed
into the density normalisation) gives
\begin{align}
  \mathcal{I}_N
  &\;=\;
  \int\!\mathrm{d}\bm{\hat m}\;
  e^{N\Psi(\bm{\hat m})}\,
  \sinh^{2}(F^{a})\prod_{c\neq a}\cosh(2F^{c})\;
  (1+o(1))
  \notag\\
  &\;=\;
  e^{N\Psi(m^{*}\bm{1})}\,
  \sinh^{2}\!\bigl((L-1)m^{*}\bigr)\,
  \cosh\!\bigl(2(L-1)m^{*}\bigr)^{L-1}
  \int_{\mathbb{R}^L}\!\mathrm{d}\bm{y}\;
  e^{-\frac{1}{2}\bm{y}^{\mathsf{T}}\mathcal{H}_L^{*}\bm{y}}\;
  (1+o(1)).
  \label{eq:IN-assembled}
\end{align}
The standard Gaussian integral evaluates
to $(2\pi)^{L/2}/\sqrt{\det\mathcal{H}_L^{*}}$, where
$\det\mathcal{H}_L^{*}=\lambda_{\parallel}\,\lambda_{\perp}^{L-1}$ is the
product of the eigenvalues computed above. Recalling
from~\eqref{eq:rate-final} that $e^{N\Psi(m^{*}\bm{1})}=e^{-N\rho_L}(1+o(1))$,
we conclude
\begin{equation}
  \mathcal{I}_N \;=\; K_L\,e^{-N\rho_L}\,\bigl(1+o(1)\bigr),
\end{equation}
with the prefactor
\begin{equation}
  K_L \;=\;
  \frac{(2\pi)^{L/2}}{\sqrt{\det\mathcal{H}_L^{*}}}\,
  \sinh^{2}\!\bigl((L-1)m^{*}\bigr)\,
  \cosh\!\bigl(2(L-1)m^{*}\bigr)^{L-1},
  \label{eq:KL-final-app}
\end{equation}
in agreement with~\eqref{eq:KL-main} of the main text.

\subsection*{Numerical values and asymptotics}
\label{app:saddle-numbers}

\begin{table}[H]
\centering
\begin{tabular}{c|cccc}
\toprule
$L$ & $x^{*}$ & $m^{*}$ & $\phi_L(x^{*})$ & $\rho_L$ \\
\midrule
2  & 1.9150  & 0.9575    & 0.3265 & 1.3470 \\
3  & 3.9973  & 0.9993    & 1.3072 & 2.0784 \\
4  & 5.9999  & 0.99999   & 2.3069 & 2.7726 \\
5  & 8.0000  & $\approx 1$ & 3.3069 & 3.4657 \\
10 & 18.0000 & $\approx 1$ & 8.3069 & 6.9315 \\
\bottomrule
\end{tabular}
\caption{Saddle data and noise rate $\rho_L$ for moderate $L$.}
\label{tab:saddle}
\end{table}

For $L\to\infty$ the saddle drifts to $x^{*}\sim 2(L-1)$ and $m^{*}\to 1^{-}$, giving $\phi_L(x^{*})\sim (L-1)-\log 2$ and the asymptotic rate
\begin{equation}
  \rho_L \;\sim\; L\log 2,
  \qquad L\to\infty.
  \label{eq:asymptotics}
\end{equation}
The rate is positive for every $L\ge 2$, ensuring exponential suppression of the noise contribution at fixed pattern.

\paragraph{Numerical validation: saddle and rate scaling}
Figure~\ref{fig:doubling}(a) plots the variational functional $\phi_L(x)$ for $L\in\{2,3,4\}$, whose interior maximiser is the saddle $x^{*}$ of Table~\ref{tab:saddle}; the equivalent graphical solution of $\tanh(x^{*})=x^{*}/[2(L-1)]$ appears in the main-text Figure~\ref{fig:storage}(a), where the rapid drift of $x^{*}$ towards $2(L-1)$ is visible. Figure~\ref{fig:doubling}(b) verifies the scaling~\eqref{eq:asymptotics}: the rate $\rho_L$ tracks the asymptote $L\log 2$ from above while the Gaussian prefactor $K_L$ grows rapidly with the width $L$. The doubling identity~\eqref{eq:doubling-identity} at $L=2$, $\rho_2=2\rho_1^{\mathrm{Dem}}$ with $\rho_1^{\mathrm{Dem}}=1-\phi_2(x^{*})=0.6735$, is derived at the end of this appendix and marked in the inset of Figure~\ref{fig:storage}(b).

\begin{figure}[H]
  \centering
  \includegraphics[width=\linewidth]{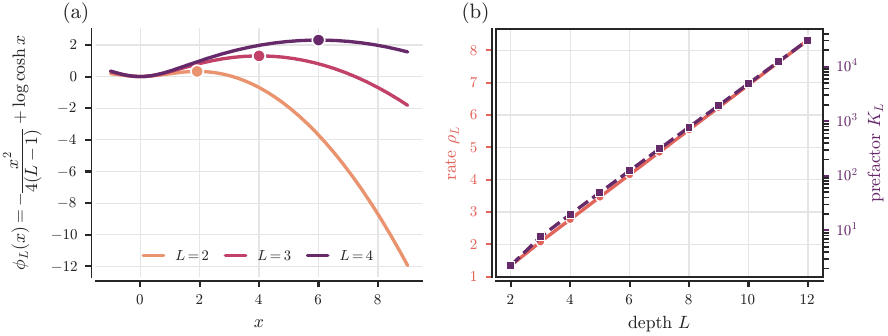}
  \caption{\textbf{Saddle landscape, rate and prefactor.}
  \textbf{(a)} The variational functional $\phi_L(x)=-x^{2}/[4(L-1)]+\log\cosh x$
  for $L=2,3,4$; markers locate the interior maximiser $x^{*}$ of
  Table~\ref{tab:saddle}.
  \textbf{(b)} The noise rate $\rho_L$ (left axis, linear, approaching $L\log2$)
  and the Gaussian fluctuation prefactor $K_L$ (right axis, logarithmic) against
  width $L$. Both are smooth and monotonically increasing; the two curves are
  drawn on independent scales, so where they appear to meet or separate (around
  $L=4$) reflects only the choice of axes, not any feature of $\rho_L$ or $K_L$.}
  \label{fig:doubling}
\end{figure}

\subsection*{Doubling identity for $L=2$}
\label{app:doubling}

For $L=2$ the saddle equation~\eqref{eq:phi} becomes $\tanh(x)=x/2$, with $\phi_2(x)=-x^{2}/4+\log\cosh(x)$. This is exactly the saddle of the single-layer exponential Hopfield model, with rate $\rho_1^{\mathrm{Dem}}=1-\phi_2(x^{\ast})$. From~\eqref{eq:rho-main},
\begin{equation}
  \rho_2 \;=\; 2\bigl[1-\phi_2(x^{\ast})\bigr] \;=\; 2\,\rho_1^{\mathrm{Dem}},
  \label{eq:doubling-identity}
\end{equation}
so that $\rho_2/[L(L-1)]=\rho_1^{\mathrm{Dem}}$ and the critical overlap at $L=2$ coincides with the single-layer one: a noise pattern must generate a large overlap simultaneously in two independent layers, and the two layer-overlaps being mutually independent yields the squared exponential suppression. Numerically $\rho_1^{\mathrm{Dem}}=1-\phi_2(x^{\ast})=0.6735$ and $\rho_2=1.3470=2\rho_1^{\mathrm{Dem}}$, as marked in Figure~\ref{fig:doubling}\footnote{We note a numerical caveat: for $L \ge 10$, the saddle magnetization $m^{\ast}$ approaches $1$ so rapidly that $1-m^{\ast,2}$ falls below double-precision machine epsilon, leading to catastrophic cancellation in the numerical evaluation of $\mathcal{H}_L^{\ast}$. To compute $K_L$ accurately at large $L$, one must substitute $1/(1-m^{\ast,2})$ with the analytically equivalent expression $\cosh^{2}(x^{\ast})$ derived from the saddle equation.}.

\section{Corrupted-state cavity expansion}
\label{app:corrupted}

This appendix details the analysis summarised in \S\ref{sec:basins}. Throughout, the network is initialised at the corrupted state \eqref{eq:corrupted-init}, $\sigma_j^{c}=s_j^{c}\,\xi_j^{1,c}$, with masks $s_j^{c}\in\{-1,+1\}$ i.i.d.\ across $(j,c)$, $\E[s_j^{c}]=r\in(0,1]$, and all cavity quantities are evaluated at this state.

\subsection*{Cavity magnetisations and local field under corruption}
\label{app:corrupted-mhat}

For the recalled archetype,
\begin{equation}
  \hat m_1^{a}\;=\;\frac{1}{N}\sum_{j\neq i}\xi_j^{1,a}\,\sigma_j^{a}
            \;=\;\frac{1}{N}\sum_{j\neq i}s_j^{a},
\end{equation}
an empirical mean of i.i.d.\ variables with mean $r$ and variance $1-r^{2}$, independent across layers. By the Central Limit Theorem,
\begin{equation}
  \hat m_1^{a}\;=\;r+\frac{\zeta^{a}}{\sqrt{N}}+\mathcal{O}(N^{-1}),
  \qquad
  \zeta^{a}\overset{d}{\sim}\mathcal{N}\bigl(0,1-r^{2}\bigr),
  \label{eq:mhat-corrupted}
\end{equation}
with $\{\zeta^{a}\}_{a=1}^{L}$ mutually independent. Consequently $F_1^{a}=(L-1)r+\mathcal{O}(N^{-1/2})$ and, substituting \eqref{eq:mhat-corrupted} in $\hat E_1=N\sum_{a<b}\hat m_1^{a}\hat m_1^{b}-N\binom{L}{2}$ and using $\sum_{a<b}(\zeta^{a}+\zeta^{b})=(L-1)\sum_{c}\zeta^{c}$,
\begin{equation}
  \hat E_1
  \;=\;
  -N\binom{L}{2}\bigl(1-r^{2}\bigr)+\mathcal{Z}_r+\mathcal{O}(N^{-1/2}),
  \qquad
  \mathcal{Z}_r:=r(L-1)\sqrt{N}\sum_{c}\zeta^{c}+\sum_{a<b}\zeta^{a}\zeta^{b}.
  \label{eq:Ehat-corrupted}
\end{equation}
The on-site factor involves only the masks at site $i$: since $\xi_i^{1,c}\sigma_i^{c}=s_i^{c}$,
\begin{equation}
  \Phi_1^{(\backslash a)}
  \;=\;
  \exp\!\Bigl(\sum_{c\neq a}s_i^{c}\,F_1^{c}\Bigr),
\end{equation}
which is independent of the cavity variables \eqref{eq:mhat-corrupted}. Using the oddness of $\sinh$, the signal contribution to $X_i^{a}=\xi_i^{1,a}h_i^{a}$ is therefore
\begin{equation}
  X_i^{(1|a)}
  \;=\;
  e^{\hat E_1}\,\sinh\bigl(F_1^{a}\bigr)\,
  \exp\!\Bigl(\sum_{c\neq a}s_i^{c}\,F_1^{c}\Bigr).
  \label{eq:X1-corrupted}
\end{equation}

\subsection*{Noise channel under corruption}
\label{app:corrupted-noise}

For $\mu\neq 1$ define $v_j^{(\mu,c)}:=\xi_j^{\mu,c}\,s_j^{c}\,\xi_j^{1,c}$. Since $\xi_j^{\mu,c}$ is a symmetric Rademacher sign independent of $(s_j^{c},\xi_j^{1,c})$, the family $\{v_j^{(\mu,c)}\}$ is i.i.d.\ symmetric Rademacher for every $r$, across $j$, $c$ and $\mu\neq 1$. The noise contributions
\begin{equation}
  X_i^{(\mu|a)}
  =\xi_i^{1,a}\xi_i^{\mu,a}\,e^{\hat E_\mu}\sinh(F_\mu^{a})\,\Phi_\mu^{(\backslash a)},
  \qquad
  \hat m_\mu^{c}=\frac{1}{N}\sum_{j\neq i}v_j^{(\mu,c)},
  \qquad
  \Phi_\mu^{(\backslash a)}=\exp\!\Bigl(\sum_{c\neq a}v_i^{(\mu,c)}F_\mu^{c}\Bigr),
\end{equation}
have therefore exactly the same joint law as at the uncorrupted state $r=1$, where $v_j^{(\mu,c)}=u_j^{(\mu,c)}$ of \ref{app:moments}. In particular: $\E[X_i^{(\mu|a)}]=0$ through the on-site factor $\xi_i^{1,a}$, as in \ref{app:noise-mean}; the off-diagonal second moments vanish as in \ref{app:offdiag}; the signal--noise cross terms vanish because $\xi_i^{\mu,a}$ is centred and independent of the masks; and the per-pattern variance is $K_L\,e^{-N\rho_L}(1+o(1))$ as computed in \ref{app:saddle}. The noise statistics are therefore insensitive to the corruption level, as anticipated in \S\ref{sec:basins}: $\sigma^{2}(r)=(P-1)\,K_L\,e^{-N\rho_L}(1+o(1))$ for every $r$, exactly as in \eqref{eq:variance}.

\subsection*{Annealed signal: tilted large deviations}
\label{app:corrupted-signal}

The first moment of \eqref{eq:X1-corrupted} factorises over the independent on-site and cavity randomness. Averaging the on-site masks first, conditionally on the cavity variables, with $\E[e^{xs}]=\cosh x+r\sinh x$ for $s\in\{-1,+1\}$ of mean $r$,
\begin{equation}
  \mu_1(r)
  \;=\;
  \E\Bigl[\,e^{\hat E_1}\,\sinh(F_1^{a})
  \prod_{c\neq a}\bigl(\cosh F_1^{c}+r\sinh F_1^{c}\bigr)\Bigr],
  \label{eq:mu1-r-int}
\end{equation}
the remaining expectation running over $\bm{\hat m}_1\in[-1,1]^{L}$. Equation \eqref{eq:mu1-r-int} is in large-deviation form, exactly as \eqref{eq:noise-sq-main}, with two differences: the exponent carries a single power of $\hat E_1$, and the reference measure is biased. By Cram\'er's theorem the empirical mean of the masks of one layer obeys an LDP with the tilted rate function
\begin{equation}
  I_r(m)
  \;=\;
  \frac{1+m}{2}\log\frac{1+m}{1+r}+\frac{1-m}{2}\log\frac{1-m}{1-r}
  \;=\;
  I_R(m)-a_r\,m+\log\cosh(a_r),
  \qquad
  a_r=\tanh^{-1}(r),
  \label{eq:Ir-tilted}
\end{equation}
which vanishes only at $m=r$, and the joint rate is $\sum_{c}I_r(m^{c})$ by independence across layers. Varadhan's lemma then gives
\begin{equation}
  \frac{1}{N}\log\mu_1(r)
  \;\xrightarrow[N\to\infty]{}\;
  -\binom{L}{2}+\sup_{\bm m\in[-1,1]^{L}}\Psi_r(\bm m),
  \qquad
  \Psi_r(\bm m)\;:=\;\sum_{c<d}m^{c}m^{d}-\sum_{c}I_r(m^{c}),
  \label{eq:varadhan-r}
\end{equation}
the smooth insertions of \eqref{eq:mu1-r-int} contributing only to the prefactor.

\paragraph{Symmetric saddle and uniqueness}
The first-order conditions read $\sum_{d\neq c}m^{d}=\tanh^{-1}(m^{c})-a_r$ for $c=1,\dots,L$. Subtracting two of them,
\begin{equation}
  m^{c'}-m^{c}\;=\;\tanh^{-1}(m^{c})-\tanh^{-1}(m^{c'}),
\end{equation}
and the strict monotonicity of $\tanh^{-1}$ forces $m^{c}=m^{c'}$ for every pair: every interior stationary point lies on the symmetric ray $m^{c}\equiv m$, where $\Psi_r$ reduces to
\begin{equation}
  \Lambda_{L,r}(m)\;=\;\binom{L}{2}m^{2}-L\,I_r(m),
  \label{eq:Lambda-r}
\end{equation}
with stationarity condition
\begin{equation}
  \tanh^{-1}(m_s)\;=\;(L-1)\,m_s+a_r
  \quad\Longleftrightarrow\quad
  m_s=\tanh\bigl((L-1)m_s+a_r\bigr).
  \label{eq:saddle-r-app}
\end{equation}
For $r>0$ the supremum is attained at the largest root $m_s>r$; the mirror saddle near $-m_s$ carries the extra cost $2La_r m_s$ and is exponentially subdominant. It matters only at $r=0$, where the two saddles cancel exactly in the odd integrand and $\mu_1(0)=0$ by parity, as it must for an uncorrelated input.

\paragraph{Explicit rate}
Setting $x_s:=(L-1)m_s+a_r=\tanh^{-1}(m_s)$ and using $I_R(m_s)=m_s x_s-\log\cosh(x_s)$ in \eqref{eq:Ir-tilted},
\begin{align}
  \Lambda_{L,r}(m_s)
  &=\;\binom{L}{2}m_s^{2}
   -L\bigl[m_s x_s-\log\cosh x_s-a_r m_s+\log\cosh a_r\bigr] \notag\\
  &=\;-\binom{L}{2}m_s^{2}+L\bigl[\log\cosh x_s-\log\cosh a_r\bigr],
\end{align}
the second line following from $x_s-a_r=(L-1)m_s$. Substituting in \eqref{eq:varadhan-r} yields the signal rate quoted in \eqref{eq:Sigma-main},
\begin{equation}
  \Sigma_L(r)
  \;=\;\binom{L}{2}-\Lambda_{L,r}(m_s)
  \;=\;L\Bigl[\tfrac{L-1}{2}\bigl(1+m_s^{2}\bigr)-\log\cosh(x_s)+\log\cosh(a_r)\Bigr].
  \label{eq:Sigma-app}
\end{equation}
\emph{Sanity checks.} At $r\to1$, $a_r\to\infty$ forces $m_s\to1$; using $\log\cosh y=y-\log 2+o(1)$ for both $x_s$ and $a_r$, the bracket tends to $(L-1)-(x_s-a_r)=0$, hence $\Sigma_L(1)=0$ and the $\mathcal{O}(1)$ signal of \ref{app:mu1} is recovered. At fixed $r<1$, comparing $\Lambda_{L,r}$ at $m_s$ and at $m=r$ (where $I_r$ vanishes) gives
\begin{equation}
  0\;<\;\Sigma_L(r)\;\le\;\binom{L}{2}\bigl(1-r^{2}\bigr),
  \label{eq:Sigma-bounds}
\end{equation}
with strict upper inequality for $r<1$, since the saddle satisfies $m_s(r)>r$: the annealed rate is strictly smaller than the typical one, see below.

\paragraph{Gaussian fluctuations and prefactor}
Expanding $\Psi_r$ to second order around $m_s\bm 1$, in full analogy with \ref{app:saddle}, the negative Hessian is
\begin{equation}
  (\mathcal{H}_{L,r}^{*})_{cd}
  \;=\;
  \frac{1}{1-m_s^{2}}\,\delta_{cd}-\bigl(1-\delta_{cd}\bigr)
  \;=\;
  \Bigl(\frac{1}{1-m_s^{2}}+1\Bigr)I_L-\bm J,
\end{equation}
with eigenvalues $\lambda_{\parallel}=\frac{1}{1-m_s^{2}}-(L-1)$ on the symmetric mode and $\lambda_{\perp}=\frac{1}{1-m_s^{2}}+1$ on $\bm 1^{\perp}$. Positivity of $\lambda_{\parallel}$ follows from the saddle geometry: at the largest root the curve $\tanh^{-1}(m)$ crosses the line $(L-1)m+a_r$ from below, i.e.\ $\frac{1}{1-m_s^{2}}>L-1$. Freezing the smooth insertions of \eqref{eq:mu1-r-int} at the saddle, where $F_1^{c}=(L-1)m_s$ for every $c$, and performing the Gaussian integral,
\begin{equation}
  C_L(r)
  \;=\;
  \frac{(2\pi)^{L/2}}{\sqrt{\det\mathcal{H}_{L,r}^{*}}}\;
  \sinh\bigl((L-1)m_s\bigr)\,
  \bigl[\cosh\bigl((L-1)m_s\bigr)+r\sinh\bigl((L-1)m_s\bigr)\bigr]^{L-1},
  \label{eq:CLr-app}
\end{equation}
with $\det\mathcal{H}_{L,r}^{*}=\lambda_{\parallel}\lambda_{\perp}^{L-1}$, in the same normalisation convention as \eqref{eq:KL-final-app}; see Remark~\ref{rem:prefactors-r} for the finite-$N$ corrected version.

\subsection*{Annealed versus typical signal}
\label{app:corrupted-typical}

The annealed average \eqref{eq:mu1-r-int} is controlled by mask configurations with $\hat m_1^{c}\equiv m_s(r)>r$, which carry probability $e^{-NL\,I_r(m_s)}$: exponentially rare. The typical behaviour is read off \eqref{eq:mhat-corrupted}--\eqref{eq:Ehat-corrupted} directly. This is already visible at the level of the Gaussian fluctuation $\mathcal{Z}_r$: writing $\sum_{a<b}\zeta^{a}\zeta^{b}=\frac12[(\sum_c\zeta^{c})^{2}-\sum_c(\zeta^{c})^{2}]$ and performing the Gaussian integral exactly,
\begin{equation}
  \E_{\zeta}\bigl[e^{\mathcal{Z}_r}\bigr]
  \;=\;
  \exp\!\left(\frac{L(L-1)^{2}\,r^{2}(1-r^{2})}{2\bigl[1-(L-1)(1-r^{2})\bigr]}\,N+\mathcal{O}(1)\right),
  \qquad\text{if }(L-1)(1-r^{2})<1,
  \label{eq:Zr-annealed}
\end{equation}
while for $(L-1)(1-r^{2})\ge 1$ the average diverges at this (CLT) level of description, the quadratic term $\sum_{a<b}\zeta^a\zeta^b$ overwhelming the Gaussian decay of the $\zeta$'s. In either case the annealed average of the fluctuation grows exponentially in $N$ (or worse): it is dominated by rare $\zeta$'s, not by typical ones. The exact gap between annealed and typical rates is $\binom{L}{2}(1-r^{2})-\Sigma_L(r)>0$, to which \eqref{eq:Zr-annealed} reduces at quadratic order in the deviation $m-r$; the divergence of \eqref{eq:Zr-annealed} for $(L-1)(1-r^{2})\ge1$ merely signals that the dominant deviations then leave the central $\mathcal{O}(N^{-1/2})$ window, where they are correctly controlled by the tilted rate function \eqref{eq:Ir-tilted} rather than by its quadratic approximation. By contrast $\mathcal{Z}_r/N\to 0$ almost surely, so
\begin{equation}
  \frac{1}{N}\hat E_1\;\xrightarrow[N\to\infty]{\mathrm{a.s.}}\;-\binom{L}{2}\bigl(1-r^{2}\bigr),
  \qquad
  \frac{1}{N}\log X_i^{(1|a)}\;\xrightarrow[N\to\infty]{\mathrm{a.s.}}\;-\binom{L}{2}\bigl(1-r^{2}\bigr),
  \label{eq:typical-rate}
\end{equation}
the on-site factor being $\mathcal{O}(1)$: at the typical state $F_1^{c}\to(L-1)r$ and the average over the on-site masks gives $\bigl[\cosh((L-1)r)+r\sinh((L-1)r)\bigr]^{L-1}$, although the $\mathcal{O}(\sqrt N)$ Gaussian fluctuation $\mathcal{Z}_r$ prevents any deterministic $\mathcal{O}(1)$ prefactor from being attached to the typical signal. Comparing the typical rate \eqref{eq:typical-rate} with the noise variance \eqref{eq:variance} amounts to the replacement $\Sigma_L(r)\to\binom{L}{2}(1-r^{2})$ in the capacity exponent of \eqref{eq:Pc-r},
\begin{equation}
  \varepsilon_L^{\mathrm{typ}}(r)\;=\;\rho_L-L(L-1)\bigl(1-r^{2}\bigr),
  \label{eq:eps-typ}
\end{equation}
which is positive iff $1-r^{2}<\rho_L/[L(L-1)]$. Using $\rho_L=L[(L-1)-\phi_L(x^{*})]$ this condition takes the closed form
\begin{equation}
  r\;>\;r_c^{\mathrm{typ}}(L)\;=\;\sqrt{\frac{\phi_L(x^{*})}{L-1}}
  \;\underset{L\to\infty}{\sim}\;\sqrt{1-\frac{\log 2}{L-1}}\;\longrightarrow\;1,
  \label{eq:rc-typ-app}
\end{equation}
equivalently $r_c^{\mathrm{typ}}=\sqrt{1-\rho_L/[L(L-1)]}$, as quoted in \eqref{eq:rc-def}. At $r=1$ the masks are deterministic, $\zeta^{a}\equiv 0$ and $\mathcal{Z}_r\equiv 0$, and annealed and typical coincide, both reducing to the exact computation of \ref{app:mu1}: the bracket $\cosh(L-1)+\sinh(L-1)=e^{L-1}$ gives $e^{(L-1)^{2}}$, the leading rate vanishes, and the signal is $e^{-(L-1)}\sinh(L-1)$.

\paragraph{Which criterion the dynamics realises}
The two rates $\Sigma_L(r)$ (annealed) and $\binom{L}{2}(1-r^{2})$ (typical)
bracket the finite-$N$ recovery threshold, and it is legitimate to ask which one
the one-step dynamics actually obeys. The point is settled empirically in
\S\ref{sec:hmm}, and the answer is the annealed one; here we record why
this is the theoretically consistent reading, not a coincidence. The one-step
overlap~\eqref{eq:m1step-r} is by construction $\erf\!\bigl(\mu_1(r)/\sqrt{2\sigma^2(r)}\bigr)$
with $\mu_1(r)=\E[X_i^{(1|a)}]$ the annealed first moment, precisely the
prescription of the single-layer model~\cite{ALBANESE2026131223}. There the
cavity exponent is linear in the masks, the annealed average factorises,
$\mu_1=e^{-1}\sinh(1)\,[\tfrac12((1+r)+(1-r)e^{-2})]^{N-1}$, and it equals
the typical signal because $\log X^{(1)}$ is a sum of i.i.d.\ per-site
contributions and self-averages; the annealed/typical distinction is empty and
the threshold $r_c\simeq0.337$ matches the Monte Carlo. The multilayer exponent
$\hat E_1$ is quadratic in the masks~\eqref{eq:Ehat-corrupted}, $\log X^{(1)}$ no
longer self-averages, and the distinction opens. But the signal is still carried
by a single pattern: $\hat E_1$ is one random variable per disorder
realisation, and the retrieval curve is an average over independent realisations,
which samples its $\mathcal{O}(\sqrt N)$ upward fluctuations. Averaging
$X^{(1)}=e^{\hat E_1}(\cdots)$ over the masks returns $\mu_1(r)$ by definition, so
the annealed moment is the operative one and the closed form built on it is the
one the data confirm. Replacing $\mu_1(r)$ by its typical value would discard the
fluctuations that dominate the average at every finite $N$; the typical rate
would control the recall only in an $N\to\infty$ limit at fixed sub-exponential
load $P$, incompatible with $P\sim e^{N\varepsilon_L(r)}$. Finally, the same
rare-event inflation makes the annealed noise variance~\eqref{eq:variance}
a conservative overestimate (Remark~\ref{rem:annealed-noise}), so the measured
basin is if anything slightly larger than the annealed-signal /
annealed-noise erf predicts --an overshoot toward larger basins, never toward the
typical threshold.

\subsection*{Numerical thresholds and finite-$N$ checks}
\label{app:corrupted-numbers}

Table~\ref{tab:basins} of the main text collects the thresholds $r_c(L)$ (from $\Sigma_L(r_c)=\rho_L/2$) and $r_c^{\mathrm{typ}}(L)$ for $L=2,\dots,10$; Table~\ref{tab:basin-L2} resolves the $L=2$ case in $r$.

\paragraph{Large-$L$ asymptotics of the thresholds}
As $L\to\infty$ the two criteria behave in opposite ways. For the typical one, $\phi_L(x^{*})\sim(L-1)-\log 2$ in \eqref{eq:rc-typ-app} gives $1-r_c^{\mathrm{typ}\,2}\simeq\log 2/(L-1)\to 0$, i.e.\ a tolerated Hamming radius shrinking as $d_c^{\mathrm{typ}}\simeq\log2/[4(L-1)]$. For the annealed one, the saddle freezes: $m_s\to1$ and $\log\cosh(x_s)=x_s-\log2+o(1)$, so that
\begin{equation}
  \Sigma_L(r)\;=\;L\bigl[\log 2-a_r+\log\cosh(a_r)\bigr]+o(L),
\end{equation}
and the threshold condition $\Sigma_L(r_c)=\rho_L/2\sim(L/2)\log2$ reduces to
\begin{equation}
  a_r-\log\cosh(a_r)\;=\;\tfrac12\log 2
  \quad\Longleftrightarrow\quad
  e^{-2a_r}\;=\;\sqrt2-1,
\end{equation}
whose solution is $r_c=\tanh(a_r)=\sqrt2-1\approx0.4142$, the finite limit quoted in \S\ref{sec:basins} and approached from below in Table~\ref{tab:basins}.

\begin{table}[H]
\centering
\begin{tabular}{c|ccc}
\toprule
$r$ & $m_s(r)$ & $\Sigma_2(r)$ & $\varepsilon_2(r)=\rho_2-2\Sigma_2(r)$ \\
\midrule
0.2    & 0.7335 & 0.8066 & $-0.2663$ \\
0.3    & 0.8060 & 0.6951 & $-0.0433$ \\
0.3195 & 0.8172 & 0.6735 & $0$ \\
0.4    & 0.8565 & 0.5851 & $+0.1768$ \\
0.5    & 0.8945 & 0.4781 & $+0.3908$ \\
0.5714 & 0.9164 & 0.4039 & $+0.5391$ \\
0.7    & 0.9484 & 0.2754 & $+0.7961$ \\
0.9    & 0.9854 & 0.0882 & $+1.1706$ \\
\bottomrule
\end{tabular}
\caption{Tilted saddle data at $L=2$: saddle point $m_s(r)$, annealed signal rate $\Sigma_2(r)$ and capacity exponent $\varepsilon_2(r)$. The exponent vanishes at the annealed threshold $r_c=0.3195$; the row $r=0.5714$ marks the typical threshold $r_c^{\mathrm{typ}}=\sqrt{\phi_2(x^{*})}$.}
\label{tab:basin-L2}
\end{table}

\begin{remark}[Finite-$N$ prefactors]
\label{rem:prefactors-r}
As in \ref{app:saddle}, the prefactor \eqref{eq:CLr-app} is written at leading Laplace order, treating the cavity magnetisations as empirical means of $N$ (rather than $N-1$) variables and absorbing the local-CLT normalisation. Restoring both, the corrected prefactor reads
\begin{equation}
  \widehat C_L(r)
  \;=\;
  \sinh\bigl((L-1)m_s\bigr)
  \bigl[\cosh\bigl((L-1)m_s\bigr)+r\sinh\bigl((L-1)m_s\bigr)\bigr]^{L-1}
  \frac{(1-m_s^{2})^{-L/2}}{\sqrt{\det\mathcal{H}_{L,r}^{*}}}\,
  e^{\Sigma_L(r)-\binom{L}{2}(1+m_s^{2})}.
  \label{eq:CL-r-hat}
\end{equation}
As $r\to 1$ one finds $\widehat C_L(r)\to e^{-(L-1)}\sinh(L-1)$, matching exactly the perfect-recall signal \eqref{eq:mu1}. We verified \eqref{eq:CL-r-hat} at $L=2$ against the exact double-binomial enumeration of $\mu_1(r)$: at $N=800$ the enumerated ratio $\mu_1(r)\,e^{N\Sigma_2(r)}$ equals $0.5359$ at $r=0.5$ and $0.4887$ at $r=0.7$, against $\widehat C_2(0.5)=0.5368$ and $\widehat C_2(0.7)=0.4896$, with residual $\mathcal{O}(N^{-1})$ corrections. The same finite-$N$ bookkeeping applies to the noise prefactor $K_L$ (cavity exclusion plus a factor $2$ from the two symmetric saddles $\pm m^{*}\bm 1$ of the even integrand \eqref{eq:IN}); none of it affects the rates $\rho_L$, $\Sigma_L(r)$ or the thresholds of Table~\ref{tab:basins}.
\end{remark}

\section{Single-layer squared-overlap model}
\label{app:squared_overlap}

We consider a single-layer ($L=1$) variant of the exponential Hopfield model in which the exponent is built from the squared Mattis magnetisation rather than from the linear one.  The cost function reads
\begin{equation}
  \mathcal{H}_{\mathrm{sq}}(\bm\sigma\,|\,\bm\xi)
  \;=\;
  -\,N\sum_{\mu=1}^{P}\exp\!\bigl[\,N\,(m_\mu^{2}-1)\,\bigr],
  \label{eq:H_sq}
\end{equation}
with $m_\mu = \frac{1}{N}\sum_i \xi_i^\mu\sigma_i$ the standard Mattis magnetisation.

The key structural difference with respect to the linear-exponent model \cite{ALBANESE2026131223}
\begin{equation}
  \mathcal{H}_{\mathrm{lin}}(\bm\sigma\,|\,\bm\xi)
  \;=\;
  -\,N\sum_{\mu=1}^{P}\exp\!\bigl[\,N\,(m_\mu-1)\,\bigr]
  \label{eq:H_lin_app}
\end{equation}
is the restored $\mathbb{Z}_2$ symmetry: since $\mathcal{H}_{\mathrm{sq}}$ depends only on $m_\mu^2$, it is invariant under the global spin flip $\bm\sigma\mapsto-\bm\sigma$, and each stored pattern is retrieved as a pair $\{\bm\xi^\mu,-\bm\xi^\mu\}$.  The linear-exponent model \eqref{eq:H_lin_app} breaks this symmetry, favouring $m_\mu=+1$ only.

\subsection*{Cavity decomposition and local field}
\label{sec:sq_cavity}

Introducing the cavity magnetisation $\hat m_\mu = \frac{1}{N}\sum_{j\neq i}\xi_j^\mu\sigma_j$, the squared overlap expands as
\begin{equation}
  N\,(m_\mu^2-1) \;=\; \underbrace{N\,(\hat m_\mu^2-1)}_{=:\,\hat E_\mu}
  \;+\; 2\,\hat m_\mu\,\xi_i^\mu\sigma_i
  \;+\; \mathcal{O}(N^{-1}),
\end{equation}
so that the Hamiltonian decomposes as $\mathcal{H}_{\mathrm{sq}} = -N[\,C_i + \sigma_i\,h_i\,](1+\mathcal{O}(N^{-1}))$ with local field
\begin{equation}
  {\;
  h_i \;=\; \sum_{\mu=1}^{P}\xi_i^\mu\,e^{\hat E_\mu}\,\sinh(2\hat m_\mu),
  \qquad
  \hat E_\mu = N\,(\hat m_\mu^2-1).
  \;}
  \label{eq:hi_sq}
\end{equation}
The zero-temperature Glauber update is $\sigma_i(t+1) = \mathrm{sign}[h_i(t)]$, identical in structure to \eqref{eq:update} for the hetero-associative network (specialised to $L=1$ with the replacement $\sinh(F_\mu^a)\to\sinh(2\hat m_\mu)$).

\subsection*{Signal-to-noise analysis}
\label{sec:sq_sn}

We test stability of the recalled ground state $\bm\sigma=\bm\xi^1$ by computing the moments of $X_i:=\xi_i^1 h_i\big|_{\bm\sigma=\bm\xi^1}$.

\paragraph{First moment}
The self-signal ($\mu=1$) is deterministic at leading order, $\hat m_1 = (N-1)/N$, giving $\hat E_1 = -2+\mathcal{O}(N^{-1})$. The noise patterns ($\mu\neq 1$) have zero mean by Rademacher independence at site $i$.  Hence
\begin{equation}
  \mu_1 \;=\; \E[X_i]
  \;=\; e^{-2}\sinh(2) + \mathcal{O}(N^{-1})
  .
  \label{eq:mu1_sq}
\end{equation}

\paragraph{Second moment}
Off-diagonal contributions ($\mu\neq\nu$) vanish by the same Rademacher-independence argument as in~\S\ref{sec:stability}.  The diagonal noise contribution ($\mu\neq 1$) reduces to the cavity expectation
\begin{equation}
  \E\bigl[(X_i^{(\mu)})^2\bigr]
  \;=\;
  \E\!\bigl[\,e^{2N(\hat m_\mu^2-1)}\,\sinh^2(2\hat m_\mu)\,\bigr],
  \qquad \mu\neq 1.
  \label{eq:noise_sq_app}
\end{equation}
In contrast to the linear model, where the cavity exponent $N(\hat m_\mu-1)$ is \emph{linear} in the Rademacher variables (and the expectation factorises into a closed-form product $\cosh(2)^{N-1}$), here the exponent $N(\hat m_\mu^2-1)$ is \emph{quadratic} and the expectation no longer factorises.  A genuine saddle-point treatment is required.

\subsection*{Noise rate and prefactor}
\label{sec:sq_saddle}

Following \ref{app:saddle}, by Cram\'er's theorem, $\hat m_\mu$ satisfies a large-deviation principle with rate function $I_R(m)$.  Varadhan's lemma identifies the leading exponential decay of~\eqref{eq:noise_sq_app} as a one-dimensional variational problem.  The functional $g(m) := 2m^2 - I_R(m)$ is even in $m$; its unique positive maximiser $m^*$ satisfies
\begin{equation}
  m^* = \tanh(4\,m^*),
  \label{eq:saddle_sq}
\end{equation}
or equivalently, with $x^*=4m^*$,
\begin{equation}
  \tanh(x^*) = \frac{x^*}{4}.
  \label{eq:saddle_sq_x}
\end{equation}

A Laplace evaluation around the saddle yields
\begin{equation}
  \E\bigl[(X_i^{(\mu)})^2\bigr]
  \;=\;
  K\,e^{-N\rho}\bigl(1+o(1)\bigr),
  \label{eq:Xmu_sq_decay}
\end{equation}
with the noise rate
\begin{equation}
  {\;
  \rho \;=\; 2 - \phi(x^*),
  \qquad
  \phi(x) = -\frac{x^2}{8}+\log\cosh(x),
  \;}
  \label{eq:rho_sq}
\end{equation}
and the polynomial prefactor
\begin{equation}
  K \;=\; \frac{\sinh^2(2m^*)}{\sqrt{1-4(1-m^{*2})}}.
  \label{eq:K_sq}
\end{equation}
Numerically, the saddle data are collected in Table~\ref{tab:sq_saddle}.

\begin{table}[h]
  \centering
  \begin{tabular}{ccccc}
    \toprule
    $x^*$ & $m^*=x^*/4$ & $\phi(x^*)$ & $\rho=2-\phi(x^*)$ & $K$ \\
    \midrule
    $3.9973$ & $0.99933$ & $1.3072$ & $0.6928$ & $13.15$ \\
    \bottomrule
  \end{tabular}
  \caption{Saddle data, noise rate and prefactor for the squared-overlap model.}
  \label{tab:sq_saddle}
\end{table}

\subsection*{Storage capacity}
\label{sec:sq_storage}

Combining the signal \eqref{eq:mu1_sq} and the variance $\sigma^2=(P-1)\,K\,e^{-N\rho}$, the Mattis magnetisation after one parallel update reads
\begin{equation}
  m_1^{(1)}
  \;=\;
  \mathrm{erf}\!\left(\frac{e^{-2}\sinh(2)}{\sqrt{2(P-1)\,K\,e^{-N\rho}}}\right),
  \label{eq:m1step_sq}
\end{equation}
which tends to unity as long as $P\,e^{-N\rho}\to 0$.

The qualitative condition $P\,e^{-N\rho}\to 0$ can be turned into a quantitative bound on the load, mirroring the derivation of \eqref{eq:Pmax} for the hetero-associative network. Within the Gaussian approximation, the single-spin stability condition $X_i>0$ holds with probability
\begin{equation}
  \P\bigl(X_i>0\bigr)
  \;=\;
  1-\tfrac12\,\mathrm{erfc}\!\Bigl(\tfrac{\mu_1}{\sqrt{2\sigma^{2}}}\Bigr).
\end{equation}
Requiring a per-spin error probability $\le N^{-a}$, $a>0$, so that the union bound over the $N$ sites remains summable in the thermodynamic limit, gives
\begin{equation}
  \frac{\mu_1^{2}}{2\sigma^{2}} \;\ge\; a\log N+\mathcal{O}(\log\log N),
\end{equation}
and substituting the signal \eqref{eq:mu1_sq} and the variance $\sigma^{2}=(P-1)\,K\,e^{-N\rho}$,
\begin{equation}
  {\;
  P \;\le\; 1+\frac{e^{-4}\sinh^{2}(2)}{2\,a\,K\,\log N}\;e^{N\rho}.
  \;}
  \label{eq:Pmax_sq}
\end{equation}
The leading-order storage capacity is therefore
\begin{equation}
  P_c \;\sim\; e^{N\rho},
  \qquad \rho \approx 0.6928.
  \label{eq:Pc_sq}
\end{equation}

\paragraph{Consistency with the configuration-space ceiling $2^{N}$}
Since the network possesses only $2^{N}$ distinct configurations, any sensible storage estimate must satisfy $P<2^{N}$. The bound \eqref{eq:Pmax_sq} does, for every $N$, for two concurrent reasons. First, the rate lies strictly below $\log 2$: since $\phi(x^{*})=\max_{x}\phi(x)\ge\phi(4)$ and $\phi(4)=-2+\log\cosh(4)=2-\log 2+\log\bigl(1+e^{-8}\bigr)$,
\begin{equation}
  \rho \;=\; 2-\phi(x^{*}) \;\le\; \log\frac{2}{1+e^{-8}} \;<\; \log 2,
  \label{eq:rho_ceiling}
\end{equation}
numerically $\rho=0.692811<\log 2=0.693147$, so that $e^{N\rho}/2^{N}=e^{-N(\log 2-\rho)}\to 0$, albeit slowly ($\log 2-\rho\simeq 3.4\times 10^{-4}$). The ceiling \eqref{eq:rho_ceiling} is the squared-overlap counterpart of the linear-model rate $\rho_{\mathrm{lin}}=\log[2/(1+e^{-4})]$: the $\mathbb{Z}_2$-symmetric exponent doubles the argument of the exponentially small correction, $e^{-4}\to e^{-8}$, pushing the rate closer to the absolute bound $\log 2$. Second, the prefactor in \eqref{eq:Pmax_sq} is itself small, $e^{-4}\sinh^{2}(2)/(2aK\log N)\simeq 9.2\times 10^{-3}/(a\log N)$. Table~\ref{tab:Pmax2N} reports the numerical comparison at $a=1$: the estimate \eqref{eq:Pmax_sq} stays more than two orders of magnitude below $2^{N}$ over the whole range of sizes, with a gap that widens as $N$ grows. The same conclusion holds {a fortiori} if the leading-Laplace prefactor $K$ is replaced by the finite-$N$-corrected value $\widehat K\approx 0.97$ of Remark~\ref{rem:prefactors} below, which raises the bound by roughly one decade but leaves it well under the ceiling.

\begin{table}[H]
\centering
\begin{tabular}{c|ccc}
\toprule
$N$ & $\log_{10}P_{\max}$ from \eqref{eq:Pmax_sq} & $\log_{10}2^{N}$ & $\log_{10}\bigl(P_{\max}/2^{N}\bigr)$ \\
\midrule
10   & 0.61   & 3.01   & $-2.40$ \\
20   & 3.50   & 6.02   & $-2.52$ \\
50   & 12.41  & 15.05  & $-2.64$ \\
100  & 27.39  & 30.10  & $-2.72$ \\
200  & 57.41  & 60.21  & $-2.79$ \\
500  & 147.61 & 150.52 & $-2.91$ \\
1000 & 298.01 & 301.03 & $-3.02$ \\
\bottomrule
\end{tabular}
\caption{Storage estimate \eqref{eq:Pmax_sq} (at $a=1$) against the configuration-space ceiling $2^{N}$ for several system sizes: the bound is below $2^{N}$ at every $N$, as required.}
\label{tab:Pmax2N}
\end{table}

\subsection*{Comparison with the linear-exponent model}
\label{sec:sq_comparison}

The linear-exponent model \eqref{eq:H_lin_app} has a closed-form noise rate $\rho_{\mathrm{lin}} = \log\bigl[\frac{2}{1+e^{-4}}\bigr] \approx 0.6750$, obtained without saddle-point analysis since the cavity exponent is linear.

Two key differences emerge:
\begin{enumerate}
  \item \emph{Analytic structure.} The linear model factorises over sites ($\E[e^{2N\hat m_\mu}]=\cosh(2)^{N-1}$); the squared model requires a genuine saddle-point integral.
  \item \emph{$\mathbb{Z}_2$ symmetry and capacity.}  The squared-overlap rate is the larger of the two, $\rho_{\mathrm{sq}} \approx 0.6928 > \rho_{\mathrm{lin}} \approx 0.6750$, and lies closer to the ceiling $\log 2 \approx 0.6931$: penalising $1-m_\mu^2$ rather than $1-m_\mu$ rewards alignment with either sign of the pattern and thereby tightens the noise suppression.
\end{enumerate}

Figure~\ref{fig:comparison} displays the theoretical predictions \eqref{eq:m1step_sq} and the corresponding linear-model formula from~\cite{ALBANESE2026131223} at $N=10$, with Monte-Carlo markers from the zero-temperature one-step dynamics of both models 
. The two models are nearly indistinguishable for $P \lesssim 5\times 10^{2}$, where both achieve perfect recall ($m_1^{(1)}\approx 1$). As $P$ approaches the critical storage, the curves separate: the larger noise rate of the squared model is offset at finite $N$ by its larger prefactor $K\approx 13.15$, placing the effective transition at somewhat lower $P$. In the thermodynamic limit, however, the exponential rate $\rho_{\mathrm{sq}} > \rho_{\mathrm{lin}}$ implies that the squared model stores strictly more patterns at leading order (panel~(b)).

\begin{figure}[t]
  \centering
  \includegraphics[width=\linewidth]{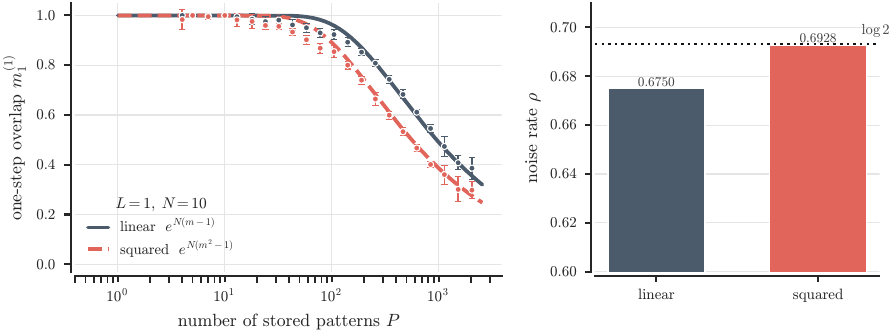}
  \caption{\textbf{Squared-overlap versus linear-exponent model ($L=1$).}
  \textbf{(a)} One-step magnetisation $m_1^{(1)}$ versus load $P$ at $N=10$ for
  the linear-exponent model $\mathcal{H}_{\mathrm{lin}}=-N\sum_\mu e^{N(m_\mu-1)}$
  (solid) and the squared-overlap model
  $\mathcal{H}_{\mathrm{sq}}=-N\sum_\mu e^{N(m_\mu^{2}-1)}$ (dashed). The two
  coincide up to $P\!\sim\!5\times 10^{2}$; at finite $N$ the squared model's larger
  prefactor $K\approx13.15$ moves its knee to slightly lower $P$ despite its
  larger rate. Markers are Monte-Carlo (mean\,$\pm$\,std over seeds), from the
  one-step dynamics of each model in the collision-free regime $P\ll2^{N}$.
  \textbf{(b)} Noise rates against the absolute ceiling $\log2$: the
  $\mathbb{Z}_2$-symmetric squared exponent ($\rho\approx0.6928$) sits closer to
  the ceiling than the linear one ($\rho\approx0.6750$).}
  \label{fig:comparison}
\end{figure}

\subsection*{Basins of attraction}
\label{sec:sq_basins}

Following the protocol of \cite{ALBANESE2026131223}, we now quantify the basins of attraction of the stored patterns by feeding the zero-temperature dynamics with a corrupted version of the archetype. The initial configuration is $\sigma_j^{(0)}=\tilde\xi_j^{1}:=s_j\,\xi_j^{1}$, where the corruption masks $s_j\in\{-1,+1\}$ are i.i.d.\ with $\E[s_j]=r\in(0,1]$, so that the initial Mattis overlap is $m_1^{(0)}=r$ and the Hamming distance per neuron is $d=(1-r)/2$. The quantity controlling one parallel update of spin $i$ is again $X_i=\xi_i^{1}h_i\big|_{\bm\sigma=\tilde{\bm\xi}^{1}}$, now with cavity magnetisations
\begin{equation}
  \hat m_1=\frac{1}{N}\sum_{j\neq i}s_j,
  \qquad
  \hat m_\mu=\frac{1}{N}\sum_{j\neq i}\xi_j^{\mu}\,\tilde\xi_j^{1}
  \quad (\mu\neq 1),
\end{equation}
and, within the Gaussian (CLT) approximation and the large-$N$ replacement of the empirical site average by an expectation, exactly as in \cite{ALBANESE2026131223}, the post-update magnetisation is $m_1^{(1)}=\mathrm{erf}\bigl[\mu_1(r)/\sqrt{2(\mu_2-\mu_1^{2})}\bigr]$.

\paragraph{The noise channel is insensitive to the corruption}
For $\mu\neq 1$ the products $u_j:=\xi_j^{\mu}\tilde\xi_j^{1}$ are i.i.d.\ symmetric Rademacher variables for every value of $r$, because $\xi_j^{\mu}$ is a symmetric sign independent of $\tilde\xi_j^{1}$. The entire noise computation leading to \eqref{eq:Xmu_sq_decay} therefore carries over verbatim: each of the $P-1$ noise patterns contributes zero mean and per-pattern variance $K\,e^{-N\rho}$, with $\rho$, $K$ given by \eqref{eq:rho_sq}--\eqref{eq:K_sq}, the off-diagonal contributions vanish by on-site independence, and
\begin{equation}
  \sigma^{2}\;=\;\mu_2-\mu_1^{2}\;\simeq\;(P-1)\,K\,e^{-N\rho},
  \qquad\text{independently of } r.
  \label{eq:noise_r_sq}
\end{equation}
(As in \cite{ALBANESE2026131223}, the fluctuations of the signal term itself are dropped from $\mu_2-\mu_1^{2}$; see Remark~\ref{rem:annealed_typical} for the caveats attached to this step.)

\paragraph{The corrupted signal requires a tilted saddle point}
The signal term reads
\begin{equation}
  X_i^{(1)}
  \;=\;
  e^{N(\hat m_1^{2}-1)}\,\sinh(2\hat m_1),
  \qquad
  \hat m_1=\frac{1}{N}\sum_{j\neq i}s_j .
\end{equation}
In the linear-exponent model the corresponding average factorises exactly over sites, $\E\bigl[e^{\sum_{j\neq i}s_j}\bigr]=\bigl(\cosh 1+r\sinh 1\bigr)^{N-1}$, which produces the closed form used in \cite{ALBANESE2026131223}. Here this step is not possible: the exponent is quadratic in $\hat m_1$ and the expectation does not factorise. The viable substitute, as for the noise evaluation \eqref{eq:noise_sq_app}, is a large-deviation analysis, now with a biased reference measure. By Cram\'er's theorem the empirical mean of the masks satisfies an LDP with the tilted rate function
\begin{equation}
  I_r(m)
  \;=\;
  \frac{1+m}{2}\log\frac{1+m}{1+r}+\frac{1-m}{2}\log\frac{1-m}{1-r}
  \;=\;
  I_R(m)-a_r\,m+\log\cosh(a_r),
  \qquad
  a_r:=\tanh^{-1}(r),
  \label{eq:tilted_rate}
\end{equation}
which vanishes at $m=r$ only, and Varadhan's lemma gives the annealed signal
\begin{equation}
  \mu_1(r)
  \;=\;
  \E\bigl[X_i^{(1)}\bigr]
  \;=\;
  C(r)\,e^{-N\Sigma(r)}\bigl(1+o(1)\bigr),
  \qquad
  \Sigma(r)
  \;:=\;
  1-\sup_{m\in[-1,1]}\bigl[\,m^{2}-I_r(m)\,\bigr].
  \label{eq:mu1_r_sq}
\end{equation}
The supremum is attained at the unique positive solution $m_s=m_s(r)$ of
\begin{equation}
  \tanh^{-1}(m_s)\;=\;2m_s+a_r
  \quad\Longleftrightarrow\quad
  m_s=\tanh\bigl(2m_s+a_r\bigr),
  \label{eq:saddle_basin}
\end{equation}
which generalises the noise saddle \eqref{eq:saddle_sq} (the bias shifts the linear branch, pulling $m_s$ towards $1$ as $r\to 1$), and the same algebra leading to \eqref{eq:rho_sq} yields the explicit rate
\begin{equation}
  {\;
  \Sigma(r)
  \;=\;
  1+m_s^{2}-\log\cosh\bigl(2m_s+a_r\bigr)+\log\cosh(a_r).
  \;}
  \label{eq:Sigma_r}
\end{equation}
The Laplace prefactor at leading order is
\begin{equation}
  C(r)\;=\;\frac{\sinh(2m_s)}{\sqrt{1-2\,(1-m_s^{2})}},
  \label{eq:C_r}
\end{equation}
well defined because $m_s(r)\ge m_s(0)=0.9575>1/\sqrt{2}$ for every $r\ge 0$. Sanity checks: as $r\to 1$ the rate function freezes the saddle at $m_s\to 1$ and $\Sigma(r)\to 0$, recovering the $\mathcal{O}(1)$ signal of the perfect-recall analysis; at $r\to 0^{+}$ the saddle equation reduces to $m_s=\tanh(2m_s)$ and $\Sigma(0)=0.6735$, while $\mu_1(0)=0$ identically by parity (the two saddles $\pm m_s$ cancel in the odd integrand), consistently with the loss of any retrieval drive for an uncorrelated input.

\paragraph{One-step magnetisation and storage under corruption}
Combining \eqref{eq:noise_r_sq} and \eqref{eq:mu1_r_sq}, the Mattis magnetisation after one parallel update reads
\begin{equation}
  m_1^{(1)}
  \;=\;
  \mathrm{erf}\!\left(
  \frac{C(r)\,e^{-N\Sigma(r)}}{\sqrt{2(P-1)\,K\,e^{-N\rho}}}
  \right)
  \label{eq:m1step_basin}
\end{equation}
(with the finite-$N$ prefactor corrections of Remark~\ref{rem:prefactors}, eq.~\eqref{eq:m1step_basin} reduces exactly to \eqref{eq:m1step_sq} as $r\to 1$), and it tends to unity if and only if
\begin{equation}
  P\,e^{-N\varepsilon(r)}\;\to\;0,
  \qquad
  \varepsilon(r)\;:=\;\rho-2\,\Sigma(r).
  \label{eq:eps_r}
\end{equation}
Parametrising the load as $P=\gamma\,\frac{C(r)^{2}}{K}\,e^{N\varepsilon(r)}$, in analogy with the load parametrisation adopted in \cite{ALBANESE2026131223}, eq.~\eqref{eq:m1step_basin} collapses onto the universal profile
\begin{equation}
  m_1^{(1)}\;=\;\mathrm{erf}\!\left(\frac{1}{\sqrt{2\gamma}}\right),
  \label{eq:m1_universal}
\end{equation}
and the union-bound argument leading to \eqref{eq:Pmax_sq} now gives the corruption-dependent storage estimate
\begin{equation}
  {\;
  P \;\le\; 1+\frac{C(r)^{2}}{2\,a\,K\,\log N}\;e^{N\varepsilon(r)},
  \qquad
  P_c(r)\;\sim\;e^{N\varepsilon(r)}.
  \;}
  \label{eq:Pmax_basin}
\end{equation}
The storage capacity therefore remains exponential in $N$ for every corruption level such that $\varepsilon(r)>0$, i.e.
\begin{equation}
  \Sigma(r)<\frac{\rho}{2}\simeq 0.3464
  \quad\Longleftrightarrow\quad
  r>r_c\simeq 0.4032,
  \qquad
  d<d_c=\frac{1-r_c}{2}\simeq 0.2984,
  \label{eq:rc_basin}
\end{equation}
the threshold being obtained by solving $\Sigma(r_c)=\rho/2$ numerically along \eqref{eq:saddle_basin}--\eqref{eq:Sigma_r}. Table~\ref{tab:basin_saddle} collects the saddle data. As in the linear model, enlarging the basins (decreasing $r$) lowers the storage exponent without ever destroying its exponential character above threshold; at $r=1$ one recovers $\varepsilon(1)=\rho$ and the fixed-point results of the previous subsections. For comparison, the same (annealed) criterion applied to the linear-exponent model gives $r_c^{\mathrm{lin}}\simeq 0.3374$ \cite{ALBANESE2026131223}: the squared-overlap model trades slightly narrower basins of attraction for its larger storage rate $\rho_{\mathrm{sq}}>\rho_{\mathrm{lin}}$, in line with the intuition that a sharper energy landscape stores more memories at the price of robustness. Note also that the threshold \eqref{eq:rc_basin} coincides with the $L=3$ value of the hetero-associative analysis of \S\ref{sec:basins} (Table~\ref{tab:basins}): the squared-overlap saddle \eqref{eq:saddle_basin} is the $L=3$ specialisation of the tilted saddle \eqref{eq:saddle-r-app}, and signal and noise rates both rescale by the same factor $L=3$, leaving the threshold $\Sigma(r_c)=\rho/2$ invariant.

\begin{table}[H]
\centering
\begin{tabular}{c|cccc}
\toprule
$r$ & $m_s(r)$ & $\Sigma(r)$ & $\varepsilon(r)=\rho-2\Sigma(r)$ & $C(r)$ \\
\midrule
0.0    & 0.9575 & 0.6735 & $-0.6541$ & 3.636 \\
0.2    & 0.9732 & 0.4980 & $-0.3033$ & 3.628 \\
0.4032 & 0.9836 & 0.3464 & $0$       & 3.625 \\
0.5    & 0.9872 & 0.2814 & $+0.1299$ & 3.625 \\
0.7    & 0.9934 & 0.1592 & $+0.3743$ & 3.626 \\
0.9    & 0.9981 & 0.0503 & $+0.5922$ & 3.626 \\
1.0    & 1      & 0      & $+0.6928$ & --    \\
\bottomrule
\end{tabular}
\caption{Tilted saddle data for the corrupted-input analysis: saddle point $m_s(r)$, signal rate $\Sigma(r)$, storage exponent $\varepsilon(r)$ and leading-Laplace prefactor $C(r)$. The capacity exponent vanishes at $r_c\simeq 0.4032$.}
\label{tab:basin_saddle}
\end{table}

\begin{remark}[Annealed versus typical signal]
\label{rem:annealed_typical}
The estimate \eqref{eq:Pmax_basin} is built on the moments of $X_i$, i.e.\ on the annealed signal $\mu_1(r)=\E[X_i^{(1)}]$, exactly as in the linear-exponent analysis of \cite{ALBANESE2026131223}. For $r<1$, however, this average is dominated by exponentially rare realisations of the corruption masks, those with $\hat m_1\simeq m_s(r)>r$. The typical realisation has instead $\hat m_1=r+\zeta/\sqrt{N}$ with $\zeta\sim\mathcal{N}(0,1-r^{2})$, so that
\begin{equation}
  N(\hat m_1^{2}-1)\;=\;-N(1-r^{2})+2r\zeta\sqrt{N}+\zeta^{2},
  \qquad
  \frac{1}{N}\log X_i^{(1)}
  \;\xrightarrow[N\to\infty]{\mathrm{a.s.}}\;
  -(1-r^{2});
\end{equation}
the $\mathcal{O}(\sqrt{N})$ Gaussian term moreover forbids the assignment of any deterministic $\mathcal{O}(1)$ prefactor to the typical signal. Since $m_s(r)>r$ for every $r<1$, comparing the supremum in \eqref{eq:mu1_r_sq} with the value of the variational functional at $m=r$ gives $\Sigma(r)<1-r^{2}$ strictly: the annealed signal overestimates the typical one at the exponential scale. Replacing the annealed rate $\Sigma(r)$ by the typical rate $1-r^{2}$ in the comparison with the noise yields the more conservative estimates
\begin{equation}
  P_c^{\mathrm{typ}}(r)\;\sim\;e^{N[\rho-2(1-r^{2})]},
  \qquad
  r_c^{\mathrm{typ}}=\sqrt{1-\rho/2}\;\simeq\;0.8085,
  \qquad
  d_c^{\mathrm{typ}}\simeq 0.0958 .
  \label{eq:rc_typ}
\end{equation}
The same dichotomy is present, though not discussed, in the linear-exponent model: the factorised $\mu_1(r)$ of \cite{ALBANESE2026131223} is also an annealed average, with rate $\log(\cosh 1+r\sinh 1)-1$ strictly above the typical rate $-(1-r)$ for $r<1$; the typical criterion would give there $r_c^{\mathrm{typ,lin}}=1-\rho_{\mathrm{lin}}/2\simeq 0.6625$. Which of the two thresholds, \eqref{eq:rc_basin} or \eqref{eq:rc_typ}, is realised by the finite-$N$ dynamics is best settled by MCMC simulations; the two estimates coincide at $r=1$, where the basin analysis reduces to the fixed-point stability analysis above.
\end{remark}

\begin{remark}[Finite-$N$ prefactors]
\label{rem:prefactors}
Throughout this appendix the prefactors are evaluated at leading Laplace order, treating the cavity magnetisation as the empirical mean of $N$ (rather than $N-1$) variables. Restoring the $j\neq i$ exclusion and, for the noise, the contribution of both symmetric saddles $\pm m^{*}$ of the even integrand \eqref{eq:noise_sq_app}, the corrected prefactors read
\begin{equation}
  \widehat C(r)\;=\;C(r)\,e^{\Sigma(r)-1-m_s^{2}},
  \qquad
  \widehat K\;=\;2\,K\,e^{\rho-2(1+m^{*2})}\;\approx\;0.966 .
  \label{eq:hat_prefactors}
\end{equation}
We verified both expressions against the exact binomial enumeration of the cavity magnetisation: at $N=1600$ the enumerated per-pattern noise variance equals $0.9646\,e^{-N\rho}$ and the enumerated signal at $r=0.5$ equals $0.6663\,e^{-N\Sigma(0.5)}$, to be compared with $\widehat K=0.9659$ and $\widehat C(0.5)=0.6668$, with residual $\mathcal{O}(N^{-1})$ corrections. Consistently, $\widehat C(r)\to e^{-2}\sinh(2)$ as $r\to 1$, matching the exact perfect-recall signal \eqref{eq:mu1_sq}. These $\mathcal{O}(1)$ factors affect only the prefactors of the storage estimates \eqref{eq:Pmax_sq} and \eqref{eq:Pmax_basin}, never the rates $\rho$, $\Sigma(r)$ nor the thresholds \eqref{eq:rc_basin}--\eqref{eq:rc_typ}; they are however essential for quantitative comparisons at moderate $N$, such as the finite-$N$ storage curves of Figure~\ref{fig:comparison}.
\end{remark}


\section{The price of exponential capacity}
\label{app:cost}

Section~\ref{sec:dynamics} showed that one parallel sweep of
Algorithm~\ref{algo:hetero} costs $\Theta(NLP)$ in both time and memory: to
leading order, a single read of the $L$ stored datasets. That cost is linear in
the number of stored patterns $P$, and would be unremarkable were $P$ moderate.
It is not. This appendix spells out what the exponential capacity
$P_c\sim e^{N\rho_L}$ of Section~\ref{sec:storage} implies for anyone who would
try to use it, and in what precise sense the capacity theorem is a
statement about a limit rather than about a machine.

\paragraph{The cost inherits the storage rate}
Fix a load fraction $\gamma=P/P_c$ and let the network run at it. The per-sweep
cost is then
\begin{equation}
  \Theta(NLP)\;=\;\Theta\!\bigl(\gamma\,NL\,e^{N\rho_L}\bigr),
  \label{eq:cost-exp}
\end{equation}
exponential in $N$ with exactly the rate $\rho_L$ that measures the capacity. The
two are the same number wearing two hats: every bit of storage the network gains
by increasing $N$ (or, through $\rho_L\sim L\log2$, by adding layers) is paid for,
one-to-one in the exponent, by the cost of writing the data down and sweeping it
once. Width is in this sense not a free lunch: the same factor $e^{N\rho_L}$ that
buys exponentially more memories as $L$ grows multiplies the cost of touching them
by the identical amount.

\paragraph{A concrete instance}
Take the $N=64$, $L=2$ network already flagged in Section~\ref{sec:storage} as
out of computational reach, and follow the arithmetic. With $\rho_2=1.3470$
(Table~\ref{tab:saddle}) the capacity is
\begin{equation}
  P_c\;\approx\;e^{N\rho_2}\;=\;e^{86.2}\;\approx\;2.8\times10^{37}
\end{equation}
patterns. Merely storing the two datasets, at one byte per spin, already
requires
\begin{equation}
  NLP_c\;\approx\;64\cdot2\cdot2.8\times10^{37}\;\approx\;3.5\times10^{39}\ \text{bytes},
\end{equation}
some $6\times10^{15}$ times Avogadro's number of bytes, or $3.5\times10^{15}$
yottabytes: several SI prefixes past any storage medium that exists or
plausibly ever will. Running the modest $n_{\mathrm{steps}}=5$ sweeps used
throughout this paper (\ref{app:repro}) at that load would take
$\Theta(n_{\mathrm{steps}}NLP_c)\approx1.8\times10^{40}$ elementary operations;
a hypothetical exaFLOP/s machine ($10^{18}$ operations per second, a rate no
single computer has yet sustained) would grind at it for
$\approx1.8\times10^{22}$ seconds, about $4\times10^{4}$ times the current age of
the universe ($\approx4.4\times10^{17}$~s). And $N=64$ is small: a network wide
enough to be biologically interesting is exponentially further out of reach still.

\paragraph{Why the experiments live at $N\sim10$}
By contrast, every experiment reported here keeps $N\le128$ with $P$ several
orders of magnitude below its own $P_c$ (Tables~\ref{tab:repro-hmm},
\ref{tab:repro-vdjdb}); at those sizes $NLP$ never exceeds $\sim10^{6}$ per
sweep: seconds of work on a laptop. This is not merely a convenient choice.
The storage transition sits at $P\sim e^{N\rho_L}$, so it enters an observable
window {only} at small $N$: widen the network and the transition marches off
to loads no simulation can reach, leaving the memory in perfect-recall for every
$P$ one can actually store. Small $N$ is where the physics is visible at all.

\paragraph{A statement about a limit, not a machine}
The tension is worth stating plainly, because it is generic to the whole
exponential family: Demircigil et al.'s binary model, Ramsauer et al.'s modern
Hopfield network, the present hetero-associative construction all share it. ``The
network stores $e^{N\rho_L}$ patterns'' is an exact statement about the fixed
points of the $N\to\infty$ theory: for every finite $N$ the recalled archetype is
provably stable up to that many patterns. It is not, and structurally cannot be,
a statement about any device that simultaneously holds that many patterns in
memory, because no such device fits in the universe once $N$ is more than a few
dozen. Exponential capacity is therefore best read as a statement about the
{shape} of the energy landscape, how many well-separated minima the
construction admits in principle, and not as a promise of a usable database.
What the finite-$N$ experiments certify is the complementary, and physically
operative, half: at the sizes a machine can occupy, the landscape has exactly the
minima the theory predicts, with the basins the theory predicts, and the memory
behaves accordingly.

\section{Hidden Manifold Model: construction and protocols}
\label{app:hmm_protocol}

This appendix specifies the generator of Section~\ref{sec:hmm}
and the four retrieval protocols run on it. The construction is the Hidden Manifold Model of Goldt et
al.~\cite{Goldt2020} in its multilayer hetero-associative form; the
generalisation analysis follows the random-features tradition of Gerace et
al.~\cite{Gerace2020}. Every reported quantity is a mean $\pm$ one standard
deviation over $N_{\mathrm{seeds}}$ independent dataset re-draws, each per-seed
value being itself an average over $N_{\mathrm{eval}}$ retrieval trials.

\paragraph{The two indices $P$ and $K$, and why both are needed}
The model stores a genuinely many-to-one map, and this forces two independent
counts that the notation keeps deliberately separate.
\begin{itemize}
\item $P$ defines the \emph{load}: the number of pattern indices $\mu=1,\dots,P$
actually stored, i.e.\ the number of cues the network holds, exactly as in
the Rademacher theory of Sections~\ref{sec:model}--\ref{sec:storage}. Each index
$\mu$ owns its own latent code $z^{\mu}$ and hence its own set of $L$ layer
patterns $\{\xi^{\mu,a}\}_{a=1}^{L}$.
\item $K$ is the number of distinct \emph{target prototypes} (equivalently,
target regions) available in the last layer. It is a property of the rule,
not of the sample: $K=2^{n_{\mathrm{bits}}}$ is fixed before any latent is drawn,
where $n_{\mathrm{bits}}$ is the number of latent coordinates read by the target
rule below.
\end{itemize}
The map is surjective precisely because $P\gg K$: many of the $P$ stored cues
share one of the $K$ targets. The biological reading is immediate: $P$ counts
receptors, $K$ counts epitopes, and $P/K$ is the mean number of receptors
converging on one epitope, the surjective compression of
Figure~\ref{fig:hmm-capacity}(c). Collapsing $P$ and $K$ into a single index would
forbid the very many-to-one structure the model exists to study. A third count,
$n_{\mathrm{seen}}\le K$, enters the generalisation protocol: the number of target
regions actually populated during training.

\paragraph{Latent code and cue layers}
Fix $L$ layers, $N$ neurons per layer, a latent (manifold) dimension $D$ with
aspect ratio $\alpha_D=D/N\in(0,1]$ (always $D\le N$), and $K=2^{n_{\mathrm{bits}}}$
target regions. Each index $\mu=1,\dots,P$ owns a latent code
\begin{equation}
  z^{\mu}\;\sim\;\mathcal{N}(0,I_D)\quad\text{i.i.d.\ across }\mu,\qquad
  z^{\mu}\in\mathbb{R}^{D}.
  \label{eq:hmm-latent}
\end{equation}
Each cue layer $a\in\{1,\dots,L-1\}$ carries a fixed random feature matrix
$F^{a}\in\mathbb{R}^{N\times D}$ with i.i.d.\ $\mathcal{N}(0,1)$ entries, drawn
once per dataset and shared by all $P$ indices, and produces the binary pattern
\begin{equation}
  \xi^{\mu,a}\;=\;\sign\!\bigl(F^{a}z^{\mu}/\sqrt D\bigr)\;\in\;\{-1,+1\}^{N},
  \qquad a=1,\dots,L-1.
  \label{eq:hmm-cue}
\end{equation}
The $1/\sqrt D$ normalisation makes each pre-activation coordinate
$(F^{a}z^{\mu})_i/\sqrt D$ unit-variance, matching the standard HMM form
$\xi=\varphi(Fz/\sqrt D)$ with $\varphi=\sign$. Two cue layers $a\neq b$ at the
same index $\mu$ are correlated only through the shared latent $z^{\mu}$
(their feature matrices $F^{a},F^{b}$ are independent); this shared cause is
exactly the inter-layer correlation that the idealised theory of
Section~\ref{sec:model} forbids and that makes hetero-association possible.
Distinct indices are independent given the $\{F^{a}\}$.

\paragraph{The surjective target and the region map}
The last layer $a=L$ carries the target. Two variants are used.

\emph{Surjective target} (default; the many-to-one rule the model is meant to
store). Partition latent space into $K$ regions by the sign pattern of the first
$n_{\mathrm{bits}}$ latent coordinates,
\begin{equation}
  k(z)\;=\;\sum_{j=0}^{n_{\mathrm{bits}}-1}2^{j}\,\mathbf{1}[z_j>0]
         \;\in\;\{0,1,\dots,K-1\},\qquad K=2^{n_{\mathrm{bits}}},
  \label{eq:region-map}
\end{equation}
so that $k(z)$ reads the first $n_{\mathrm{bits}}$ signs of $z$ as a binary
integer ($\mathbf{1}[\cdot]$ is the indicator). Fix $K$ prototypes
$\rho_0,\dots,\rho_{K-1}$, each an independent Rademacher vector
$\rho_k\in\{-1,+1\}^{N}$ drawn once, and set the target of index $\mu$ to
\begin{equation}
  \xi^{\mu,L}\;=\;\rho_{k(z^{\mu})}.
  \label{eq:surjective}
\end{equation}
Every index whose latent falls in the same region (i.e.\ shares the first
$n_{\mathrm{bits}}$ signs) is mapped to the same stored target: on average
$P/K$ cues per prototype, with a per-region multiplicity that is
$\mathrm{Binomial}(P,1/K)$. The rule discards the continuous magnitudes and every
coordinate $j\ge n_{\mathrm{bits}}$, so within a region the cue $\xi^{\mu,a}$
still varies continuously with $z^{\mu}$ while the target is pinned to $\rho_{k}$,  the geometric content of ``many distinct cues, one target''.

\emph{Symmetric target} (control). The last layer is instead given its own
independent feature matrix and no repetition,
$\xi^{\mu,L}=\sign(F^{L}z^{\mu}/\sqrt D)$, so each index owns a distinct target.
This is the ensemble used whenever we compare directly against the i.i.d.\
closed-form theory of Sections~\ref{sec:storage}--\ref{sec:basins}, whose
assumption of $P$ essentially-orthogonal targets a fixed $K\ll P$ would badly
violate.

\paragraph{The manifold signature: pairwise overlap and the arcsine law}
The order parameter that certifies the manifold is the \emph{pairwise pattern
overlap}, defined for two indices $\mu\neq\nu$ in a cue layer $a$ exactly as the
Mattis overlap of two stored patterns,
\begin{equation}
  m_{\mu\nu}\;:=\;\frac1N\sum_{i=1}^{N}\xi_i^{\mu,a}\,\xi_i^{\nu,a}\;\in\;[-1,1],
  \label{eq:mmunu}
\end{equation}
i.e.\ the normalised inner product (cosine on the hypercube) between the two
binary codes: $m_{\mu\nu}=1$ for identical codes, $0$ for orthogonal ones. This
is the quantity plotted in Figure~\ref{fig:hmm-manifold}(a), and the reference
$m_{\mu\nu}$ used throughout Section~\ref{sec:hmm}; we now derive its law.

Write $\xi_i^{\mu,a}=\sign(u_i^{\mu})$ with
$u_i^{\mu}=(F^{a}z^{\mu})_i/\sqrt D=\langle f_i,z^{\mu}\rangle/\sqrt D$, where
$f_i\in\mathbb{R}^{D}$ is the $i$-th row of $F^{a}$, i.i.d.\ $\mathcal N(0,I_D)$.
For fixed $z^{\mu},z^{\nu}$ and over the randomness of $f_i$, the pair
$(u_i^{\mu},u_i^{\nu})$ is jointly centred Gaussian with
\begin{equation}
  \mathrm{Var}(u_i^{\mu})=\frac{\|z^{\mu}\|^{2}}{D},\qquad
  \mathrm{Cov}(u_i^{\mu},u_i^{\nu})=\frac{\langle z^{\mu},z^{\nu}\rangle}{D},\qquad
  \mathrm{corr}(u_i^{\mu},u_i^{\nu})=\frac{\langle z^{\mu},z^{\nu}\rangle}
       {\|z^{\mu}\|\,\|z^{\nu}\|}=:\rho_z.
\end{equation}
By Grothendieck's identity for the sign of correlated Gaussians (equivalently the
orthant / arcsine formula), $\E[\sign(X)\sign(Y)]=\tfrac{2}{\pi}\arcsin\rho$ for
standard jointly Gaussian $(X,Y)$ of correlation $\rho$; hence
$\E[\xi_i^{\mu,a}\xi_i^{\nu,a}\mid z^{\mu},z^{\nu}]=\tfrac2\pi\arcsin\rho_z$ for
every $i$. Averaging over the $N$ i.i.d.\ rows and using concentration,
\begin{equation}
  \E[m_{\mu\nu}]\;=\;\frac{2}{\pi}\arcsin(\rho_z),\qquad
  \rho_z=\frac{\langle z^{\mu},z^{\nu}\rangle}{\|z^{\mu}\|\,\|z^{\nu}\|},
  \label{eq:arcsine}
\end{equation}
an exact, parameter-free curve: the ``signature'' verified in
Figure~\ref{fig:hmm-manifold}(a). Fluctuations of $m_{\mu\nu}$ about
\eqref{eq:arcsine} at fixed $\rho_z$ are $\mathcal{O}(N^{-1/2})$.

\emph{The cost of curvature.} For independent latents $\rho_z$ is itself random:
$\langle z^{\mu},z^{\nu}\rangle\sim\mathcal N(0,D)$ and $\|z\|^{2}\approx D$, so
$\rho_z$ has zero mean and standard deviation $D^{-1/2}$. Linearising
\eqref{eq:arcsine} for small $\rho_z$, $m_{\mu\nu}\approx\tfrac2\pi\rho_z$,
\begin{equation}
  \E[m_{\mu\nu}]\approx0,\qquad
  \mathrm{Var}(m_{\mu\nu})\;\approx\;\Bigl(\tfrac{2}{\pi}\Bigr)^{2}\frac1D,
  \label{eq:curv-cost}
\end{equation}
to be compared with the i.i.d.\ Rademacher value $\mathrm{Var}(m_{\mu\nu})=1/N$. The
manifold patterns are therefore more correlated than independent ones
whenever
\begin{equation}
  \Bigl(\tfrac2\pi\Bigr)^{2}\frac1D>\frac1N
  \quad\Longleftrightarrow\quad
  D<\Bigl(\tfrac2\pi\Bigr)^{2}N\approx0.405\,N,
\end{equation}
i.e.\ once the manifold is small enough. This excess correlation is precisely the
extra noise the clean theory does not carry: it adds to the $e^{-N\rho_L}$ floor
and lowers capacity as $\alpha_D=D/N$ falls, the mechanism behind the capacity
drop in Figure~\ref{fig:hmm-capacity}(a). At $\alpha_D=1$ the excess
$(2/\pi)^{2}-1<0$ vanishes with room to spare: the sign-map ensemble is then
less correlated than Rademacher and sits at the near-i.i.d.\ edge used as a
baseline.

\paragraph{The collision subtlety}
The map $z\mapsto\sign(F^{a}z)$ is piecewise constant: it depends on $z$ only
through which side of each hyperplane $\{f_i\cdot z=0\}$ the latent falls. The $N$
central hyperplanes cut $\mathbb{R}^{D}$ into
\begin{equation}
  C(N,D)\;=\;2\sum_{k=0}^{D-1}\binom{N-1}{k}\;=\;\mathcal{O}\!\bigl(N^{D-1}\bigr)
  \label{eq:cover}
\end{equation}
regions (Cover's function-counting theorem for a central hyperplane
arrangement), so at most $C(N,D)$ distinct sign codes exist in layer $a$.
When $D$ is small this is far fewer than the $P$ latents drawn, and many latents
collide onto identical codes. Two consequences must be handled honestly.
(i) A pattern stored in $c$ identical copies is reinforced as if it carried weight
$c$, so a naive one-step-stability score is spuriously high at small $D$,
even as the number of distinct memories, the real capacity,
collapses. We therefore always report the distinct-pattern fraction alongside
stability (Figure~\ref{fig:hmm-manifold}(c)) and read the honest ``capacity falls
with $D$'' statement from the matched-load transition of
Figure~\ref{fig:hmm-capacity}(a), never from raw stability. (ii) Collisions
inflate the measured mean overlap above the arcsine prediction at small $D$
(identical codes contribute $m_{\mu\nu}=1$), the upward bias visible in
Figure~\ref{fig:hmm-generalization}(c).

\paragraph{Held-out regions and the novel-region control}
The generalisation experiment tests whether the network routes a
never-stored cue to the correct target, and its whole logic rests on a
construction we now spell out. Of the $K$ regions, a subset of size
$n_{\mathrm{seen}}<K$ is declared \texttt{allowed} (``seen''); the remaining
$K-n_{\mathrm{seen}}$ are held out. During training the $P$ latents are
drawn by rejection sampling: any $z^{\mu}$ with $k(z^{\mu})$ held-out is discarded
and redrawn, so every stored cue maps to a seen region and the held-out
prototypes $\{\rho_k:k\text{ held out}\}$, though they exist as vectors, are
backed by no stored cue. Four rates are measured on the same trained
network:
\begin{itemize}
\item \emph{Memorisation}: fraction of the $P$ stored cues that recall
their own target under the clamped-cue dynamics.
\item \emph{Generalisation}: fraction of fresh latents drawn from
seen regions whose cue recalls that region's (correctly stored) prototype.
\item \emph{Novel-region control}: the same for fresh latents drawn from
held-out regions, whose prototype was never stored.
\item \emph{Chance}: $1/n_{\mathrm{seen}}$, the rate of guessing the target
uniformly among the seen prototypes ($=1/6\approx0.167$ for $K=8$, two held out).
\end{itemize}
The control is what makes the test conclusive. A held-out prototype is not stored,
so no property of the trained network can point a fresh cue at it except by
artefact (a coincidental collision in the encoder, a leak in the split, a bug). If
novel-region recall exceeded chance the ``generalisation'' signal would be
suspect; that it stays pinned at chance (Figure~\ref{fig:hmm-generalization}(a,b))
while seen-region generalisation sits significantly above it certifies that the
latter is genuine exploitation of manifold geometry, not an encoding artefact.
\emph{Why it works at all}: a fresh cue from a seen region shares the first
$n_{\mathrm{bits}}$ latent signs, hence a large block of sign structure, with the
stored cues of that region; those stored cues carve an energy basin around the
shared prototype $\rho_k$, and a fresh cue landing inside it flows to $\rho_k$.
Denser sampling (larger coverage $P/n_{\mathrm{seen}}$) tiles more of each
region's cue-manifold with stored cues, so a larger fraction of fresh cues fall
into the right basin, which is why generalisation climbs with coverage toward
memorisation (the ``manifold-as-attractor'' limit) in
Figure~\ref{fig:hmm-generalization}(b), while never certifying the network as a
classifier of truly novel regions.

\paragraph{Retrieval protocols and the reduced load}
Three protocols share one dynamics engine. \emph{Auto-associative capacity}
(Figure~\ref{fig:hmm-capacity}(a)) uses one-step stability at $r=1$: all layers
initialised at a stored pattern, unclamped dynamics. \emph{Hetero-associative
recall} (Figure~\ref{fig:hmm-capacity}(b)) clamps the $L-1$ cue layers at their
stored values and cleans the target layer from a cue-driven seed. \emph{Basins}
(Figure~\ref{fig:basins}(a)) corrupt all layers to a common initial overlap $r$.
It is convenient to summarise ``how close to capacity'' by the reduced load of
Eq.~\eqref{eq:gamma-main},
\begin{equation}
  \gamma(P,N,L)\;=\;\frac{(P-1)\,K_L\,e^{-N\rho_L}}{\mu_1(L)^{2}},
  \qquad \mu_1(L)=e^{-(L-1)}\sinh(L-1),
  \label{eq:gamma}
\end{equation}
with $K_L,\rho_L$ the noise prefactor and rate of Section~\ref{sec:stability} and
$\mu_1(L)$ the clean one-step signal~\eqref{eq:mu1}, for which the asymptotic
one-step overlap collapses onto the universal curve
$m_1^{(1)}=\erf(1/\sqrt{2\gamma})$: $\gamma\ll1$ is deep retrieval,
$\gamma\approx1$ the transition, $\gamma\gg1$ failure. We scan the number of
stored patterns $P$ on a logarithmic grid and read off $\gamma$ as the
corresponding intensive coordinate. The two baselines are the classical i.i.d.\
Rademacher ensemble (\texttt{iid\_full}, the exact closed-form regime of
Sections~\ref{sec:storage}--\ref{sec:basins}) and the $\alpha_D=1$ edge of the
same $\sign(Fz)$ pipeline, so that a curve-to-curve comparison isolates the effect
of manifold dimension alone.

\section{VDJdb: cleaning and encoding}
\label{app:vdjdb_protocol}

This appendix documents how the raw VDJdb release~\cite{Shugay2018,Bagaev2020}
becomes the $\{-1,+1\}^{N}$ patterns of Section~\ref{sec:vdjdb}. Two principles
govern the pipeline: the stored map must be a genuine function, and the encoding
must be a fixed, deterministic map that introduces no learned representation. Both
choices are standard in the sequence-immunology and locality-sensitive-hashing
literatures, cited below.

\paragraph{Cleaning into a single-valued map}
We start from the full \texttt{vdjdb} export ($137{,}484$ records) and apply a
lean biological clean followed by a function filter. The receptor key is
$X=(V_\alpha,J_\alpha,\mathrm{CDR3}_\alpha,V_\beta,J_\beta,\mathrm{CDR3}_\beta)$
and the target is $Y=\text{epitope}$. The steps and surviving counts are in
Table~\ref{tab:funnel}. The decisive step is the last: deduplicating triples is
not enough to make $X\mapsto Y$ single-valued, because a receptor may appear
against several epitopes; we therefore keep only receptors mapped to exactly one
epitope, dropping the $52$ ambiguous ones. The result is $1{,}052$ clean triples
over $220$ epitopes, strongly surjective (up to $146$ receptors per epitope, $120$
singletons): the ``convergent recognition'' the network is meant to store. A
single-valued map is the mathematical prerequisite of a well-posed
generalisation/surjectivity test.

\begin{table}[h]
\centering
\begin{tabular}{l r}
\toprule
stage & records \\
\midrule
raw rows                                & $137{,}484$ \\
\emph{Homo sapiens}                     & $131{,}574$ \\
curation score $\ge 1$                  & $10{,}178$ \\
drop 10x-Genomics demo                  & $9{,}908$ \\
paired \texttt{complex.id}              & $3{,}032$ \\
valid amino-acid sequence + length      & $3{,}031$ \\
paired $\alpha\beta$ synapses           & $1{,}227$ \\
dedup $(X,\text{epitope})$              & $1{,}227$ \\
function filter (drop $52$ ambiguous $X$) & $\mathbf{1{,}052}$ \\
\bottomrule
\end{tabular}
\caption{Cleaning funnel imposing a single-valued receptor$\to$epitope map.}
\label{tab:funnel}
\end{table}

\paragraph{Encoding, overview}
Each amino-acid sequence becomes a pattern $\xi^{\mu,a}\in\{-1,+1\}^{N}$ (here
$N=128$) in two deterministic, learning-free stages: (A) a positional
Atchley-factor feature map turning the sequence into a fixed-length real vector,
and (B) a locality-sensitive (SimHash) binariser turning that vector into a
balanced, near-orthogonal sign code whose Hamming overlap tracks the cosine angle
of the features. One encoder is fitted per layer ($\alpha$, $\beta$,
epitope), because the three modalities differ in length and composition, and it
is fitted on the training split only: test sequences are transformed with
the frozen maps, so no held-out sequence informs the encoder.

\paragraph{Stage A: positional Atchley-factor features}
Atchley factors~\cite{Atchley2005} summarise each amino acid by five numbers
$(z_1,\dots,z_5)$, obtained from a factor analysis of several hundred
physicochemical amino-acid indices and standardised to zero mean and unit
variance across the $20$ residues; the five axes are, in order, polarity /
hydrophobicity ($z_1$), secondary-structure propensity ($z_2$), molecular size /
volume ($z_3$), codon composition / refractivity ($z_4$) and electrostatic charge
($z_5$). Writing $a(s)\in\mathbb{R}^{5}$ for the factor vector of residue $s$, a
sequence $w=(s_1,\dots,s_{\ell})$ of length $\ell$ is mapped to a fixed
$L_{\max}\times5$ array by centre padding: the N-terminal half of $w$ is
written from the left, the C-terminal half from the right, and the
$L_{\max}-\ell$ empty middle positions are set to zero (the mean of the
standardised factors). This keeps the conserved CDR3 anchors (the N-terminal
cysteine, the C-terminal phenylalanine--glycine) at fixed indices and lets the
padding fall in the hypervariable middle, so positional information is preserved
--a charge at position $3$ is a different feature from a charge at position $10$.
Flattening gives the real feature vector
\begin{equation}
  x^{\mu,a}\;=\;\bigl(a(s_1),\dots\bigr)\in\mathbb{R}^{5L_{\max}},
  \qquad L_{\max}=\min\!\bigl(\lceil q_{99}\rceil,\ \ell_{\max}\bigr),
  \label{eq:atchley}
\end{equation}
with $L_{\max}$ the per-layer alignment length (the $99$th percentile of the
training lengths, capped at the observed maximum $\ell_{\max}$). This replaces the
earlier order-destroying bag-of-$k$-mers summary, retained only as the ablation
baseline below.

\paragraph{Stage B: the locality-sensitive (SimHash) binariser}
The feature vector is binarised by standardise $\to$ PCA-whiten $\to$ Gaussian
random projection $\to$ sign~\cite{Charikar2002}. Let $\Pi:\mathbb{R}^{5L_{\max}}
\to\mathbb{R}^{d'}$ be the affine map that standardises each coordinate and
projects onto the top $d'=\min(50,\text{rank})$ whitened PCA directions (fitted on
the training features, so $\Pi$ has unit-covariance output), and let
$W\in\mathbb{R}^{N\times d'}$ be a fixed matrix with i.i.d.\ $\mathcal{N}(0,1)$
entries. Then
\begin{equation}
  \xi^{\mu,a}\;=\;\sign\!\bigl(W\,\Pi(x^{\mu,a})\bigr)\;\in\;\{-1,+1\}^{N},
  \qquad W_{ki}\sim\mathcal{N}(0,1),
  \label{eq:lsh}
\end{equation}
with ties ($=0$) mapped to $+1$. Because $W$ is Gaussian and $\Pi(x)$ has unit
covariance, each bit is an unbiased $\pm1$ sign, and distinct bits (distinct rows
of $W$) are near-independent: the codes are close to the Rademacher ideal the
theory assumes.

\paragraph{The SimHash law connects the encoding to the manifold signature}
For two feature vectors with whitened images $v^{\mu}=\Pi(x^{\mu,a})$,
$v^{\nu}=\Pi(x^{\nu,a})$ and angle $\theta=\arccos c$,
$c=\langle v^{\mu},v^{\nu}\rangle/(\|v^{\mu}\|\,\|v^{\nu}\|)$, a single Gaussian
hyperplane separates them with probability $\theta/\pi$, so each bit agrees with
probability $1-\theta/\pi$ and
\begin{equation}
  \E[m_{\mu\nu}]\;=\;1-\frac{2}{\pi}\arccos(c)\;=\;\frac{2}{\pi}\arcsin(c),
  \label{eq:simhash-law}
\end{equation}
with $m_{\mu\nu}=\frac1N\sum_i\xi_i^{\mu,a}\xi_i^{\nu,a}$ the pattern overlap of
Eq.~\eqref{eq:mmunu}. This is the same Grothendieck identity as the arcsine
law~\eqref{eq:arcsine} of the Hidden Manifold Model, now with the biophysical
cosine $c$ in place of the latent cosine $\rho_z$: biochemically similar receptors
receive proximate codes, and Figure~\ref{fig:vdjdb-data}(d) shows the measured
overlap tracking~\eqref{eq:simhash-law}. Table~\ref{tab:encoding} confirms the
Rademacher quality: per-bit balance $\approx0.03$ on the receptor layers and mean
absolute pattern overlap $\approx0.10$, close to the i.i.d.\ value
$1/\sqrt N\approx0.088$ at $N=128$. The epitope layer, having only $220$ distinct
codes, is measurably less balanced ($0.12$), as expected for a small surjective
target. Alternative encodings ($k$-mer, random) are retained only as ablation
baselines (Figure~\ref{fig:vdjdb-ablation}): the biophysical Atchley features and
$k$-mers carry comparable signal, while a text-blind random projection generalises
at chance, confirming that the retrieval signal is in the features and not in the
hash.

\begin{table}[h]
\centering
\begin{tabular}{l c c c c}
\toprule
layer & $N$ & distinct codes & per-bit balance & overlap $|m_{\mu\nu}|$ \\
\midrule
$\alpha$-CDR3 & $128$ & $1042$ & $0.028$ & $0.104$ \\
$\beta$-CDR3  & $128$ & $1048$ & $0.028$ & $0.103$ \\
epitope       & $128$ & $220$  & $0.123$ & $0.120$ \\
\bottomrule
\end{tabular}
\caption{Near-Rademacher quality of the Atchley$+$SimHash encoding (means over
seeds). ``Per-bit balance'' is the mean absolute per-bit magnetisation (ideal
$0$); the receptor layers are close to Rademacher.}
\label{tab:encoding}
\end{table}

\begin{figure}[t!]
  \centering
  \includegraphics[width=\linewidth]{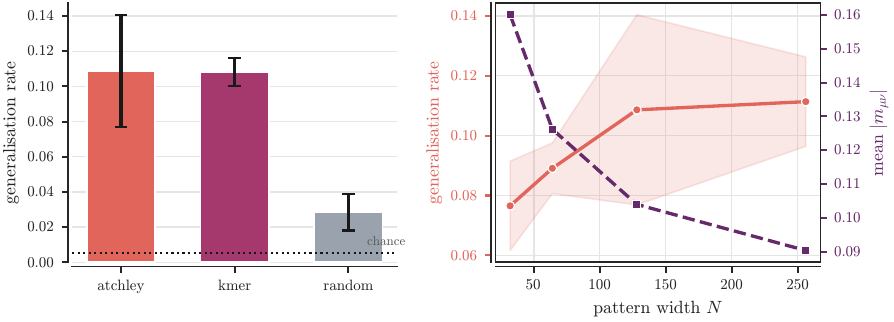}
  \caption{\textbf{Encoding ablation.}
  \textbf{(a)} Generalisation to unseen receptors under the three encodings:
  Atchley factors and $k$-mers carry comparable biological signal (well above
  chance), while a text-blind random projection collapses to chance: the signal
  is in the biophysical features, not the hashing.
  \textbf{(b)} Effect of the layer size $N$ (Atchley): longer codes lower the
  pairwise overlap (right axis) and slowly raise generalisation (left axis).}
  \label{fig:vdjdb-ablation}
\end{figure}

\paragraph{Layer assignment, tasks and reproducibility}
The three layers are $a=1$ ($\alpha$-CDR3), $a=2$ ($\beta$-CDR3), $a=3$ (epitope),
so $L=3$. Hetero-associative tasks clamp the cue layers and recall the target;
generalisation holds out a fraction of the receptors of each epitope and tests
recall on them, against both a chance level $1/n_{\mathrm{epitopes}}$ and a
label-permutation null (epitope labels shuffled before encoding), which is the
control of Figure~\ref{fig:vdjdb-results}(c). All quantities are means $\pm$ one
standard deviation over independent projection seeds and train/test splits; the
dynamics engine is identical to the HMM battery, so the two are directly
comparable.

\paragraph{A caveat on the held-out split}
The receptors held out for generalisation are sampled uniformly at random
within each epitope cluster, after deduplication on the exact
$(V,J,\mathrm{CDR3}_\alpha,V,J,\mathrm{CDR3}_\beta)$ key; the split does not
additionally enforce a minimum sequence distance between held-out and training
receptors of the same epitope. TCR repertoires are known to converge on
near-identical ``public'' sequences for a shared epitope~\cite{GlanvilleDash2017},
so some held-out receptors may sit a few substitutions from a training one,
which would inflate the measured generalisation rate relative to a split
enforcing sequence-level separation. The label-permutation null controls for
structure in the encoding, not for this specific leakage channel; a
similarity-aware split is the natural follow-up and is not expected to remove
the signal (the two chains still recall the epitope from first principles,
Figure~\ref{fig:vdjdb-results}(a)) but could lower its measured size.

\section{The CLINC150 protocol}
\label{app:clinc}

This appendix gives the construction behind Section~\ref{sec:clinc}: the
cleaning, the encoder, the retrieval protocols and the ablations. It is the
natural-language counterpart of \ref{app:vdjdb_protocol}, and deliberately
shares with it every step that can be shared, so that the two batteries differ
only where the data force them to. Every reported number is a mean $\pm$ one
standard deviation over $N_{\mathrm{seeds}}=6$ dataset re-draws (projection seed
plus train/test split), each itself an average over evaluation trials, exactly as
in \ref{app:repro}.

\paragraph{Data and the function filter}
CLINC150~\cite{Larson2019} is a benchmark of short user utterances labelled by
one of $150$ intents (``set an alarm'', ``what is my balance''), balanced at
about $150$ utterances per intent, and shipped with a set of out-of-scope (OOS)
utterances belonging to none of them. We pool the in-domain splits, normalise
and deduplicate $(\text{utterance},\text{intent})$ pairs, and apply the same
function filter as for VDJdb (keep only utterances mapped to a single intent,
as Remark~\ref{rem:surjective} requires) which here removes just $4$
ambiguous utterances, leaving $22{,}491$ clean records over $150$ intents.
The mean compression is $P/K\approx150$ cues per target, against $\approx5$ for
VDJdb: the same surjective structure, sampled thirty times more densely. The OOS
utterances are held aside as a novelty control.

\paragraph{Encoding}
The two layers are $a=1$ (utterance) and $a=2$ (intent). Utterances become
$\{-1,+1\}^{N}$ patterns through a fixed, deterministic pipeline with no learned
representation: a TF--IDF vector over word $(1,2)$-grams concatenated with
character $(3,5)$-grams (lexical and short-phrase content on one side,
morphology and sub-word cues on the other, the linguistic analogue of the
positional Atchley map) reduced by a PCA whitening to at most $50$ components
fitted on the training split, then passed through the same SimHash sign
map~\cite{Charikar2002} used in \ref{app:vdjdb_protocol}. The intent layer
stores one code per intent {name}, so that every phrasing of an intent
shares a single target and semantically close intents receive close codes. As in
the manifold and receptor cases the pattern overlap obeys the arcsine law
$\E[m_{\mu\nu}]=1-\tfrac2\pi\arccos c$ in the feature cosine $c$
(Figure~\ref{fig:clinc}(b)), with per-bit imbalance $0.029\pm0.002$ and mean
absolute overlap $0.097\pm0.002$ on the utterance layer at $N=128$, against the
i.i.d.\ value $1/\sqrt N=0.088$. Choosing $L=2$ makes the closed forms of
Sections~\ref{sec:storage}--\ref{sec:basins} directly usable, with
$\rho_2=1.3470$, $x^\ast=1.9150$, $m^\ast=0.9575$, annealed $r_c=0.3195$ and
typical $r_c=0.5714$.

\paragraph{Retrieval protocols}
Four protocols are run, all with the engine of Algorithm~\ref{algo:hetero}
unchanged from the other two batteries. \emph{Capacity}: a logarithmic grid
$P\in[30,6000]$, one-step overlap of a stored pattern against~\eqref{eq:m1step},
real and i.i.d.\ ensembles at matched $P$. \emph{Basins}: at $P=2000$, a sweep of
the cue overlap $r$ on a $14$-point grid, against~\eqref{eq:m1step-r}.
\emph{Tasks}: at $P=3000$, the forward direction utterance$\to$intent and the
reverse intent$\to$utterance, plus the cue-content scan in which only the leading
fraction $f\in\{0.15,\dots,1\}$ of the utterance's tokens is revealed and
re-encoded with the frozen transform. \emph{Generalisation}: a quarter of the
utterances of each intent held out, then memorisation (stored cues),
generalisation (held-out cues of seen intents), the OOS control and chance
$1/150$, together with top-$k$ retrieval and a label-permutation null over $20$
permutations.

\paragraph{Results, and the comparison with VDJdb}
Table~\ref{tab:clinc} collects the outcome beside the receptor battery. The
memory side coincides in the two datasets and with the i.i.d.\ theory; the
classifier side differs by an order of magnitude. Top-$k$ retrieval of the
correct intent for a fresh utterance rises from $0.527\pm0.009$ at $k=1$ to
$0.600$, $0.629$ and $0.649$ at $k=3,5,10$: the correct intent is usually
either first or not in the shortlist at all, which is the signature of a basin
that either contains the fresh cue or does not.

Two things must be kept apart here. The OOS entry of Table~\ref{tab:clinc} is a
\emph{null control}: an out-of-scope utterance, whose intent is by construction
absent from the codebook, is matched to its (unstored) target at rate $0.014$,
i.e.\ at chance, which is what certifies that the generalisation signal is not
an artefact of the encoding. It is not a rejection capability, and the
distinction matters because the network does not have one: using the retrieval
confidence as an open-set score separates in-domain from out-of-scope utterances
with an AUROC of $0.513$, that is, not at all. An exponential associative memory
recognises what it has stored; it has no built-in notion of ``none of the
above'', and equipping it with one is beyond the present scope.

\begin{table}[H]
\centering
\small
\begin{tabular}{lcc}
\toprule
Quantity & CLINC150 (language) & VDJdb (receptors) \\
\midrule
memorisation                 & $0.72$ & $0.997$ \\
generalisation (unseen cue)  & $0.58$ & $0.11$  \\
novelty / null control       & $0.014$ & $0.048$ \\
chance                       & $0.0067$ & $0.005$ \\
top-$1$ retrieval            & $0.53$ & $0.09$  \\
top-$10$ retrieval           & $0.65$ & --      \\
$(\text{cue})\to\text{target}$ recall & $0.78$ & $0.9994$ \\
mean cues per target $P/K$   & $\approx150$ & $\approx5$ \\
\bottomrule
\end{tabular}
\caption{The two real-data batteries side by side, at $L=2$ and $L=3$
respectively. The memory side (recall, capacity, basins) matches the i.i.d.\
theory in both. The classifier side differs by an order of magnitude: language
generalises far better than receptor sequences.}
\label{tab:clinc}
\end{table}

\paragraph{Encoding ablation: where the generalisation lives}
The storage rule, the dynamics and the closed-form basins are identical across
the two datasets, so the gap in generalisation cannot come from the memory.
Three scans, all at fixed load and fixed protocol, locate it in the encoder
(Table~\ref{tab:clinc-ablation}). Replacing TF--IDF by a text-blind
\emph{random} map, a deterministic hash of the utterance to a fixed Gaussian
vector, so that similar utterances receive unrelated codes, leaves
memorisation at $0.349$ and collapses generalisation onto chance, $0.008$
against $1/150=0.0067$: the network stores as well as ever and has nothing to
interpolate between. A plain bag-of-words map, which discards ordering and
sub-word structure but keeps lexical overlap, is if anything slightly
{better} than TF--IDF ($0.652$ against $0.584$), confirming that what
matters is co-location of same-target cues rather than the sophistication of the
features. Finally, widening the code ($N=32\to256$) or the whitened
representation ($d'=10\to100$) buys generalisation monotonically, in step with
the fall of the mean pattern overlap towards its i.i.d.\ floor $1/\sqrt N$: less
crowded pattern space, wider basins, more of the manifold covered by each stored
cue.

\begin{table}[H]
\centering
\small
\begin{tabular}{llccc}
\toprule
scan & setting & memorisation & generalisation & $\langle|m_{\mu\nu}|\rangle$ \\
\midrule
\multirow{3}{*}{features}
 & TF--IDF        & $0.70$ & $0.58$ & $0.097$ \\
 & bag of words   & $0.75$ & $0.65$ & $0.104$ \\
 & random (blind) & $0.35$ & $0.008$ & $0.101$ \\
\midrule
\multirow{4}{*}{layer size $N$}
 & $32$  & $0.44$ & $0.35$ & $0.158$ \\
 & $64$  & $0.62$ & $0.50$ & $0.123$ \\
 & $128$ & $0.70$ & $0.58$ & $0.097$ \\
 & $256$ & $0.76$ & $0.63$ & $0.085$ \\
\midrule
\multirow{4}{*}{PCA cap $d'$}
 & $10$  & $0.41$ & $0.38$ & $0.190$ \\
 & $25$  & $0.54$ & $0.44$ & $0.122$ \\
 & $50$  & $0.70$ & $0.58$ & $0.097$ \\
 & $100$ & $0.81$ & $0.69$ & $0.085$ \\
\bottomrule
\end{tabular}
\caption{Encoding ablation on CLINC150 ($L=2$, chance $=0.0067$, $6$ seeds;
standard deviations are below $0.035$ throughout). The reference row
(TF--IDF, $N=128$, $d'=50$) differs marginally from Table~\ref{tab:clinc}
because the ablation restricts to intents with at least five utterances. The
text-blind encoder preserves memorisation and destroys generalisation, which is
the sharpest form of the claim that generalisation is a property of the
representation and not of the storage rule.}
\label{tab:clinc-ablation}
\end{table}

\section{Reproducibility: figure parameters}
\label{app:repro}

Every simulated point in the figures is a mean $\pm$ one standard deviation over
$N_{\mathrm{seeds}}$ independent re-draws of the disorder (patterns, feature
maps, train/test splits); each per-seed value is itself an average over
$N_{\mathrm{eval}}$ independent retrieval trials. The dynamics is the
zero-temperature parallel Glauber update of Algorithm~\ref{algo:hetero}, run for a
small fixed number of sweeps ($n_{\mathrm{steps}}=5$). The theory panels
(Figures~\ref{fig:storage}(a,b), \ref{fig:basins}(a,b), \ref{fig:doubling},
\ref{fig:comparison}) are {analytic}: the saddle $x^{\ast}$ and the tilted
saddle $m_s(r)$ are solved numerically from the closed forms of
\ref{app:saddle} and~\ref{app:corrupted}, with no simulation. 
Tables~\ref{tab:repro-hmm} and~\ref{tab:repro-vdjdb}
list the control parameters used for each panel.

\begin{table}[t!]
\centering
\small
\renewcommand{\arraystretch}{1.25}
\begin{tabular}{@{}p{2.3cm} p{3.5cm} p{8.4cm}@{}}
\toprule
Figure (panel) & Protocol & Parameters \\
\midrule
\ref{fig:storage}(a)   & one-step stability, $r=1$ (i.i.d.) &
  $N\in\{8,9,10,11,12\}$, $L=2$, $P$ on a log grid to $\sim\!2\!\times\!10^{7}$ (22 pts), $N_{\mathrm{seeds}}=6$, $N_{\mathrm{eval}}=200$  \\
\ref{fig:basins}(a)    & corrupted-cue recovery &
  $N=10$, $L=2$, $\alpha_D=0.5$ ($d=5$), $P=3093$ ($\gamma=0.1$), $r\in[0,1]$ (11 pts) \\
\ref{fig:hmm-manifold}(a,b) & manifold geometry  &
  $N=10$, $L=2$, $P=800$, $2\!\times\!10^{4}$ pattern pairs, $\alpha_D\in[0.1,1]$ \\
\ref{fig:hmm-manifold}(c)   & dimension scan &
  $N=10$, $L=2$, $P=3000$, $n_{\mathrm{bits}}=2$, $\alpha_D\in[0.1,1]$ \\
\ref{fig:hmm-capacity}(a)   & capacity scan &
  $N=10$, $L=2$, $\alpha_D\in\{0.3,0.6,1.0\}$ + i.i.d., $P\in[4,5\!\times\!10^{4}]$ \\
\ref{fig:hmm-capacity}(b)   & width scan  &
  $N=10$, $L\in\{2,3,4\}$, $\alpha_D=0.5$, $P\in[4,2\!\times\!10^{4}]$ \\
\ref{fig:hmm-capacity}(c)   & surjective compression  &
  $N=12$, $L=2$, $P=4000$, $K=2^{1..6}$, $r\in\{0.6,\dots,1.0\}$ \\
\ref{fig:hmm-generalization}(a,b) & generalisation  &
  $N=12$, $L=2$, $d=6$, $K=8$ ($6$ seen, $2$ held out), $P\in[16,2\!\times\!10^{4}]$, $n_{\mathrm{test}}=300$ \\
\ref{fig:hmm-generalization}(c)   & theory vs empirical  &
  $N=10$, $L=2$, $\gamma=1$ ($P=30920$), $r=0.68$, $\alpha_D\in[0.1,1]$ \\
\bottomrule
\end{tabular}
\caption{Hidden Manifold Model figures. Rates used as theory reference:
$\rho_2=1.347$, $\rho_3=2.078$, $\rho_4=2.773$; $L=2$ thresholds
$r_c^{\mathrm{ann}}=0.3195$, $r_c^{\mathrm{typ}}=0.5714$. All runs use
$N_{\mathrm{seeds}}=6$, $N_{\mathrm{eval}}=200$, $n_{\mathrm{steps}}=5$.}
\label{tab:repro-hmm}
\end{table}

\begin{table}[t!]
\centering
\small
\renewcommand{\arraystretch}{1.25}
\begin{tabular}{@{}p{2.3cm} p{3.5cm} p{8.4cm}@{}}
\toprule
Figure (panel) & Content & Parameters \\
\midrule
\ref{fig:vdjdb-data}(a)    & cleaning funnel &
  \texttt{vdjdb-2025-09-25}: $137{,}484\!\to\!1{,}052$ triples, $220$ epitopes \\
\ref{fig:vdjdb-data}(b)    & epitope cluster sizes &
  max $146$ receptors/epitope, $120$ singletons \\
\ref{fig:vdjdb-data}(c,d)  & encoding quality / SimHash &
  Atchley + PCA($50$) + SimHash, $N=128$; $4000$ sampled pairs (d) \\
\ref{fig:vdjdb-results}(a) & biological tasks  &
  $N=128$, $L=3$, $P=800$, $N_{\mathrm{seeds}}=6$ \\
\ref{fig:vdjdb-results}(b) & basins  &
  $N=128$, $L=3$, $P=600$, $r\in[0,0.82]$ \\
\ref{fig:vdjdb-results}(c) & generalisation vs null  &
  $N=128$, $P=800$, $r=0.68$; $6$ true vs $20$ label-permuted splits \\
\ref{fig:vdjdb-results}(d) & per-epitope / top-$k$  &
  $N=128$, test fraction $0.25$, $\min$ cluster size $\ge2$ \\
\ref{fig:vdjdb-ablation}   & encoding ablation &
  feature $\in\{$Atchley, $k$-mer, random$\}$; $N\in\{32,64,128,256\}$ \\
\bottomrule
\end{tabular}
\caption{VDJdb figures. Layers are $\alpha$-CDR3, $\beta$-CDR3 and epitope
($L=3$); the encoder is Atchley positional features reduced by PCA and binarised
by a SimHash (\ref{app:vdjdb_protocol}). All runs use
$N_{\mathrm{seeds}}=6$, $N_{\mathrm{eval}}=200$.}
\label{tab:repro-vdjdb}
\end{table}

\bibliographystyle{elsarticle-num}
\bibliography{references}

\end{document}